\documentclass[11pt]{article}%
\usepackage{amsmath}
\usepackage{fullpage}
\usepackage{amssymb}
\usepackage{graphicx}
\usepackage[font={footnotesize},labelfont={bf}]{caption}
\usepackage{hyperref}
\usepackage{verbatim}
\usepackage{amsfonts}%
\providecommand{\U}[1]{\protect\rule{.1in}{.1in}}
\newtheorem{theorem}{Theorem}

\newtheorem{definition}[theorem]{Definition}

\newtheorem{lemma}[theorem]{Lemma}

\newenvironment{proof}[1][Proof]{\noindent\textbf{#1.} }{\ \rule{0.5em}{0.5em}}

\def\>{\rangle}
\def\<{\langle}

\let\originalleft\left
\let\originalright\right
\def\left#1{\mathopen{}\originalleft#1}
\def\right#1{\originalright#1\mathclose{}}

\begin{document}

\title{\vspace{-2.5cm}\textbf{The information-theoretic costs of simulating quantum measurements}}
\author{Mark M. Wilde and Patrick Hayden\\\textit{School of Computer Science, McGill University,}\\\textit{Montreal, Quebec, Canada H3A 2A7}\\
\\Francesco Buscemi\\\textit{Institute for Advanced Research, Nagoya University,}\\\textit{Chikusa-ku, Nagoya 464-8601, Japan}\\
\\Min-Hsiu Hsieh\\\textit{Centre for Quantum Computation and Intelligent Systems (QCIS),}\\\textit{Faculty of Engineering and Information Technology,}\\\textit{University of Technology, Sydney, NSW 2007, Australia}}
\maketitle

\begin{abstract}
Winter's measurement compression theorem stands as 
one of the most penetrating insights of quantum information theory. In addition
to making an original and profound statement about measurement in quantum
theory, it also underlies several other general protocols used for
entanglement distillation and local purity distillation. The theorem provides
for an asymptotic decomposition of any quantum measurement into {\it noise} and
{\it information}. This decomposition leads to
an optimal protocol for having a sender simulate many independent instances of a quantum
measurement and send the
measurement outcomes to a receiver, using as little communication as
possible. The protocol assumes that
the parties have access to some amount of common randomness, which is a
strictly weaker resource than classical communication.

In this paper, we provide a full review of Winter's measurement compression
theorem, detailing the information processing task, giving examples for
understanding it, reviewing Winter's achievability proof, and detailing a new
approach to its single-letter converse theorem. We prove an extension of the
theorem to the case in which the sender is not required to receive the
outcomes of the simulated measurement. The total cost of common randomness and
classical communication can be lower for such a \textquotedblleft
non-feedback\textquotedblright\ simulation, and we prove a single-letter
converse theorem demonstrating optimality. We then review the Devetak-Winter
theorem on classical data compression with quantum side information, providing
new proofs of its achievability and converse parts. From there, we outline a
new protocol that we call \textquotedblleft measurement compression with
quantum side information,\textquotedblright\ announced previously by two of us
in our work on triple trade-offs in quantum Shannon theory. This protocol has
several applications, including its part in the \textquotedblleft
classically-assisted state redistribution\textquotedblright\ protocol, which
is the most general protocol on the static side of the quantum information
theory tree, and its role in reducing the classical communication cost in a
task known as local purity distillation. We also outline a connection between
measurement compression with quantum side information and recent work on
entropic uncertainty relations in the presence of quantum memory. Finally, we
prove a single-letter theorem characterizing measurement compression with
quantum side information when the sender is not required to obtain the
measurement outcome.

\end{abstract}

\section{Introduction}

Measurement plays an important role in quantum theory. It is the interface between
the macroscopic world of everyday experience and the quantum world, which is characterized by noncommutativity and superposition. The translation is imperfect, however, with superposition and noncommutativity leading necessarily to uncertainty in the outcomes
of measurements. In any given measurement, there will be noise inherent to the measurement procedure, uncertainty due to the state being measured and, most importantly, \emph{information}.   The objective of this article is to explain how to separate out these components, precisely identifying and quantifying them in the data produced by a quantum measurement. 
To do so, it will be crucial to adopt an information-theoretic point of view, not just to
provide the necessary techniques to solve the problem, but even to figure out how to 
properly formulate the question.

If we are only concerned with capturing the statistics of the
outcomes of a quantum measurement, the most general mathematical description
is to use the positive operator-valued measure (POVM) formalism
\cite{D76,H82,K83}. In the POVM\ formalism, a quantum measurement is specified
as a set $\Lambda\equiv\left\{  \Lambda_{x}\right\}  $\ of operators indexed
by some classical label $x$ corresponding to the classical outcomes of the
measurement. These operators should be positive and form a resolution of the
identity on the Hilbert space of the system that is being measured:%
\[
\forall x:\Lambda_{x}\geq0,\ \ \ \ \ \ \ \ \sum_{x}\Lambda_{x}=I.
\]
Given a quantum state described by a density operator $\rho$ (a positive, unit
trace operator) and a POVM $\Lambda$, a measurement of $\rho$ specified by $\Lambda$ induces a random variable $X$,
and the probability $p_{X}\left(  x\right)  $ for the classical
outcome $x$ to occur is given by the Born rule:%
\begin{equation}
p_{X}\left(  x\right)  =\text{Tr}\left\{  \Lambda_{x}\rho\right\}  .
\label{eq:first-POVM-dist}%
\end{equation}
Positivity of the operators $\Lambda_{x}$ and $\rho$ guarantees positivity of
the distribution $p_{X}\left(  x\right)  $, and that the set $\Lambda$ forms a
resolution of the identity and the density operator $\rho$ has unit trace
guarantees normalization of the distribution $p_{X}\left(  x\right)  $.

The above definition of a POVM\ makes it clear that the set of all POVMs is a
convex set, i.e., given a POVM\ $\Lambda\equiv\left\{  \Lambda_{x}\right\}  $
and another $\Gamma\equiv\left\{  \Gamma_{x}\right\}  $, with $0 <
\lambda < 1$, the convex combination $\lambda\Lambda+\left(  1-\lambda
\right)  \Gamma\equiv\left\{  \lambda\Lambda_{x}+\left(  1-\lambda\right)
\Gamma_{x}\right\}  $ is also a POVM. 
The physical interpretation of this
convexity is that it might be possible to decompose any particular measuring
apparatus into noise and information. If an
apparatus does not admit a decomposition of this form, then it is an extremal
POVM, lying on the boundary of the convex set. If it does, however, as in the above example apparatus,
one could first flip a biased coin with distribution $\left(  \lambda,1-\lambda\right)  $
to determine whether to perform $\Lambda$ or $\Gamma$ and then perform the
corresponding measurement. The coin flip is a source of noise
because it is independent of the physical measurement outcome, and the
distribution for the outcome corresponds to the information. Decomposing an apparatus in this
way is a useful idea with many applications. 

To develop a robust quantitative theory, however, it is surprisingly effective to consider the above ideas from an information-theoretic
standpoint, in the sense of Shannon \cite{bell1948shannon}. In this context, that approach
will have three main features: a tolerance for small imperfections, a focus on asymptotics, and an emphasis on communication. To begin with, let us focus on the first two. From an operational point of view, there is little justification for requiring an exact convex decomposition of a given measurement. As long as any imperfections are very small, approximation by a convex decomposition leads to experimentally indistinguishable consequences. Moreover, measurement statistics are most meaningful in a setting in which the measurement is repeated many times on identical state preparations. As such, it is sensible, and remarkably powerful, to ask about approximate convex decomposition of repeated measurements, with the permissible imperfection required to vanish in the limit of infinite repetitions. 

The relevance of communication is less immediate. In the example described above, the measurement $\lambda \Lambda + (1-\lambda)\Gamma$ could be implemented by first flipping a coin and then either measuring $\Lambda$ or $\Gamma$. This opens up the possibility of significantly compressing the measurement outcomes because there will generically be less uncertainty about the outcome of either $\Lambda$ or $\Gamma$ alone than the convex combination $\lambda \Lambda + (1-\lambda)\Gamma$.
To formalize this notion, one could imagine that two parties, traditionally named Alice and Bob, are trying to collectively implement a measurement. They share some common random bits that can be used to perform the $(\lambda,1-\lambda)$ coin flip without communicating, and Alice holds the quantum system on which $\lambda \Lambda + (1-\lambda)\Gamma$ is to be measured. Based on the result of the coin flip, Alice would apply either $\Lambda$ or $\Gamma$ and compress the outcome as much as possible, minimizing the number of bits she needs to send to Bob in order to allow him to reconstruct the outcome of the measurement. Optimizing the number of bits required over all possible measurement simulation strategies, of which we have only described one, then provides a robust operational measure of the amount of information generated by the quantum measurement. 

In a seminal paper, Winter successfully performed this information-theoretic analysis of
measurement, and in so doing, was able to make a profound and original statement 
about the nature of information in quantum measurement~\cite{W01}. The
content of his \textquotedblleft measurement compression
theorem\textquotedblright\ is the specification of an optimal two-dimensional
rate region, characterizing the resources needed for an asymptotically
faithful simulation of a quantum measurement $\Lambda$ on a state $\rho$\ in
terms of common randomness and classical communication. The sender (Alice) and
receiver (Bob) both obtain the outcome of the measurement, and as such, this is
known as a \textquotedblleft feedback simulation\textquotedblright
 (terminology introduced in a different though related context \cite{BDHSW09}). 
His measurement compression protocol achieves one important optimal rate
pair in this region: if, to first order, at least $nH\left(  X|R\right)  $ bits of common
randomness are available, then it is possible to simulate the measurement
$\Lambda^{\otimes n}$ on the state $\rho^{\otimes n}$ with only about $nI\left(  X;R\right)  $ bits of
classical communication. We allow $n$, the number of repetitions of the measurement $\Lambda$, to go to infinity, in which limit the
simulation becomes asymptotically faithful. The entropies $H\left(
X|R\right)  $ and $I\left(  X;R\right)  $ are defined as%
\begin{align*}
H\left(  X|R\right)   &  \equiv H\left(  XR\right)  -H\left(  R\right)  ,\\
I\left(  X;R\right)   &  \equiv H\left(  X\right)  -H\left(  X|R\right)  ,
\end{align*}
with the von Neumann entropy of a state $\sigma$ defined as $H\left(
\sigma\right)  \equiv-$Tr$\left\{  \sigma\log_{2}\sigma\right\}  $. The above
entropies are taken with respect to the state%
\begin{equation}  \label{eqn:post.meas.state}
\sum_{x}\left\vert x\right\rangle \left\langle x\right\vert ^{X}%
\otimes\text{Tr}_{A}\left\{  \left(  I^{R}\otimes\Lambda_{x}^{A}\right)
\phi_{\rho}^{RA}\right\}  ,
\end{equation}
where $\phi_{\rho}^{RA}$ is any purification of the state $\rho$, meaning that
$\phi_{\rho}^{RA}$ is a rank-one density operator satisfying Tr$_{R}\left\{
\phi_{\rho}^{RA}\right\}  =\rho$.\footnote{Here and throughout this paper,
we use superscripts such as $A$, $B$, $R$, and $E$ to denote quantum systems
with corresponding Hilbert spaces $\mathcal{H}_A$,
$\mathcal{H}_B$, $\mathcal{H}_R$, and $\mathcal{H}_E$.
Such a labeling is useful in quantum information theory because we often
deal with states that are defined over many systems. We also use the shorthand
$\phi \equiv \vert \phi \rangle \langle \phi \vert $ to denote
a pure-state density operator. So, for example,
the state $\phi_{\rho}^{RA}$ is shared between systems $A$ and $R$, implying that
$\phi_{\rho}^{RA}$ is an operator acting on the tensor-product Hilbert
space $\mathcal{H}_R \otimes \mathcal{H}_A$. We also freely identify Roman capital letters
$W$, $X$, $Y$, and $Z$ with both random variables and
quantum systems containing only classical data (as in (\ref{eqn:post.meas.state})). There should be no confusion
here because these entities are in direct correspondence.} One can think of (\ref{eqn:post.meas.state}) as the
post-measurement state, including both the classical outcome of the
measurement and the subsequent state of the reference system $R$. The other
important rate pair corresponds to Shannon's protocol. If no common randomness
is available and both the sender and receiver are to obtain the measurement
outcome, then the lowest achievable rate of classical communication is
$H\left(  X\right)  $, the Shannon entropy of the distribution of measurement outcomes in
(\ref{eq:first-POVM-dist}). Time-sharing between these two protocols,
converting classical communication to common randomness, and wasting common
randomness then give all other optimal rate pairs. (See
Figure~\ref{fig:IC-rate-region} for an example plot of the region.)

Winter's measurement compression protocol has an important place in the
constellation of quantum Shannon theoretic
protocols. It evolved from earlier work in Refs.~\cite{MP00,WM01}, and it is
the predecessor to the quantum reverse Shannon theorem, which was conjectured
in Ref.~\cite{BSST01} and proved later in Refs.~\cite{BDHSW09,BCR11}%
.\footnote{We should clarify here that, while
Ref.~\cite{BDHSW09} appeared on the arXiv in 2009, that
article contains ideas developed and publicized by the
authors over a nine year period starting with the publication of Ref.~\cite{BSST01} in 2001. 
Ref.~\cite{BCR11} features a different proof from that in Ref.~\cite{BDHSW09},
but it exploits many of the important ingredients developed in
Ref.~\cite{BDHSW09}.} The quantum reverse Shannon theorem quantifies
the noiseless resources required to simulate a noisy quantum channel.
Since Winter's measurement compression theorem applies to a quantum
measurement and a quantum measurement is a special type of quantum channel
with quantum input and classical output, it is clear that the measurement
compression protocol gives a special type of quantum reverse Shannon theorem.
The quantum reverse Shannon theorem may seem on first encounter to correspond to a pointless
task. After all, in the words of Ref.~\cite{BSST01}, why would we want to dilute fresh water into salt water? First appearances notwithstanding, it has at least two nearly immediate and significant information-theoretic
applications: in proving strong converses
\cite{BSST01,W02,BDHSW09,BCR11,BBCW11}\ and in lossy data compression, otherwise known as rate
distortion theory \cite{W02,LD09,L09,DHW11}. The connection to strong
converses follows from a \textit{reductio ad absurdum} argument:\ if one were
able to simulate a channel at a rate larger than its capacity, then it would
be possible to bootstrap a channel code and a simulation code to achieve more
communication than a noiseless channel would allow for. With the aid of an
appropriate reverse Shannon theorem, one can then argue that coding at a rate
beyond the capacity should make the error probability converge to one
exponentially fast in the number of channel uses. The connection to rate
distortion theory \cite{B71}\ follows from the observation that a reverse
Shannon theorem achieves a task strictly stronger than the usual average
distortion criterion considered in rate distortion theory. There, one requires
that an information source be represented by the receiver up to some average
distortion $D\geq0$. If one were to simulate a channel on the information
source that does not distort it by more than $D$ on average, then clearly such
a protocol would already satisfy the demands of rate distortion.

There are two other useful applications of Winter's measurement compression
theorem. The first is in local purity distillation
\cite{HHHHOSS03,HHHOSSS05,D05,KD07}, where the task is for two spatially
separated parties to distill local pure states from an arbitrary bipartite
mixed state $\rho^{AB}$ by using only local unitary operations and classical
communication. The measurement compression theorem is helpful in determining
the classical communication cost of such protocols, as considered in
Ref.~\cite{KD07}. Another application of measurement compression is in
realizing the first step of the so-called \textquotedblleft grandmother
protocol\textquotedblright\ of quantum information theory \cite{DHW08}, where
the objective is to distill entanglement from a noisy bipartite state
$\rho^{AB}$ with the help of noiseless classical and quantum
communication. It is possible to improve upon both of these protocols by
exploiting one of the new measurement compression theorems that we outline in
this paper.

Once one takes the first step of splitting the implementation of a measurement between Alice and Bob, it becomes natural to consider different notions of simulation. What if only Bob needs to get the outcome of the measurement, not Alice? What if Bob holds a quantum system entangled with the system being measured? These and related variations provide a very precise and diverse set of tools for analyzing the dichotomy between noise and information in quantum measurements. Beyond providing a detailed review of Winter's theorem, the main contribution of this article will be to develop these variations and generalizations of his original theorem. More specifically, our contributions are as follows: 

\begin{itemize}
\item We provide a full review of Winter's measurement compression theorem,
detailing the basic information processing task, the statement of the theorem,
Winter's achievability proof, and a simple converse theorem that demonstrates
an optimal characterization of the rate region. We also review Winter's
extension of the theorem to quantum instruments.

\item We extend Winter's measurement compression theorem to the setting in
which the sender is not required to receive the outcome of the measurement
simulation. Such a task is known as a \textquotedblleft
non-feedback\textquotedblright\ simulation, in analogy with a similar setting
in the quantum reverse Shannon theorem \cite{BDHSW09}. A benefit of a
\textquotedblleft non-feedback\textquotedblright\ simulation is that the total
cost of common randomness and classical communication can be lower than that
of a \textquotedblleft feedback\textquotedblright\ simulation, leading to
interesting, non-trivial trade-off curves for the rates of these resources.
Also, we prove a single-letter converse theorem for this case, demonstrating
that our protocol is optimal.

\item We then review Devetak and Winter's theorem regarding classical data
compression with quantum side information (CDC-QSI) \cite{DW02}. The setting
of the problem is that an information source distributes a random classical
sequence to one party and a quantum state correlated with the sequence to another party. The
objective is for the first party to transmit the classical sequence to the
second party using as few noiseless classical bit channels as possible. As
such, it is one particular quantum generalization of the classic Slepian-Wolf
problem \cite{SW73}. In the Slepian-Wolf protocol, the first party hashes the
sequence received from the source, transmits the hash, and the second party
uses his side information to search among all the sequences for any that are
consistent with the hash and are a \textquotedblleft reasonable
cause\textquotedblright\ for his side information. We provide a novel
achievability proof for CDC-QSI that is a direct quantization of this
strategy, replacing the latter search with binary-outcome quantum
measurements. We also provide a simple converse proof that is along the lines
of the standard converses in Refs.~\cite{CT91,el2010lecture}.

\item The above reviews of measurement compression and CDC-QSI\ then prepare
us for another novel contribution: measurement compression in the presence of
quantum side information (MC-QSI). The setting for this new protocol is that a
sender and receiver share many copies of some bipartite state $\rho^{AB}$, and
the sender would like to simulate the action of many independent and identical
measurements on the $A$ system according to some POVM\ $\Lambda$. The protocol
is a \textquotedblleft feedback simulation,\textquotedblright\ such that the
sender also obtains the outcomes of the measurement (though we still refer to
it as MC-QSI for short). The MC-QSI\ protocol combines ideas from the
measurement compression theorem and CDC-QSI in order to reduce the classical
communication rate and common randomness needed to simulate the measurement.
The idea is that Alice performs the measurement compression protocol as she
would before, but she hashes the output of the simulated measurement and sends
this along to Bob. Bob then searches among all the post-measurement states
that are consistent with the hash and his share of the common randomness,
similar to the way that he would in the CDC-QSI protocol. The result is a
reduction in the classical communication and common randomness rate to
$I\left(  X;R|B\right)  $ and $H\left(  X|RB\right)  $, respectively, where
the entropies are with respect to the following state:%
\[
\sum_{x}\left\vert x\right\rangle \left\langle x\right\vert ^{X}%
\otimes\text{Tr}_{A}\left\{  \left(  I^{RB}\otimes\Lambda_{x}^{A}\right)
\phi_{\rho}^{RBA}\right\}  ,
\]
and $\phi_{\rho}^{RBA}$ is a purification of the state $\rho^{AB}$. These
rates are what we would intuitively expect of such a protocol---they are the
same as in Winter's original theorem, except the entropic quantities are
conditioned on Bob's quantum side information in the system $B$.

\item After developing MC-QSI, we briefly\ discuss three of its applications.
The first is an application that two of us announced in Ref.~\cite{HW10}%
:\ MC-QSI along with state redistribution \cite{DY08,YD09} acts as a
replacement for the \textquotedblleft grandmother\textquotedblright\ protocol
discussed above. The resulting protocol uses less classical and quantum
communication and can in fact generate the grandmother by combining it with
entanglement distribution. As such, MC-QSI\ and state redistribution form the
backbone of the best known \textquotedblleft static\textquotedblright%
\ protocols in quantum Shannon theory (though, one should be aware that these
results are only optimal up to a regularization, so it could very well be that further
improvements are possible). The second application is an observation that the
above protocol leads to a quantum reverse Shannon theorem for a quantum
instrument, that is, a way to simulate the action of a quantum instrument on a
quantum state by employing common randomness, classical communication,
entanglement, and quantum communication. The third application is an
improvement of the local purity distillation protocol from Ref.~\cite{KD07},
so that we can lower the classical communication cost from $I\left(
Y;BE\right)  $ to $I\left(  Y;E|B\right)  $, as one should expect when taking
quantum side information into account.

\item We then discuss a way that we can relate recent work on entropic
uncertainty relations with quantum side information
\cite{RB09,BCCRR10,TR11,CCYZ11,FL12}\ to provide a lower bound on the
classical resources required in two different complementary MC-QSI\ protocols.

\item Finally, we analyze the MC-QSI problem in the case where the
sender is not required to receive the outcomes of the measurement simulation.
For this non-feedback MC-QSI problem, we once again develop optimal protocols and find a single-letter characterization of the achievable rate region. While the necessary protocols are simply the natural combinations of those used in MC-QSI with those used for non-feedback MC, the optimality proof is different and remarkably subtle.
\end{itemize}

All the simulation theorems appearing in this paper are
\textquotedblleft single-letter,\textquotedblright\ meaning that we can
calculate the optimal rate regions as simple entropic functions of one copy of
the state or resource. This type of result occurs more often in quantum
information theory when the resources considered are of a hybrid
classical-quantum nature, as is our case here. The single-letter results here
mean that we can claim to have a complete information-theoretic understanding
of the tasks of MC, non-feedback MC, CDC-QSI, MC-QSI, and non-feedback MC-QSI.

\section{Measurement compression}

\label{sec:meas-comp}This section provides a detailed review of the main
results in Winter's original paper on measurement compression \cite{W01}, and it also
serves to establish notation used in the rest of the paper.
Consider a quantum state $\rho$ and a POVM $\Lambda\equiv\left\{  \Lambda
_{x}\right\}_{x \in \mathcal{X}}  $, such that $\Lambda_{x}\geq0$ and $\sum_{x}\Lambda_{x}=I$.
Measuring the POVM\ $\Lambda$ on the state $\rho$ induces
a random variable $X$ with the following
distribution $p_{X}\left(  x\right)  $:%
\[
p_{X}\left(  x\right)  \equiv\text{Tr}\left\{ \Lambda_{x} \rho\right\}  .
\]
Suppose that a quantum information source outputs many copies of the state
$\rho$\ and the POVM\ is performed many times, producing the IID\ distribution
$p_{X^{n}}\left(  x^{n}\right)  $ (where $x^n=x_1 x_2 \cdots x_n$):%
\begin{align*}
p_{X^{n}}\left(  x^{n}\right)   &  \equiv\text{Tr}\left\{ \Lambda_{x^{n}} \rho^{\otimes
n}\right\} \\
&  =\text{Tr}\left\{ \left(  \Lambda_{x_{1}}\otimes\Lambda_{x_{2}}\otimes\cdots\otimes
\Lambda_{x_{n}}\right) \left(  \rho\otimes\rho\otimes\cdots\otimes\rho\right)
  \right\} \\
&  =\prod\limits_{i=1}^{n}\text{Tr}\left\{  \Lambda_{x_{i}} \rho\right\}  .
\end{align*}

In order to communicate the result of the measurement to a receiver using a
noiseless classical channel, one could compress the data sequence $x^{n}$
using Shannon compression \cite{CT91}\ and communicate the sequence $x^{n}$
faithfully by transmitting only $nH\left(  X\right)  $ bits. Such a strategy
is optimal if no other resource is shared between the sender and receiver. But
supposing that the sender and receiver have access to some shared randomness
(a fairly innocuous resource), would it be possible for the sender to simulate
the outcome of the measurement using some of the shared randomness and then
communicate fewer classical bits to the receiver in order for him to
reconstruct the sequence $x^{n}$?

The goal of Winter's POVM\ compression protocol \cite{W01}\ is to do exactly
that: accurately simulate the distribution produced by the POVM, by
exploiting shared randomness. The starting point for Winter's protocol is the
observation that any POVM\ $\Lambda$ can be decomposed as a convex combination
of some other POVMs $\left\{  \Gamma^{\left(  m\right)  }\right\}
=\{\{\Gamma_{x}^{\left(  m\right)  }\}\}$, such that%
\begin{equation}
\Lambda_{x}=\sum_{m}p_{M}\left(  m\right)  \Gamma_{x}^{\left(  m\right)  }.
\label{eq:POVM-decomposition}%
\end{equation}
This is due to the fact that the set of all POVMs is a convex
set.\footnote{See Ref.~\cite{DPP05}\ for an explicit algorithm that decomposes
any non-extremal POVM\ in this way.} The set of POVMs $\left\{  \Gamma
^{\left(  m\right)  }\right\}  $ then provides a simulation of the original
POVM\ $  \Lambda  $ by the following procedure:

\begin{enumerate}
\item Generate the variable $M$ according to the distribution $p_{M}\left(
m\right)  $.

\item Measure the state $\rho$ with the POVM\ $\Gamma^{\left(  M\right)  }$.
\end{enumerate}

The resulting distribution for the random variable $X$, when marginalizing over the
random variable $M$, is then as follows:%
\[
\sum_{m}p_{M}\left(  m\right)  \text{Tr}\left\{  \Gamma_{x}^{\left(
m\right)  } \rho\right\}  =\text{Tr}\left\{  \sum_{m}p_{M}\left(  m\right)
\Gamma_{x}^{\left(  m\right)  }\rho\right\}  =\text{Tr}\left\{  \Lambda
_{x}\rho\right\}  =p_{X}\left(  x\right)  .
\]
Thus, the random variable $M$ is a source of noise for simulating
the POVM $\Lambda$, and the output $X$ represents information.
Separating these two components is a useful idea, and it is
what allows us to simulate a POVM\ by a protocol similar to the above.

\subsection{Information processing task for measurement compression}

We can now define the information processing task for a measurement
compression protocol. Given the original POVM $\Lambda\equiv\left\{
\Lambda_{x}\right\}  $, suppose that it acts on an $n$-fold tensor product
state $\rho^{\otimes n}$. The POVM\ then has the form%
\[
\Lambda^{\otimes n}\equiv\left\{  \Lambda_{x^{n}}\right\}  _{x^{n}%
\in\mathcal{X}^{n}},
\]
where%
\[
\Lambda_{x^{n}}\equiv\Lambda_{x_{1}}\otimes\Lambda_{x_{2}}\otimes\cdots
\otimes\Lambda_{x_{n}}.
\]
The ideal measurement compression protocol would be for the sender Alice to
simply perform this measurement on each copy of her state and transmit the classical
output to the receiver Bob. Figure~\ref{fig:ideal-IC}\ depicts this ideal protocol.

\begin{figure}[ptb]
\begin{center}
\includegraphics[
width=3.0381in
]{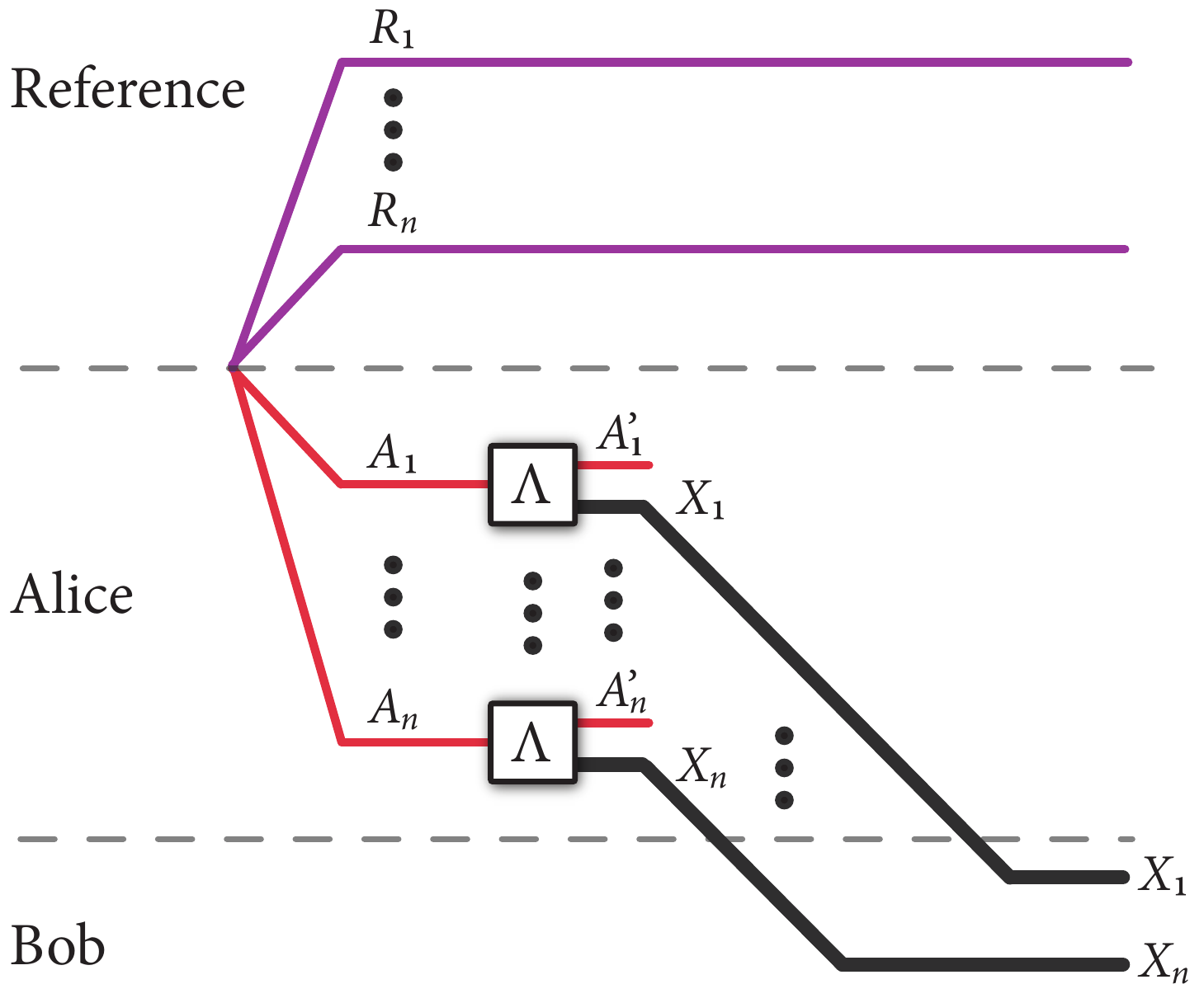}
\end{center}
\caption{\textbf{Ideal measurement compression.} In an ideal protocol for
measurement compression, Alice performs the POVM\ $\Lambda\equiv\left\{
\Lambda_{x}\right\}  $ on $n$ copies of the state $\rho$, which for the
$i^{\text{th}}$ state leads to a quantum system $A_{i}^{\prime}$ and a
classical output $X_{i}$. The goal of the protocol is to transmit the
classical output $X^{n}$ to a receiver. Doing so perfectly would require
$n\log\left\vert \mathcal{X}\right\vert $ bits of communication, where
$\mathcal{X}$ is the alphabet for the random variable $X$.~Winter's
measurement compression protocol gives a way of doing so by allowing for a
small error but demanding that this error vanish in the asymptotic limit of
many copies of the state $\rho$. The idea is to simulate the measurement in
such a way that a third party would not be able to distinguish between the
true measurement and the simulated one. An assumption of this protocol is that
the sender obtains the outcome of the simulated measurement in addition to the
receiver.}%
\label{fig:ideal-IC}%
\end{figure}

Our goal is to find an approximate convex decomposition of the tensor-product
POVM\ of the sort in (\ref{eq:POVM-decomposition}), but in this case it should
have the form:%
\[
\Lambda_{x_{1}}\otimes\Lambda_{x_{2}}\otimes\cdots\otimes\Lambda_{x_{n}%
}\approx\widetilde{\Lambda}_{x^{n}},
\]
where%
\[
\widetilde{\Lambda}_{x^{n}}\equiv\sum_{m}p_{M}\left(  m\right)  \Gamma_{x^{n}%
}^{\left(  m\right)  },
\]
so that each POVM element$\ \Gamma_{x^{n}}^{\left(  m\right)  }$ is a
collective measurement on the $n$-fold tensor product Hilbert space. One might
expect that such a collective measurement would have some compression
capabilities built into it, in the sense that it could reduce the number of bits needed to
represent the sequence $x^{n}$. Figure~\ref{fig:IC}\ depicts the most general
protocol for measurement compression when both sender and receiver are to
obtain the outcome of the simulated measurement (known as a \textquotedblleft
feedback\textquotedblright\ simulation).

\begin{figure}[ptb]
\begin{center}
\includegraphics[
width=3.039in
]{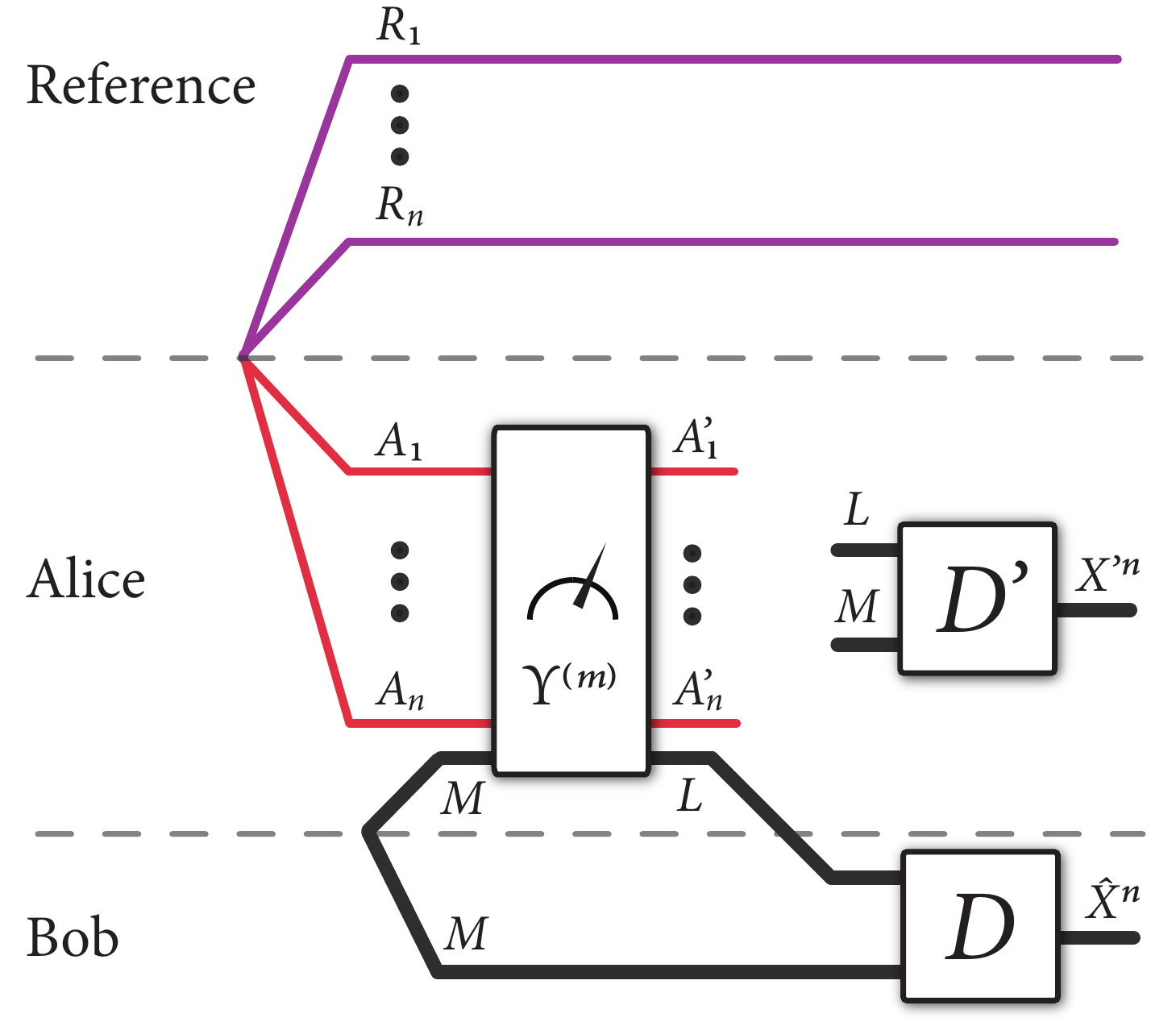}
\end{center}
\caption{\textbf{Measurement compression protocol.} The most general protocol
for \textquotedblleft feedback\textquotedblright\ measurement compression that
exploits common randomness and classical communication. Alice selects a
POVM\ $\Upsilon^{\left(  m\right)  }\equiv\{\Upsilon_{l}^{\left(  m\right)
}\}$ according to the common randomness $M$. She then performs this POVM\ on
many copies of the state $\rho$, and receives an outcome $l$ from it, modeled
by the random variable $L$. She transmits the variable $L$ over $\log
_{2}\left\vert \mathcal{L}\right\vert $ noiseless classical bit channels. Bob
receives this variable, and by combining it with his share of the common
randomness, he can reconstruct the output $\hat{X}^{n}$ of the simulated
measurement. In a feedback simulation, the sender also reconstructs a variable
$X^{\prime n}$, which is the output of the simulated measurement. The goal of
a feedback measurement compression protocol is for the classical outputs of
the simulated measurement to be statistically indistinguishable from the
output of the ideal measurement (this is from the perspective of someone
holding both the reference systems and the classical outputs).}%
\label{fig:IC}%
\end{figure}

We now make precise the above notion of the approximation of a POVM acting on
a source state. Suppose that there is some convex decomposition of the
tensor-product source $\rho^{\otimes n}$ as%
\begin{equation}
\rho^{\otimes n}=\sum_{k}p_{K}\left(  k\right)  \sigma_{k},
\label{eq:source-decomp}%
\end{equation}
where the states $\sigma_{k}$ are generally entangled states living on the
$n$-fold tensor product Hilbert space. Thus, one could view the preparation of
the source as a selection of a random variable $K$ according to $p_{K}\left(
k\right)  $, followed by a preparation of the state $\sigma_{K}$. There is
then a joint distribition $p_{K,X^{n}}\left(  k,x^{n}\right)  $ for the
selection of the source and the true measurement result:%
\begin{equation}
p_{K,X^{n}}\left(  k,x^{n}\right)  \equiv p_{K}\left(  k\right)
\text{Tr}\left\{  \Lambda_{x^{n}}\sigma_{k}\right\}  ,
\label{eq:true-joint-dist}%
\end{equation}
and a joint distribution $p_{K,\widetilde{X^{n}}}\left(  k,x^{n}\right)  $ for
the selection of the source and the approximation measurement's result:%
\begin{align*}
p_{K,\widetilde{X^{n}}}\left(  k,x^{n}\right)   &  \equiv p_{K}\left(
k\right)  \sum_{m}p_{M}\left(  m\right)  \text{Tr}\left\{  \Gamma_{x^{n}%
}^{\left(  m\right)  }\sigma_{k}\right\} \\
&  =p_{K}\left(  k\right)  \text{Tr}\left\{  \widetilde{\Lambda}_{x^{n}}%
\sigma_{k}\right\}  .
\end{align*}

\begin{definition}
[Faithful simulation]\label{def:faith-sim}A sequence of protocols provides a
faithful simulation of the POVM\ $\Lambda$ on the source $\rho$, if for all
decompositions of the source of the form in (\ref{eq:source-decomp}), the
above joint distributions are $\epsilon$-close in variational distance for all
$\epsilon>0$ and sufficiently large $n$:%
\begin{equation}
\sum_{k,x^{n}}\left\vert p_{K,X^{n}}\left(  k,x^{n}\right)  -p_{K,\widetilde
{X^{n}}}\left(  k,x^{n}\right)  \right\vert \leq\epsilon.
\label{eq:sim-approx}%
\end{equation}

\end{definition}

The following lemma states a condition for faithful simulation that implies
the above one, and it is the one that we will strive to meet when constructing
a protocol for measurement compression.

\begin{lemma}
\label{lem:faithful-sim-1}If for all $\epsilon>0$ and sufficiently large $n$,
it holds that%
\begin{equation}
\sum_{x^{n}}\left\Vert \sqrt{\omega}\left(  \Lambda_{x^{n}}-\widetilde
{\Lambda}_{x^{n}}\right)  \sqrt{\omega}\right\Vert _{1}\leq\epsilon,
\label{eq:faithful-sim-cond-1}%
\end{equation}
where $\omega\equiv\rho^{\otimes n}$,\footnote{The trace norm $\Vert A \Vert_1$ of an
operator $A$ is equal to $\Vert A \Vert_1 = \text{Tr} \{ \sqrt{A^\dag A} \}$.
The trace distance $\Vert \rho - \sigma \Vert_1$ is commonly used as a measure of
distinguishability between the states $\rho$ and $\sigma$ because it is equal to
$2(1 - 2p_e)$ where $p_e$ is the probability of error in distinguishing these states
if they are chosen uniformly at random.} then the measurement simulation is
faithful, in the sense that the above inequality implies the following one for
all decompositions of the source of the form in (\ref{eq:source-decomp}):%
\[
\sum_{k,x^{n}}\left\vert p_{K,X^{n}}\left(  k,x^{n}\right)  -p_{K,\widetilde
{X^{n}}}\left(  k,x^{n}\right)  \right\vert \leq\epsilon.
\]

\end{lemma}

\begin{proof}
We rewrite the joint distribution $p_{K,X^{n}}\left(  k,x^{n}\right)  $ in
(\ref{eq:true-joint-dist}) as follows:%
\begin{align*}
p_{K,X^{n}}\left(  k,x^{n}\right)   &  =p_{K}\left(  k\right)  \text{Tr}%
\left\{  \Lambda_{x^{n}}\sigma_{k}\right\} \\
&  =\text{Tr}\left\{  \left(  \sqrt{\omega}\Lambda_{x^{n}}\sqrt{\omega
}\right)  \left(  \omega^{-1/2}\ p_{K}\left(  k\right)  \sigma_{k}%
\ \omega^{-1/2}\right)  \right\} \\
&  =\text{Tr}\left\{  \sqrt{\omega}\Lambda_{x^{n}}\sqrt{\omega}\ S_{k}%
\right\}  ,
\end{align*}
where we define $S_{k}$ as%
\[
S_{k}\equiv\omega^{-1/2}\ p_{K}\left(  k\right)  \sigma_{k}\ \omega^{-1/2}.
\]
Observe that the operators $S_{k}$ are positive and sum to the identity on the
support of $\omega$. Thus, they form a POVM $\left\{
S_{k}\right\}  $. Similarly, we can rewrite the joint distribution
$p_{K,\widetilde{X^{n}}}\left(  k,x^{n}\right)  $ as
\[
p_{K,\widetilde{X^{n}}}\left(  k,x^{n}\right)  =\text{Tr}\left\{  \sqrt
{\omega}\widetilde{\Lambda}_{x^{n}}\sqrt{\omega}\ S_{k}\right\}  .
\]
So we can rewrite and upper bound the simulation approximation condition in
(\ref{eq:sim-approx}) as%
\begin{align*}
\sum_{k,x^{n}}\left\vert p_{K,X^{n}}\left(  k,x^{n}\right)  -p_{K,\widetilde
{X^{n}}}\left(  k,x^{n}\right)  \right\vert  &  =\sum_{k,x^{n}}\left\vert
\text{Tr}\left\{  \sqrt{\omega}\left(  \Lambda_{x^{n}}-\widetilde{\Lambda
}_{x^{n}}\right)  \sqrt{\omega}S_{k}\right\}  \right\vert \\
&  \leq\sum_{x^{n}}\left\Vert \sqrt{\omega}\left(  \Lambda_{x^{n}}%
-\widetilde{\Lambda}_{x^{n}}\right)  \sqrt{\omega}\right\Vert _{1},
\end{align*}
where the inequality follows from the following chain of inequalities that
hold for all Hermitian operators $\tau$:%
\begin{align*}
\sum_{k}\left\vert \text{Tr}\left\{  \tau S_{k}\right\}  \right\vert  &
=\sum_{k}\left\vert \text{Tr}\left\{  \left(  \tau_{+}-\tau_{-}\right)
S_{k}\right\}  \right\vert \\
&  \leq\sum_{k}\left\vert \text{Tr}\left\{  \tau_{+}S_{k}\right\}  \right\vert
+\left\vert \text{Tr}\left\{  \tau_{-}S_{k}\right\}  \right\vert \\
&  =\sum_{k}\text{Tr}\left\{  \tau_{+}S_{k}\right\}  +\text{Tr}\left\{
\tau_{-}S_{k}\right\} \\
&  =\text{Tr}\left\{  \tau_{+}\right\}  +\text{Tr}\left\{  \tau_{-}\right\} \\
&  =\left\Vert \tau\right\Vert _{1}.
\end{align*}
In the above, we exploit the decomposition $\tau=\tau_{+}-\tau_{-}$, where
$\tau_{+}$ is the positive part of $\tau$ and $\tau_{-}$ is the negative part,
and the fact that the operators $S_{k}$ form a POVM.
\end{proof}

We now introduce the quantum-to-classical measurement maps $\mathcal{M}%
_{\Lambda^{\otimes n}}$ and $\mathcal{M}_{\widetilde{\Lambda}^{n}}$, defined
as%
\begin{align}
\mathcal{M}_{\Lambda^{\otimes n}}\left(  \sigma\right)   &  \equiv\sum_{x^{n}%
}\text{Tr}\left\{  \Lambda_{x^{n}}\sigma\right\}  \left\vert x^{n}%
\right\rangle \left\langle x^{n}\right\vert ,\label{eq:m-map1}\\
\mathcal{M}_{\widetilde{\Lambda}^{n}}\left(  \sigma\right)   &  \equiv
\sum_{x^{n}}\text{Tr}\left\{  \widetilde{\Lambda}_{x^{n}}\sigma\right\}
\left\vert x^{n}\right\rangle \left\langle x^{n}\right\vert,
\label{eq:m-map2}%
\end{align}
where $\left\vert x^{n}\right\rangle \left\langle x^{n}\right\vert
\equiv \left\vert x_1\right\rangle \left\langle x_1\right\vert \otimes
\left\vert x_2\right\rangle \left\langle x_2\right\vert \otimes
\cdots \otimes
\left\vert x_n\right\rangle \left\langle x_n\right\vert$ and $\{ \vert x \rangle \}$
is some orthonormal basis.
By introducing a purification $\left\vert \phi_{\rho}\right\rangle $ of the
source $\rho$, we can then formulate another notion of faithful simulation, as
given in the following definition:

\begin{definition}
[Faithful simulation for purification]\label{def:faith-sim-2}A sequence of
protocols provides a faithful simulation of the POVM\ $\Lambda$ on the source
$\rho$, if for a purification $\left\vert \phi_{\rho}\right\rangle $ of the
source, the states on the reference and source systems after applying the
measurement maps in (\ref{eq:m-map1}-\ref{eq:m-map2}) are $\epsilon$-close in
trace distance for all $\epsilon>0$ and sufficiently large $n$:%
\begin{equation}
\left\Vert \left(  \text{\emph{id}}\otimes\mathcal{M}_{\Lambda^{\otimes n}%
}\right)  \left(  \phi_{\rho}^{\otimes n}\right)  -\left(  \text{\emph{id}%
}\otimes\mathcal{M}_{\widetilde{\Lambda}^{n}}\right)  \left(  \phi_{\rho
}^{\otimes n}\right)  \right\Vert _{1}\leq\epsilon.
\end{equation}
In the above, it is implicit that the measurement maps act on the $n$ source systems
and the identity map acts on the $n$ reference systems.
\end{definition}

One might think that the above definition of faithful simulation is stronger
than the condition in (\ref{eq:faithful-sim-cond-1}), but the following lemma
demonstrates that they are equivalent.

\begin{lemma}
[Faithful simulation equivalence]The notions of faithful simulation from
Lemma~\ref{lem:faithful-sim-1}\ and Definition~\ref{def:faith-sim-2}\ are
equivalent, in the sense that%
\begin{equation}
\sum_{x^{n}}\left\Vert \sqrt{\omega}\left(  \Lambda_{x^{n}}-\widetilde
{\Lambda}_{x^{n}}\right)  \sqrt{\omega}\right\Vert _{1}=\left\Vert \left(
\text{\emph{id}}^{\otimes n}\otimes\mathcal{M}_{\Lambda^{\otimes n}}\right)
\left(  \phi_{\rho}^{\otimes n}\right)  -\left(  \text{\emph{id}}^{\otimes
n}\otimes\mathcal{M}_{\widetilde{\Lambda}^{n}}\right)  \left(  \phi_{\rho
}^{\otimes n}\right)  \right\Vert _{1}, \label{eq:meas-comp-faithful-cond}%
\end{equation}
for all states $\omega=\rho^{\otimes n}$, purifications of $\rho^{\otimes n}$,
POVMs $\Lambda^{\otimes n}$ and $\widetilde{\Lambda}^{n}$, and the resulting
measurement maps $\mathcal{M}_{\Lambda^{\otimes n}}$ and $\mathcal{M}%
_{\widetilde{\Lambda}^{n}}$.
\end{lemma}

\begin{proof}
We can prove this result by considering the single-copy case. Consider a state
$\rho$, a purification $\phi_{\rho}$, and measurements $\left\{  \Lambda
_{x}\right\}  $ and $\{\widetilde{\Lambda}_{x}\}$. We choose the purification
$\phi_{\rho}$ to be as follows:%
\[
\sqrt{d}\left(  \sqrt{\rho}^{R}\otimes I^{A}\right)  \left\vert \Phi
\right\rangle ^{RA},
\]
where $\left\vert \Phi\right\rangle ^{RA}$ is the maximally entangled state:%
\[
\left\vert \Phi\right\rangle ^{RA}\equiv\frac{1}{\sqrt{d}}\sum_{x}\left\vert
x\right\rangle ^{R}\left\vert x\right\rangle ^{A},
\]
and $\left\{  \left\vert x\right\rangle \right\}  $ is an orthonormal basis
that diagonalizes $\rho$ (this basis is not related to the
one used in (\ref{eq:m-map1}-\ref{eq:m-map2}). Then the unnormalized state after the measurement on
$A$ is equal to%
\begin{equation}
\left(  I^{R}\otimes\sqrt{\Lambda_{x}}^{A}\right)  \left\vert \phi_{\rho
}\right\rangle \left\langle \phi_{\rho}\right\vert ^{RA}\left(  I^{R}%
\otimes\sqrt{\Lambda_{x}}^{A}\right)  =d\left(  \sqrt{\rho}^{R}\otimes
\sqrt{\Lambda_{x}}^{A}\right)  \left\vert \Phi\right\rangle \left\langle
\Phi\right\vert ^{RA}\left(  \sqrt{\rho}^{R}\otimes\sqrt{\Lambda_{x}}%
^{A}\right)  . \label{eq:proof-dev-1}%
\end{equation}
Given the following \textquotedblleft transpose trick\textquotedblright%
\ identity that holds for a maximally entangled state (and where the transpose
is with respect to the basis chosen for $\left\vert \Phi\right\rangle $)%
\begin{align*}
\left(  I\otimes M\right)  \left\vert \Phi\right\rangle  &  =\left(
M^{T}\otimes I\right)  \left\vert \Phi\right\rangle ,\\
\left\langle \Phi\right\vert \left(  I\otimes M\right)   &  =\left\langle
\Phi\right\vert \left(  M^{\ast}\otimes I\right)  ,
\end{align*}
we then have that (\ref{eq:proof-dev-1}) is equal to%
\[
d\left(  \left(  \sqrt{\rho}\sqrt{\Lambda_{x}}^{T}\right)  ^{R}\otimes
I^{A}\right)  \left\vert \Phi\right\rangle \left\langle \Phi\right\vert
^{RA}\left(  \sqrt{\Lambda_{x}}^{T}\sqrt{\rho}^{R}\otimes I^{A}\right)  ,
\]
where the rightmost equivalence $\sqrt{\Lambda_{x}}^{\ast}=\sqrt{\Lambda_{x}%
}^{T}$ follows because $\Lambda_{x}$ is Hermitian. Tracing over the $A$ system
then leaves the following unnormalized state on the reference system%
\begin{equation}
\sqrt{\rho}\Lambda_{x}^{T}\sqrt{\rho}, \label{eq:state-on-ref}%
\end{equation}
an observation first made in Ref.~\cite{HJW93}.

Now consider that a measurement map id$\ \otimes\mathcal{M}_{\Lambda}$ has the
following action on the purification$~\left\vert \phi_{\rho}\right\rangle $:%
\begin{align*}
\left(  \text{id}\otimes\mathcal{M}_{\Lambda}\right)  \left(  \left\vert
\phi_{\rho}\right\rangle \left\langle \phi_{\rho}\right\vert \right)   &
=\sum_{x}\text{Tr}_{A}\left\{  \left(  \text{id}^{R}\otimes\Lambda_{x}%
^{A}\right)  \left(  \left\vert \phi_{\rho}\right\rangle \left\langle
\phi_{\rho}\right\vert ^{RA}\right)  \right\}  \otimes\left\vert
x\right\rangle \left\langle x\right\vert ^{X}\\
&  =\sum_{x}\left(  \sqrt{\rho}\Lambda_{x}^{T}\sqrt{\rho}\right)  ^{R}%
\otimes\left\vert x\right\rangle \left\langle x\right\vert ^{X},
\end{align*}
where the last line follows from the conclusion in (\ref{eq:state-on-ref}).
Thus, we have that%
\begin{align*}
\left\Vert \left(  \text{id}\otimes\mathcal{M}_{\Lambda}\right)  \left(
\phi_{\rho}\right)  -\left(  \text{id}\otimes\mathcal{M}_{\widetilde{\Lambda}%
}\right)  \left(  \phi_{\rho}\right)  \right\Vert _{1}  &  =\left\Vert
\sum_{x}\left(  \sqrt{\rho}\Lambda_{x}^{T}\sqrt{\rho}\right)  ^{R}%
\otimes\left\vert x\right\rangle \left\langle x\right\vert ^{X}-\sum
_{x}\left(  \sqrt{\rho}\widetilde{\Lambda}_{x}^{T}\sqrt{\rho}\right)
^{R}\otimes\left\vert x\right\rangle \left\langle x\right\vert ^{X}\right\Vert
_{1}\\
&  =\left\Vert \sum_{x}\sqrt{\rho}\left(  \Lambda_{x}^{T}-\widetilde{\Lambda
}_{x}^{T}\right)  \sqrt{\rho}\otimes\left\vert x\right\rangle \left\langle
x\right\vert ^{X}\right\Vert _{1}\\
&  =\sum_{x}\left\Vert \sqrt{\rho}\left(  \Lambda_{x}^{T}-\widetilde{\Lambda
}_{x}^{T}\right)  \sqrt{\rho}\right\Vert _{1}\\
&  =\sum_{x}\left\Vert \sqrt{\rho}\left(  \Lambda_{x}-\widetilde{\Lambda}%
_{x}\right)  \sqrt{\rho}\right\Vert _{1},
\end{align*}
where the third equality follows because the trace norm of a block-diagonal
operator is just the sum of the trace norms of the blocks. The fourth equality
follows because%
\[
\sqrt{\rho}\left(  \Lambda_{x}^{T}-\widetilde{\Lambda}_{x}^{T}\right)
\sqrt{\rho}=\left[  \sqrt{\rho}\left(  \Lambda_{x}-\widetilde{\Lambda}%
_{x}\right)  \sqrt{\rho}\right]  ^{T},
\]
the trace norm depends only on the singular values of a matrix, and these are
invariant under transposition.
\end{proof}

\subsection{Measurement compression theorem}

We can now state Winter's main result:

\begin{theorem}
[Measurement compression theorem]\label{thm:meas-comp}Let $\rho$ be a source
state and $\Lambda$ a POVM\ to simulate on this state. A protocol for a
faithful feedback simulation of the POVM\ with classical
communication rate $R$ and common randomness rate $S$ exists if and only if the
following set of inequalities hold%
\begin{align*}
R  &  \geq I\left(  X;R\right)  ,\\
R+S  &  \geq H\left(  X\right)  ,
\end{align*}
where the entropies are with respect to the state%
\begin{equation}
\sum_{x}\left\vert x\right\rangle \left\langle x\right\vert ^{X}%
\otimes\text{\emph{Tr}}_{A}\left\{  \left(  I^{R}\otimes\Lambda_{x}%
^{A}\right)  \phi^{RA}\right\}  , \label{eq:IC-state}%
\end{equation}
and $\phi^{RA}$ is any purification of the state $\rho$.
\end{theorem}

Note that $I(X;R)$ and $H(X)$ are independent of the choice of purification $\phi^{RA}$.
Moreover, the entropies are invariant with respect to transposition in the
basis that diagonalizes $\rho$ so that we could instead evaluate entropies with
respect to the following classical-quantum state:%
\[
\sum_{x}\left\vert x\right\rangle \left\langle x\right\vert ^{X}%
\otimes\text{Tr}_{A}\left\{  \left(  I^{R}\otimes\left(  \Lambda_{x}%
^{T}\right)  ^{A}\right)  \phi^{RA}\right\}  .
\]

Figure~\ref{fig:IC-rate-region}\ provides a plot of the optimal rate region
given in the above theorem for this case of a feedback simulation in which the
sender also obtains the outcome of the measurement simulation.

\begin{figure}[ptb]
\begin{center}
\includegraphics[
natheight=5.199200in,
natwidth=7.260100in,
height=2.4716in,
width=3.4411in
]{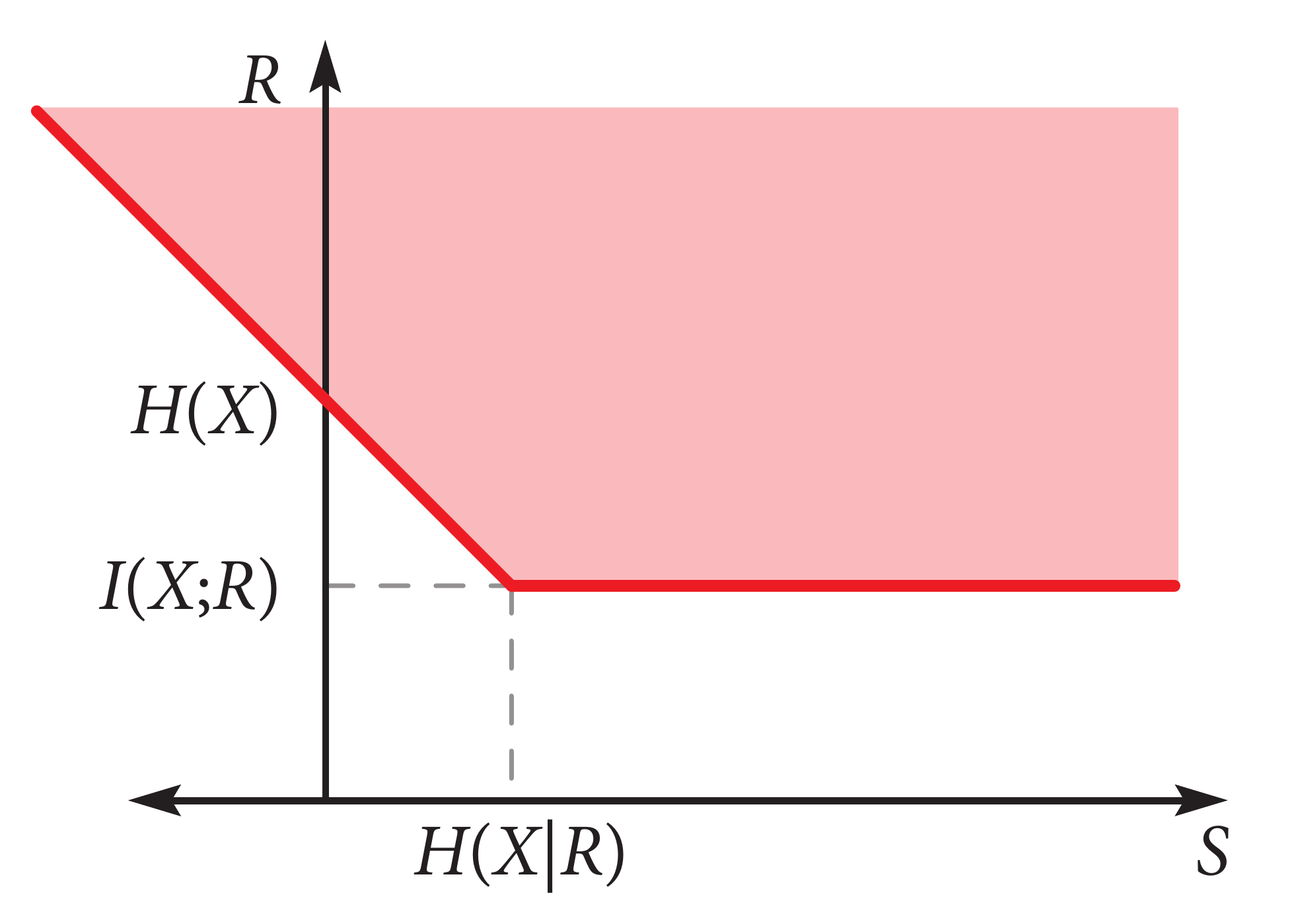}
\end{center}
\caption{\textbf{Optimal rate region for measurement compression with
feedback.} The figure plots the optimal rate region from
Theorem~\ref{thm:meas-comp}. The measurement compression protocol demonstrates
that the rate pair $\left(  S=H\left(  X|R\right)  ,R=I\left(  X;R\right)
\right)  $ is achievable. Wasting common randomness achieves all of the rate
pairs to the right of this corner point. Time-sharing between measurement
compression and Shannon compression $\left(  S=0,R=H\left(  X\right)  \right)
$ achieves all of the rate pairs between them. Finally, employing Shannon
compression and converting the extra classical communication to common
randomness achieves all of the optimal rate pairs along the line extending
northwest from Shannon compression. The converse theorem in
Section~\ref{sec:IC-converse}~proves that this rate region is optimal.}%
\label{fig:IC-rate-region}%
\end{figure}

After giving a simple example of an application of the above theorem, we prove
it in two parts.\ First, we prove that there exists a measurement compression
protocol achieving the rates in the above theorem, specifically the corner
point $\left(  S=H\left(  X|R\right)  ,R=I\left(  X;R\right)  \right)  $.
Next, we prove the converse part of the theorem:\ that one cannot do better
than the rates given in the above theorem.

\subsubsection{Examples}

\label{sec:meas-comp-example}We review two simple examples to illustrate some
applications of Theorem~\ref{thm:meas-comp}. Our first example is for the case
that the initial state on which Alice performs the measurement is some pure
state $\phi^{A}$. In this case, the state in (\ref{eq:IC-state}) becomes%
\[
\sum_{x}\left\vert x\right\rangle \left\langle x\right\vert ^{X}%
\otimes\text{Tr}_{A}\left\{  \left(  I^{R}\otimes\Lambda_{x}^{A}\right)
\left(  \psi^{R}\otimes\phi^{A}\right)  \right\}  =\sum_{x}\text{Tr}\left\{
\Lambda_{x}\phi\right\}  \left\vert x\right\rangle \left\langle x\right\vert
^{X}\otimes\psi^{R}.
\]
Thus, the reference has no correlations with the outcome of the measurement,
so that $I\left(  X;R\right)  =0$ and $H\left(  X|R\right)  =H\left(
X\right)  $, where $H\left(  X\right)  $ is the Shannon entropy of the
distribution $p\left(  x\right)  =\ $Tr$\left\{  \Lambda_{x}\phi\right\}  $.
No classical communication is required in this case---common randomness
suffices for this simulation. Indeed, the protocol just has Alice and Bob
operate as in randomness dilution, whereby they dilute their uniform, shared
randomness to match the distribution $p\left(  x\right)  $. The idea here is
that there are no correlations with some reference system, or similarly, there
is only a trivial decomposition of the form in (\ref{eq:true-joint-dist}), so
that $K$ is a degenerate random variable.

Our next example is a natural one discussed in the conclusion of
Ref.~\cite{DW02}. Consider the POVM%
\begin{equation}
\left\{  \frac{1}{2}\left\vert 0\right\rangle \left\langle 0\right\vert
,\frac{1}{2}\left\vert 1\right\rangle \left\langle 1\right\vert ,\frac{1}%
{2}\left\vert +\right\rangle \left\langle +\right\vert ,\frac{1}{2}\left\vert
-\right\rangle \left\langle -\right\vert \right\}  \label{eq:example-POVM}%
\end{equation}
acting on the maximally mixed state $\pi^{A}=I^{A}/2$. We would like to
determine the resources required to simulate the action of this measurement on
the maximally mixed state. Consider that the Bell state%
\[
\left\vert \Phi\right\rangle ^{RA}\equiv\frac{1}{\sqrt{2}}\left(  \left\vert
00\right\rangle ^{RA}+\left\vert 11\right\rangle ^{RA}\right)  =\frac{1}%
{\sqrt{2}}\left(  \left\vert ++\right\rangle ^{RA}+\left\vert --\right\rangle
^{RA}\right)
\]
is a purification of the maximally mixed state $\pi^{A}$. The post-measurement
classical-quantum state in (\ref{eq:IC-state}) for this case is as follows:%
\[
\frac{1}{4}\left(  \left\vert 0\right\rangle \left\langle 0\right\vert
^{X}\otimes\left\vert 0\right\rangle \left\langle 0\right\vert ^{R}+\left\vert
1\right\rangle \left\langle 1\right\vert ^{X}\otimes\left\vert 1\right\rangle
\left\langle 1\right\vert ^{R}+\left\vert 2\right\rangle \left\langle
2\right\vert ^{X}\otimes\left\vert +\right\rangle \left\langle +\right\vert
^{R}+\left\vert 3\right\rangle \left\langle 3\right\vert ^{X}\otimes\left\vert
-\right\rangle \left\langle -\right\vert ^{R}\right)  .
\]
A simple calculation reveals that the mutual information $I\left(  X;R\right)
$ of the above state is equal to one bit. Also, the conditional entropy
$H\left(  X|R\right)  $ of the above state is equal to one bit. Thus, one bit
of classical communication and one bit of common randomness are required to
simulate this measurement.

For this case, the simulation is straightforward since the
original POVM\ decomposes as a random choice of a $Z$ or $X$ Pauli
measurement:%
\[
\left\{  \frac{1}{2}\left\{  \left\vert 0\right\rangle \left\langle
0\right\vert ,\left\vert 1\right\rangle \left\langle 1\right\vert \right\}
,\frac{1}{2}\left\{  \left\vert +\right\rangle \left\langle +\right\vert
,\left\vert -\right\rangle \left\langle -\right\vert \right\}  \right\}  .
\]
Thus, Alice and Bob can use one bit of common randomness to select which
measurement to perform. Alice then performs the $Z$ or $X$ measurement and
sends the outcome to Bob using one classical bit channel. Bob can then
determine which of the four outcomes in (\ref{eq:example-POVM}) has occurred
by combining the two bits.

\subsection{Achievability proof for measurement compression}

\label{sec:MC-achieve}The resource inequality \cite{DHW04,DHW08}%
\ characterizing measurement compression is as follows:%
\[
I\left(  X;R\right)  \left[  c\rightarrow c\right]  +H\left(  X|R\right)
\left[  cc\right]  \geq\left\langle \Lambda\left(  \rho\right)  \right\rangle
,
\]
where the entropic quantities are with respect to a state of the following
form:%
\begin{align*}
\sum_{x}\left\vert x\right\rangle \left\langle x\right\vert ^{X}%
\otimes\text{Tr}_{A}\left\{  \left(  I^{R}\otimes\Lambda_{x}^{A}\right)
\phi^{RA}\right\}   &  =\sum_{x}p_{X}\left(  x\right)  \left\vert
x\right\rangle \left\langle x\right\vert ^{X}\otimes\theta_{x}^{R},\\
\theta_{x}^{R}  &  \equiv\text{Tr}_{A}\left\{  \left(  I^{R}\otimes\Lambda
_{x}^{A}\right)  \phi^{RA}\right\}  /p_{X}\left(  x\right)  ,\\
p_{X}\left(  x\right)   &  \equiv\text{Tr}\left\{  \left(  I^{R}\otimes
\Lambda_{x}^{A}\right)  \phi^{RA}\right\}  ,
\end{align*}
and $\phi^{RA}$ is some purification of the state $\rho$. The operators
$\theta_{x}^{R}$ take on the following special form%
\[
\theta_{x}^{R}=\sqrt{\rho}\Lambda_{x}^{T}\sqrt{\rho},
\]
if the spectral decomposition of $\rho$ is $\rho=\sum_{x}\lambda_{x}\left\vert
x\right\rangle \left\langle x\right\vert $ and the purification $\phi^{RA}$ is
taken as $\left\vert \phi\right\rangle ^{RA}=\sum_{x}\sqrt{\lambda_{x}%
}\left\vert x\right\rangle ^{R}\left\vert x\right\rangle ^{A}$. The meaning of
the above resource inequality is that $nI\left(  X;R\right)  $ bits of
classical communication $\left[  c\rightarrow c\right]  $\ and $nH\left(
X|R\right)  $ bits of common randomness $\left[  cc\right]  $\ are required in
order to simulate the action of the POVM\ $\Lambda^{\otimes n}$ on the tensor
product state $\rho^{\otimes n}$, and the simulation becomes exact in the
asymptotic limit as $n\rightarrow\infty$.

The main idea of the proof is to \textquotedblleft steer\textquotedblright%
\ the state of the reference to be close to the ensemble produced by the ideal
measurement. In order to do so, we construct a measurement at random, chosen
from an ensemble of operators built from the ideal measurement $\Lambda$ and
the state $\rho$. By employing the Ahlswede-Winter operator Chernoff bound
\cite{AW02}, we can then guarantee that there exists a particular
POVM\ satisfying the faithful simulation condition in (\ref{def:faith-sim-2}),
as long as the amount of classical communication and common randomness is
sufficiently large.

The achievability part of the theorem begins by considering the following
ensemble derived from the state $\rho$ and the POVM\ $\Lambda=\left\{
\Lambda_{x}\right\}  $:%
\begin{align*}
p_{X}\left(  x\right)   &  \equiv\text{Tr}\left\{  \Lambda_{x}\rho\right\}
,\\
\hat{\rho}_{x}  &  \equiv\frac{1}{p_{X}\left(  x\right)  }\sqrt{\rho}%
\Lambda_{x}\sqrt{\rho}.
\end{align*}
(Recall our statement that the entropies are invariant with respect to
transposition in the basis that diagonalizes $\rho$.) Observe that the
expected density operator of this ensemble is just the state$~\rho$:%
\[
\sum_{x}p_{X}\left(  x\right)  \hat{\rho}_{x}=\sum_{x}\sqrt{\rho}\Lambda
_{x}\sqrt{\rho}=\rho.
\]
We will prove that there exist POVMs $\Gamma^{\left(  m\right)  }%
=\{\Gamma_{x^{n}}^{\left(  m\right)  }\}_{x^{n}\in\mathcal{L}}$ with
$m\in\mathcal{M}$,
\begin{align}
\left\vert \mathcal{L}\right\vert  &  =2^{n\left[  I\left(  X;R\right)
+3\delta\right]  },\label{eq:POVM-size}\\
\left\vert \mathcal{M}\right\vert  &  =2^{n\left[  H\left(  X|R\right)
+\delta\right]  }, \label{eq:randomness-size}%
\end{align}
for some $\delta>0$, such that the mixed POVM\ $\widetilde{\Lambda}$ with
elements $\widetilde{\Lambda}_{x^{n}}\equiv\frac{1}{\left\vert \mathcal{M}%
\right\vert }\sum_{m}\Gamma_{x^{n}}^{\left(  m\right)  }$ provides a faithful
simulation of $\Lambda$ on $\rho$ according to the criterion in
(\ref{eq:faithful-sim-cond-1}).

In order to prove achievability, we require the Ahlswede-Winter Operator
Chernoff bound, which we recall below:

\begin{lemma}
[Operator Chernoff Bound]\label{lemma:operator-chernoff}Let $\xi_{1}%
,\ldots,\xi_{M}$ be $M$\ independent and identically distributed random
variables with values in the algebra $\mathcal{B}\left(  \mathcal{H}\right)  $
of linear operators acting on some finite dimensional Hilbert space $\mathcal{H}$. Each $\xi
_{m}$ has all of its eigenvalues between zero and one, so that the following operator inequality holds%
\begin{equation}
\forall m\in\left[  M\right]  :0\leq\xi_{m}\leq I.
\end{equation}
Let $\overline{\xi}$ denote the sample average of the $M$ random variables:%
\begin{equation}
\overline{\xi}=\frac{1}{M}\sum_{m=1}^{M}\xi_{m}.
\end{equation}
Suppose that the expectation $\mathbb{E}_{\xi}\left\{  \xi_{m}\right\}
\equiv\mu$\ of each operator $\xi_{m}$ exceeds the identity operator scaled by
a number $a>0$:%
\begin{equation}
\mu\geq aI.
\end{equation}
Then for every $\eta$ where $0<\eta<1/2$ and $a\left(  1+\eta\right)  \leq1$,
we can bound the probability that the sample average $\overline{\xi}$\ lies
inside the operator interval $\left[  \left(  1\pm\eta\right)  \mu\right]  $:%
\begin{equation}
\Pr_{\xi}\left\{  \left(  1-\eta\right)  \mu\leq\overline{\xi}\leq\left(
1+\eta\right)  \mu\right\}  \geq1-2\dim\mathcal{H}\exp\left(  -\frac{M\eta
^{2}a}{4\ln2}\right)  .
\end{equation}
Thus it is highly likely that the sample average operator $\overline{\xi}$
becomes close to the true expected operator~$\mu$ as $M$ becomes large.
\end{lemma}

We first define some operators that we will use to generate the POVM elements
$\Gamma_{x^{n}}^{\left(  m\right)  }$. For all $x^{n}\in T_{\delta}^{X^{n}}$
(where $T_{\delta}^{X^{n}}$ is the strongly typical set---see
Appendix~\ref{sec:typ-review}) consider the following positive operators with
trace less than one:%
\begin{equation}
\xi_{x^{n}}^{\prime}\equiv\Pi_{\rho,\delta}^{n}\ \Pi_{\hat{\rho}_{x^{n}%
},\delta}\ \hat{\rho}_{x^{n}}\ \Pi_{\hat{\rho}_{x^{n}},\delta}\ \Pi
_{\rho,\delta}^{n}. \label{eq:POVM-build-up}%
\end{equation}
These operators $\xi_{x^{n}}^{\prime}$ have a trace almost equal to one
because%
\begin{align}
\text{Tr}\left\{  \xi_{x^{n}}^{\prime}\right\}   &  =\text{Tr}\left\{
\Pi_{\rho,\delta}^{n}\ \Pi_{\hat{\rho}_{x^{n}},\delta}\ \hat{\rho}_{x^{n}%
}\ \Pi_{\hat{\rho}_{x^{n}},\delta}\ \Pi_{\rho,\delta}^{n}\right\} \nonumber\\
&  =\text{Tr}\left\{  \Pi_{\rho,\delta}^{n}\ \Pi_{\hat{\rho}_{x^{n}},\delta
}\ \hat{\rho}_{x^{n}}\ \Pi_{\hat{\rho}_{x^{n}},\delta}\right\} \nonumber\\
&  \geq\text{Tr}\left\{  \Pi_{\rho,\delta}^{n}\ \hat{\rho}_{x^{n}}\right\}
-\left\Vert \hat{\rho}_{x^{n}}-\Pi_{\hat{\rho}_{x^{n}},\delta}\ \hat{\rho
}_{x^{n}}\ \Pi_{\hat{\rho}_{x^{n}},\delta}\right\Vert _{1}\nonumber\\
&  \geq1-\epsilon-2\sqrt{\epsilon}. \label{eq:prime-states-large-trace}%
\end{align}
The first inequality follows from the trace inequality in
Lemma~\ref{lem:trace-inequality}, and the second inequality follows by
appealing to the properties of quantum typicality reviewed in
Appendix~\ref{sec:typ-review}. Also, we set $S$ to be the probability of the
typical set $T_{\delta}^{X^{n}}$, and recall that this probability is near to
one:%
\[
S\equiv\Pr\left\{  X^{n}\in T_{\delta}^{X^{n}}\right\}  =\sum_{x^{n}\in
T_{\delta}^{X^{n}}}p_{X^{n}}\left(  x^{n}\right)  \geq1-\epsilon.
\]
We define $\xi^{\prime}$ to be the expectation of the operators $\xi_{x^{n}%
}^{\prime}$, when each one is chosen according to a pruned distribution
$p_{X^{\prime n}}\left(  x^{n}\right)  $:%
\[
\xi^{\prime}\equiv\mathbb{E}_{X^{\prime n}}\left\{  \xi_{X^{\prime n}}%
^{\prime}\right\}  =\sum_{x^{n}}p_{X^{\prime n}}\left(  x^{n}\right)
\xi_{x^{n}}^{\prime},
\]
where we define $p_{X^{\prime n}}\left(  x^{n}\right)  $ as%
\begin{equation}
p_{X^{\prime n}}\left(  x^{n}\right)  \equiv\left\{
\begin{array}
[c]{cc}%
p_{X^{n}}\left(  x^{n}\right)  /S & x^{n}\in T_{\delta}^{X^{n}}\\
0 & \text{else}%
\end{array}
\right.  . \label{eq:pruned-dist}%
\end{equation}
It follows that Tr$\left\{  \xi^{\prime}\right\}  \geq1-\epsilon
-2\sqrt{\epsilon}$ because%
\begin{align}
\text{Tr}\left\{  \xi^{\prime}\right\}   &  =\sum_{x^{n}}p_{X^{\prime n}%
}\left(  x^{n}\right)  \text{Tr}\left\{  \xi_{x^{n}}^{\prime}\right\}
\nonumber\\
&  \geq1-\epsilon-2\sqrt{\epsilon}. \label{eq:avg-prime-large-trace}%
\end{align}
The inequality follows from the one in (\ref{eq:prime-states-large-trace}).
From properties of quantum typicality, we know that%
\begin{equation}
\Pi_{\rho,\delta}^{n}\ \rho^{\otimes n}\ \Pi_{\rho,\delta}^{n}\geq2^{-n\left[
H\left(  \rho\right)  +\delta\right]  }\ \Pi_{\rho,\delta}^{n}.
\label{eq:typ-prop-3}%
\end{equation}
We now define $\Pi$ to be the projector onto the subspace spanned by the
eigenvectors of $\xi^{\prime}$ with eigenvalue larger than $\epsilon\alpha$,
where $\alpha\equiv2^{-n\left[  H\left(  \rho\right)  +\delta\right]
}=2^{-n\left[  H\left(  R\right)  +\delta\right]  }$. Defining the operator
$\Omega$ as%
\begin{equation}
\Omega\equiv\Pi\xi^{\prime}\Pi, \label{eq:def-Omega}%
\end{equation}
it follows that Tr$\left\{  \Omega\right\}  \geq1-2\epsilon-2\sqrt{\epsilon}$
because%
\[
\text{rank}(\Omega)\leq\text{Tr}\left\{  \Pi\right\}  \leq\text{Tr}\left\{
\Pi_{\rho,\delta}^{n}\right\}  \leq2^{n\left[  H\left(  \rho\right)
+\delta\right]  }=\alpha^{-1},
\]
so that eigenvalues smaller than $\epsilon\alpha$ contribute at most
$\epsilon$ to Tr$\left\{  \Omega\right\}  $, giving%
\begin{equation}
\text{Tr}\left\{  \Omega\right\}  \geq\left(  1-\epsilon\right)
\text{Tr}\left\{  \xi^{\prime}\right\}  \geq\left(  1-\epsilon\right)  \left(
1-\epsilon-2\sqrt{\epsilon}\right)  \geq1-2\epsilon-2\sqrt{\epsilon}.
\label{eq:Omega-high-trace}%
\end{equation}

We now exploit a random selection of the operators in (\ref{eq:POVM-build-up})
in order to build up a POVM\ that has desirable properties that we can use to
prove the achievability part of this theorem. Let $\xi_{x^{n}}$ denote the
following operators:%
\[
\xi_{x^{n}}\equiv\Pi\ \xi_{x^{n}}^{\prime}\ \Pi,
\]
so that we confine them to be in the subspace onto which $\Pi$ projects.
Define $\left\vert \mathcal{L}\right\vert \left\vert \mathcal{M}\right\vert
$ random variables $X^{n}\left(  l,m\right)  $ that are chosen independently according to the
pruned distribution $p_{X^{\prime n}}\left(  x^{n}\right)  $. We can group
these variables into $\left\vert \mathcal{M}\right\vert $ sets $\mathcal{C}%
_{m}\equiv\left\{  X^{n}\left(  l,m\right)  \right\}  _{l\in\mathcal{L}}$,
according to the value of the common randomness $m$. Under the pruned
distribution, the expectation of the random operator $\xi_{X^{n}\left(
l,m\right)  }$ is equal to $\Omega$:%
\[
\mathbb{E}_{X^{n}\left(  l,m\right)  }\left\{  \xi_{X^{n}\left(  l,m\right)
}\right\}  =\sum_{x^{n}}p_{X^{\prime n}}\left(  x^{n}\right)  \xi_{x^{n}%
}=\Omega.
\]
Let $E_{m}$ denote the event that the sample average of the operators in the
$m^{\text{th}}$ set $\mathcal{C}_{m}$ falls close to its mean $\Omega$ (in the
operator interval sense):%
\begin{equation}
\Omega\left(  1-\epsilon\right)  \leq\frac{1}{\left\vert \mathcal{L}%
\right\vert }\sum_{l}\xi_{x^{n}\left(  l,m\right)  }\leq\Omega\left(
1+\epsilon\right)  . \label{eq:event-m}%
\end{equation}
The above event $E_{m}$ is equivalent to the following rescaled event:%
\[
\beta\Omega\left(  1-\epsilon\right)  \leq\frac{\beta}{\left\vert
\mathcal{L}\right\vert }\sum_{l}\xi_{x^{n}\left(  l,m\right)  }\leq\beta
\Omega\left(  1+\epsilon\right)  ,
\]
where $\beta\equiv2^{n\left[  H\left(  R|X\right)  -\delta\right]  }$. It is
then clear that the expectation of the operators $\beta\xi_{x^{n}\left(
l,m\right)  }$ satisfies the following operator inequality needed in the
operator Chernoff bound:%
\[
\mathbb{E}_{X^{\prime n}}\left\{  \beta\xi_{X^{n}\left(  l,m\right)
}\right\}  =\beta\Omega\geq\alpha\beta\epsilon\Pi.
\]
Also, each individual rescaled operator $\beta\xi_{x^{n}\left(  l,m\right)  }$
admits a tight operator upper bound with the identity operator because%
\begin{align*}
\beta\xi_{x^{n}\left(  l,m\right)  }  &  =2^{n\left[  H\left(  R|X\right)
-\delta\right]  }\ \Pi\ \Pi_{\rho,\delta}^{n}\ \Pi_{\hat{\rho}_{x^{n}},\delta
}\ \hat{\rho}_{x^{n}}\ \Pi_{\hat{\rho}_{x^{n}},\delta}\ \Pi_{\rho,\delta}%
^{n}\ \Pi\\
&  \leq2^{n\left[  H\left(  R|X\right)  -\delta\right]  }\ 2^{-n\left[
H\left(  R|X\right)  -\delta\right]  }\ \Pi\ \Pi_{\rho,\delta}^{n}\ \Pi
_{\hat{\rho}_{x^{n}},\delta}\ \Pi_{\rho,\delta}^{n}\ \Pi\\
&  =\Pi\ \Pi_{\rho,\delta}^{n}\ \Pi_{\hat{\rho}_{x^{n}},\delta}\ \Pi
_{\rho,\delta}^{n}\ \Pi\\
&  \leq I,
\end{align*}
where we applied the operator inequality $\Pi_{\hat{\rho}_{x^{n}},\delta
}\ \hat{\rho}_{x^{n}}\ \Pi_{\hat{\rho}_{x^{n}},\delta}\leq2^{-n\left[
H\left(  R|X\right)  -\delta\right]  }\ \Pi_{\hat{\rho}_{x^{n}},\delta}$ for
the first inequality. Applying the operator Chernoff bound then gives us an
upper bound on the probability that event $E_{m}$ does not occur%
\begin{align*}
\Pr\left\{  \neg E_{m} \right\}   &  \leq2\,\text{rank}(\Pi)\exp\left(  -\frac{\left\vert
\mathcal{L}\right\vert \epsilon^{2}\left(  \epsilon\alpha\beta\right)  }%
{4\ln2}\right) \\
&  \leq2\cdot2^{n\left[  H\left(  R\right)  +\delta\right]  }\exp\left(
-\frac{2^{n\left[  I\left(  X;R\right)  +3\delta\right]  }\epsilon
^{3}2^{-n\left[  H\left(  R\right)  +\delta\right]  }2^{n\left[  H\left(
R|X\right)  -\delta\right]  }}{4\ln2}\right) \\
&  \leq2\cdot2^{n\left[  H\left(  R\right)  +\delta\right]  }\exp\left(
-\frac{2^{n\delta}\epsilon^{3}}{4\ln2}\right) \\
&  =2\exp\left(  -\frac{2^{n\delta}\epsilon^{3}}{4\ln2}+n\left[  H\left(
R\right)  +\delta\right]  \ln2\right)  .
\end{align*}
Thus, by choosing $\left\vert \mathcal{L}\right\vert $ as we did in
(\ref{eq:POVM-size}), it is possible to make the probability of the complement
of $E_{m}$ doubly-exponentially small in $n$. Also, the above application
of the operator Chernoff bound makes it clear why we rescale according to
$\beta$---doing so allows for the rescaled operators $\beta\xi_{x^{n}\left(
l,m\right)  }$ to admit a tight operator upper bound with the identity (so
that the demands of the operator Chernoff bound are met), while allowing for
$\left\vert \mathcal{L}\right\vert $ to be as small as $2^{n\left[  I\left(
X;R\right)  +3\delta\right]  }$ and $\Pr\left\{  E_{m}^{c}\right\}  $ to be
arbitrarily small.

We now define a counting function $c_{x^{n}}\left(  \mathcal{L},\mathcal{M}%
\right)  $\ on the sets $\mathcal{L}$ and $\mathcal{M}$, which counts the
fraction of occurrences of a sequence $x^{n}\in T_{\delta}^{X^{n}}$ in the set
$\left\{  x^{n}\left(  l,m\right)  \right\}  _{l\in\mathcal{L},m\in
\mathcal{M}}$:%
\[
c_{x^{n}}\left(  \mathcal{L},\mathcal{M}\right)  \equiv\frac{1}{\left\vert
\mathcal{L}\right\vert \left\vert \mathcal{M}\right\vert }\left\vert \left\{
l,m:x^{n}\left(  l,m\right)  =x^{n}\right\}  \right\vert .
\]
(This is effectively a sample average of the counts.)\ When choosing the
random variables $X^{n}\left(  l,m\right)  $ IID according to the pruned
distribution, the expectation of the random counting function $C_{x^{n}%
}\left(  \mathcal{L},\mathcal{M}\right)  $ is equal to the probability of the
sequence $x^{n}$:%
\[
\mathbb{E}\left\{  C_{x^{n}}\left(  \mathcal{L},\mathcal{M}\right)  \right\}
=p_{X^{\prime n}}\left(  x^{n}\right)  \text{.}%
\]
Thus, for any sequence $x^{n}\in T_{\delta}^{X^{n}}$, the expectation of the
counting function has the following lower bound:%
\begin{align*}
\mathbb{E}\left\{  C_{x^{n}}\left(  \mathcal{L},\mathcal{M}\right)  \right\}
&  =p_{X^{\prime n}}\left(  x^{n}\right) \\
&  =\frac{1}{S}p_{X^{n}}\left(  x^{n}\right) \\
&  \geq\min\left\{  p_{X^{n}}\left(  x^{n}\right)  :x^{n}\in T_{\delta}%
^{X^{n}}\right\} \\
&  \geq\gamma\equiv2^{-n\left[  H\left(  X\right)  +\delta\right]  },
\end{align*}
where $S$, recall, is the probability that a random sequence $X^n$ is typical.

In order to appeal to the operator Chernoff bound (we could just use the
classical one, but we instead choose to exploit the operator one), we define
$\hat{P}$ as a diagonal density operator of dimension $\left\vert T_{\delta
}^{X^{n}}\right\vert \times\left\vert T_{\delta}^{X^{n}}\right\vert $, whose
diagonal entries are just the entries of the pruned distribution $p_{X^{\prime
n}}\left(  x^{n}\right)  $. Similarly, for a particular realization of the set
$\left\{  x^{n}\left(  l,m\right)  \right\}  _{l\in\mathcal{L},m\in
\mathcal{M}}$, we can define $\hat{C}$ as a diagonal density operator of the
same dimension, whose diagonal entries are just the entries of the counting
functions $c_{x^{n}}\left(  \mathcal{L},\mathcal{M}\right)  $. From the above
reasoning, it is then clear that the expectation of $\hat{C}$ under a random
choice of the $x^{n}\left(  l,m\right)  $ sequences is just $\hat{P}$:%
\[
\mathbb{E}\left\{  \hat{C}\right\}  =\hat{P}.
\]
Furthermore, (again by the above reasoning), we can establish the following
lower bound on $\hat{P}$:%
\[
\hat{P}\geq\gamma\ \Pi_{p_{X},\delta}^{n},
\]
where $\Pi_{p_{X},\delta}^{n}$ is a typical projector corresponding to the
distribution $p_{X}\left(  x\right)  $.

Let $E_{0}$ be the event that the operator $\hat{C}$ is within $\epsilon$ of
its mean $\hat{P}$:%
\begin{equation}
\left(  1-\epsilon\right)  \hat{P}\leq\hat{C}\leq\left(  1+\epsilon\right)
\hat{P}. \label{eq:condition-2}%
\end{equation}
By appealing to the operator Chernoff bound, we can bound the probability that
the above event does not occur when choosing the sequences $x^{n}\left(
l,m\right)  $ randomly as prescribed above:%
\begin{align*}
\Pr\left\{  \neg E_{0}\right\}   &  \leq2 \, \text{rank}(\Pi_{p_{X},\delta}^{n}) \exp\left(
-\frac{\left\vert \mathcal{L}\right\vert \left\vert \mathcal{M}\right\vert
\epsilon^{2}\gamma}{4\ln2}\right) \\
&  =2\cdot2^{n\left[  H\left(  X\right)  +\delta\right]  }\exp\left(
-\frac{2^{n\left[  I\left(  X;R\right)  +3\delta\right]  }2^{n\left[  H\left(
X|R\right)  +\delta\right]  }\epsilon^{2}2^{-n\left[  H\left(  X\right)
+\delta\right]  }}{4\ln2}\right) \\
&  =2\cdot2^{n\left[  H\left(  X\right)  +\delta\right]  }\exp\left(
-\frac{2^{n3\delta}}{4\ln2}\right) \\
&  =2\exp\left(  -\frac{2^{n3\delta}}{4\ln2}+n\left[  H\left(  X\right)
+\delta\right]  \ln2\right)  .
\end{align*}
Thus, by choosing $\left\vert \mathcal{L}\right\vert \left\vert \mathcal{M}%
\right\vert $ as we did in (\ref{eq:POVM-size}-\ref{eq:randomness-size}), it
is possible to make this probability be doubly-exponentially small in $n$.

We want to ensure that it is possible for all of the events $E_{m}$ and
$E_{0}$ to occur simultaneously. We can guarantee this by applying DeMorgan's
law, the union bound, and the above estimates:%
\begin{align}
&  \Pr\left\{  \neg \left[  E_{0}\cap\left(  \bigcap\limits_{m}E_{m}\right)
\right]  \right\} \nonumber\\
&  =\Pr\left\{  \neg E_{0} \cup\left(  \bigcup\limits_{m}\neg E_{m}\right)
\right\} \nonumber\\
&  \leq\Pr\left\{  \neg E_{0} \right\}  +\sum_{m}\Pr\left\{  \neg E_{m} \right\}
\nonumber\\
&  \leq2\exp\left(  -\frac{2^{n3\delta}}{4\ln2}+n\left[  H\left(  X\right)
+\delta\right]  \ln2\right)  +\left\vert \mathcal{M}\right\vert 2\exp\left(
-\frac{2^{n\delta}\epsilon^{3}}{4\ln2}+n\left[  H\left(  R\right)
+\delta\right]  \ln2\right)  , \label{eq:union-bound-arg}%
\end{align}
which becomes arbitrarily small as $n\rightarrow\infty$. (Thus, it is in fact
overwhelmingly likely for our desired conditions to hold if we choose
$\left\vert \mathcal{L}\right\vert $ and $\left\vert \mathcal{M}\right\vert $
as we did in (\ref{eq:POVM-size}-\ref{eq:randomness-size}).)

So, assume now that we have a set $\left\{  x^{n}\left(  l,m\right)  \right\}
_{l\in\mathcal{L},m\in\mathcal{M}}$ such that the corresponding operators
$\left\{  \xi_{x^{n}\left(  l,m\right)  }\right\}  $\ and $\hat{C}$\ satisfy
the conditions in (\ref{eq:event-m}) and (\ref{eq:condition-2}). We can now
construct from them a set of POVMs\ $\left\{  \Gamma^{\left(  m\right)
}\right\}  $ that will perform a faithful measurement simulation. We define
the POVM elements $\Gamma_{x^{n}}^{\left(  m\right)  }$ of $\Gamma^{\left(
m\right)  }$ as follows:%
\begin{align}
\Gamma_{x^{n}}^{\left(  m\right)  }  &  \equiv\frac{S}{1+\epsilon}%
\omega^{-1/2}\left(  \frac{1}{\left\vert \mathcal{L}\right\vert }%
\sum_{l\ :\ x^{n}\left(  l,m\right)  =x^{n}}\xi_{x^{n}\left(  l,m\right)
}\right)  \omega^{-1/2}\label{eq:POVM-def-meas-comp}\\
&  =\frac{S}{1+\epsilon}\ \frac{\left\vert \left\{  l:x^{n}\left(  l,m\right)
=x^{n}\right\}  \right\vert }{\left\vert \mathcal{L}\right\vert }%
\ \omega^{-1/2}\ \xi_{x^{n}}\ \omega^{-1/2}.\nonumber
\end{align}
We check that for each value $m$ of the common randomness that these operators
form a sub-POVM (a set of positive operators whose sum is upper bounded by the
identity). Indeed, we can appeal to the fact that the operators satisfy the
condition in (\ref{eq:event-m}):%
\begin{align*}
\sqrt{\omega}\ \sum_{x^{n}\in\mathcal{X}^{n}}\Gamma_{x^{n}}^{\left(  m\right)
}\ \sqrt{\omega}  &  =\frac{S}{1+\epsilon}\frac{1}{\left\vert \mathcal{L}%
\right\vert }\sum_{x^{n}\in\mathcal{X}^{n}}\left(  \sum_{l\ :\ x^{n}\left(
l,m\right)  =x^{n}}\xi_{x^{n}\left(  l,m\right)  }\right) \\
&  =\frac{S}{1+\epsilon}\frac{1}{\left\vert \mathcal{L}\right\vert }\sum
_{l}\xi_{x^{n}\left(  l,m\right)  }\\
&  \leq\frac{S}{1+\epsilon}\ \Omega\left(  1+\epsilon\right) \\
&  =S\ \Omega,
\end{align*}
where the inequality appeals to (\ref{eq:event-m}). Continuing with the
definition of $\Omega$ in (\ref{eq:def-Omega}), we have%
\begin{align*}
&  =S\sum_{x^{n}\in\mathcal{X}^{n}}p_{X^{\prime n}}\left(  x^{n}\right)
\Pi\ \Pi_{\rho,\delta}^{n}\ \Pi_{\hat{\rho}_{x^{n}},\delta}\ \hat{\rho}%
_{x^{n}}\ \Pi_{\hat{\rho}_{x^{n}},\delta}\ \Pi_{\rho,\delta}^{n}\ \Pi\\
&  \leq\sum_{x^{n}\in T_{\delta}^{X^{n}}}p_{X^{n}}\left(  x^{n}\right)
\Pi\ \Pi_{\rho,\delta}^{n}\ \hat{\rho}_{x^{n}}\ \Pi_{\rho,\delta}^{n}\ \Pi\\
&  \leq\sum_{x^{n}\in\mathcal{X}^{n}}p_{X^{n}}\left(  x^{n}\right)  \Pi
\ \Pi_{\rho,\delta}^{n}\ \hat{\rho}_{x^{n}}\ \Pi_{\rho,\delta}^{n}\ \Pi\\
&  =\Pi\ \Pi_{\rho,\delta}^{n}\ \rho^{\otimes n}\ \Pi_{\rho,\delta}^{n}\ \Pi\\
&  \leq\rho^{\otimes n}=\omega.
\end{align*}
The first inequality follows from the operator inequality $\Pi_{\hat{\rho
}_{x^{n}},\delta}\ \hat{\rho}_{x^{n}}\ \Pi_{\hat{\rho}_{x^{n}},\delta}\leq
\hat{\rho}_{x^{n}}$ (the projectors $\Pi_{\hat{\rho}_{x^{n}},\delta}$ are
defined with respect to the eigenbasis of $\hat{\rho}_{x^{n}}$). We can then
conclude that these operators form a sub-POVM\ because%
\[
\sqrt{\omega}\ \sum_{x^{n}\in\mathcal{X}^{n}}\Gamma_{x^{n}}^{\left(  m\right)
}\ \sqrt{\omega}\ \leq\ \omega\ \ \ \Longrightarrow\ \ \ \sum_{x^{n}%
\in\mathcal{X}^{n}}\Gamma_{x^{n}}^{\left(  m\right)  }\leq I.
\]
By filling up the rest of the space with some extra operator%
\[
\Gamma_{0}^{\left(  m\right)  }\equiv I-\sum_{x^{n}\in\mathcal{X}^{n}}%
\Gamma_{x^{n}}^{\left(  m\right)  }.
\]
we then have a valid POVM.

Note that we have chosen the measurement operators $\Gamma_{x^{n}}^{\left(
m\right)  }$ so that there is a correspondence between a sequence $x^{n}$ and
a measurement outcome. In the communication paradigm, though, we would like to
have the measurement output some index $l$ that Alice can send over noiseless
classical channels to Bob, so that he can subsequently construct the sequence
$x^{n}\left(  l,m\right)  $ from the value of $l$, the common randomness $m$,
and the codebook $\left\{  x^{n}\left(  l,m\right)  \right\}  $. So, for the
communication paradigm, we can also consider the measurement to be of the form%
\begin{equation}
\Upsilon_{l}^{\left(  m\right)  }\equiv\frac{S}{1+\epsilon}\ \frac
{1}{\left\vert \mathcal{L}\right\vert }\ \omega^{-1/2}\ \xi_{x^{n}\left(
l,m\right)  }\ \omega^{-1/2}. \label{eq:POVM-comp-meas-ops}%
\end{equation}
The POVM\ in (\ref{eq:POVM-def-meas-comp}) then just results by computing
$x^{n}$ from the codeword $x^{n}\left(  l,m\right)  $.

We define the operators $\widetilde{\Lambda}_{x^{n}}$ as%
\[
\widetilde{\Lambda}_{x^{n}}\equiv\frac{1}{\left\vert \mathcal{M}\right\vert
}\sum_{m}\Gamma_{x^{n}}^{\left(  m\right)  },
\]
or equivalently as%
\[
\widetilde{\Lambda}_{x^{n}}=\frac{1}{\left\vert \mathcal{M}\right\vert }%
\sum_{l,m}\mathcal{I}\left(  x^{n}=x^{n}\left(  l,m\right)  \right)
\Upsilon_{l}^{\left(  m\right)  },
\]
where $\mathcal{I}\left(  x^{n}=x^{n}\left(  l,m\right)  \right)  $ is an
indicator function. We now check that the constructed POVM satisfies the
condition in (\ref{eq:faithful-sim-cond-1}) for a faithful simulation:%
\begin{align*}
&  \sum_{x^{n}}\left\Vert \sqrt{\omega}\left(  \Lambda_{x^{n}}-\widetilde
{\Lambda}_{x^{n}}\right)  \sqrt{\omega}\right\Vert _{1}\\
&  =\sum_{x^{n}}\left\Vert p_{X^{n}}\left(  x^{n}\right)  \hat{\rho}_{x^{n}%
}-\frac{S}{1+\epsilon}\ \frac{\left\vert \left\{  l,m:x^{n}\left(  l,m\right)
=x^{n}\right\}  \right\vert }{\left\vert \mathcal{L}\right\vert \left\vert
\mathcal{M}\right\vert }\ \xi_{x^{n}}\right\Vert _{1}\\
&  =\sum_{x^{n}\notin T_{\delta}^{X^{n}}}\left\Vert p_{X^{n}}\left(
x^{n}\right)  \hat{\rho}_{x^{n}}\right\Vert _{1}+\sum_{x^{n}\in T_{\delta
}^{X^{n}}}\left\Vert p_{X^{n}}\left(  x^{n}\right)  \hat{\rho}_{x^{n}}%
-\frac{S}{1+\epsilon}\ \frac{\left\vert \left\{  l,m:x^{n}\left(  l,m\right)
=x^{n}\right\}  \right\vert }{\left\vert \mathcal{L}\right\vert \left\vert
\mathcal{M}\right\vert }\ \xi_{x^{n}}\right\Vert _{1}\\
&  \leq\epsilon+\sum_{x^{n}\in T_{\delta}^{X^{n}}}\left\Vert p_{X^{n}}\left(
x^{n}\right)  \hat{\rho}_{x^{n}}-p_{X^{n}}\left(  x^{n}\right)  \xi_{x^{n}%
}+p_{X^{n}}\left(  x^{n}\right)  \xi_{x^{n}}-\frac{S}{1+\epsilon}%
\ \frac{\left\vert \left\{  l,m:x^{n}\left(  l,m\right)  =x^{n}\right\}
\right\vert }{\left\vert \mathcal{L}\right\vert \left\vert \mathcal{M}%
\right\vert }\ \xi_{x^{n}}\right\Vert _{1}.
\end{align*}
The third equality above follows because the operators $\xi_{x^{n}}$ are
defined to be zero when $x^{n}\notin T_{\delta}^{X^{n}}$. Then,
the
bound in the last line follows from typicality
($\Pr\left\{  X^{n}\notin T_{\delta}^{X^{n}}\right\}  \leq\epsilon$). Continuing, we upper bound as%
\begin{align*}
&  \leq\epsilon+\sum_{x^{n}\in T_{\delta}^{X^{n}}}p_{X^{n}}\left(
x^{n}\right)  \left\Vert \hat{\rho}_{x^{n}}-\xi_{x^{n}}\right\Vert _{1}%
+\sum_{x^{n}\in T_{\delta}^{X^{n}}}\left\Vert p_{X^{n}}\left(  x^{n}\right)
\xi_{x^{n}}-\frac{S}{1+\epsilon}\ \frac{\left\vert \left\{  l,m:x^{n}\left(
l,m\right)  =x^{n}\right\}  \right\vert }{\left\vert \mathcal{L}\right\vert
\left\vert \mathcal{M}\right\vert }\ \xi_{x^{n}}\right\Vert _{1}\\
&  \leq\epsilon+\sum_{x^{n}\in T_{\delta}^{X^{n}}}\frac{p_{X^{n}}\left(
x^{n}\right)  }{S}\left\Vert \hat{\rho}_{x^{n}}-\xi_{x^{n}}\right\Vert
_{1}+\sum_{x^{n}\in T_{\delta}^{X^{n}}}\left\vert \frac{1}{S}p_{X^{n}}\left(
x^{n}\right)  -\frac{1}{1+\epsilon}\ \frac{\left\vert \left\{  l,m:x^{n}%
\left(  l,m\right)  =x^{n}\right\}  \right\vert }{\left\vert \mathcal{L}%
\right\vert \left\vert \mathcal{M}\right\vert }\right\vert \\
&  =\epsilon+\sum_{x^{n}\in T_{\delta}^{X^{n}}}\frac{p_{X^{n}}\left(
x^{n}\right)  }{S}\left\Vert \hat{\rho}_{x^{n}}-\xi_{x^{n}}\right\Vert
_{1}+\left\Vert \hat{P}-\frac{1}{1+\epsilon}\hat{C}\right\Vert _{1}.
\end{align*}
The first inequality is the triangle inequality, and the second inequality
follows by dividing the rightmost two terms by $S$. The equality follows by
invoking the definitions of the operators $\hat{P}$ and $\hat{C}$. We handle
these two remaining terms individually. Consider that%
\begin{align}
\left\Vert \hat{P}-\frac{1}{1+\epsilon}\hat{C}\right\Vert _{1}  &  =\frac
{1}{1+\epsilon}\left\Vert \left(  1+\epsilon\right)  \hat{P}-\hat
{C}\right\Vert _{1}\nonumber\\
&  \leq\frac{1}{1+\epsilon}\left(  \left\Vert \epsilon\hat{P}\right\Vert
_{1}+\left\Vert \hat{P}-\hat{C}\right\Vert _{1}\right) \nonumber\\
&  \leq\frac{2\epsilon}{1+\epsilon}\nonumber\\
&  \leq2\epsilon, \label{eq:bound-var-dist}%
\end{align}
which follows from the triangle inequality, the fact that $\hat{P}$ is a
density operator, and that $\hat{P}$ and $\hat{C}$ satisfy
(\ref{eq:condition-2}). Consider the other term:%
\begin{align}
&  \sum_{x^{n}\in T_{\delta}^{X^{n}}}\frac{p_{X^{n}}\left(  x^{n}\right)  }%
{S}\left\Vert \hat{\rho}_{x^{n}}-\xi_{x^{n}}\right\Vert _{1}\nonumber\\
&  =\sum_{x^{n}\in T_{\delta}^{X^{n}}}p_{X^{\prime n}}\left(  x^{n}\right)
\left\Vert \hat{\rho}_{x^{n}}-\xi_{x^{n}}^{\prime}+\xi_{x^{n}}^{\prime}%
-\xi_{x^{n}}\right\Vert _{1}\nonumber\\
&  \leq\sum_{x^{n}\in T_{\delta}^{X^{n}}}p_{X^{\prime n}}\left(  x^{n}\right)
\left\Vert \hat{\rho}_{x^{n}}-\xi_{x^{n}}^{\prime}\right\Vert _{1}+\sum
_{x^{n}\in T_{\delta}^{X^{n}}}p_{X^{\prime n}}\left(  x^{n}\right)  \left\Vert
\xi_{x^{n}}^{\prime}-\xi_{x^{n}}\right\Vert _{1}\nonumber\\
&  \leq2\sqrt{\epsilon^{\prime}}+2\sqrt{\epsilon^{\prime\prime}},
\label{eq:two-other-final-bounds-meas-comp}%
\end{align}
where we apply the triangle inequality in the third line. For the first bound
with $\epsilon^{\prime}$, we apply the Gentle Operator Lemma
(Lemma~\ref{lem:gentle-operator}) to the condition
in\ (\ref{eq:prime-states-large-trace}), with $\epsilon^{\prime}\equiv
\epsilon+2\sqrt{\epsilon}$. For the second bound with $\epsilon^{\prime\prime
}$, we exploit the equality $\xi_{x^{n}}=\Pi\xi_{x^{n}}^{\prime}\Pi$ and apply
the Gentle Operator Lemma for ensembles (Lemma~\ref{lem:gentle-operator-ens})
to the condition%
\begin{align*}
\sum_{x^{n}\in T_{\delta}^{X^{n}}}p_{X^{\prime n}}\left(  x^{n}\right)
\text{Tr}\left\{  \Pi\xi_{x^{n}}\Pi\right\}   &  =\sum_{x^{n}\in T_{\delta
}^{X^{n}}}p_{X^{\prime n}}\left(  x^{n}\right)  \text{Tr}\left\{  \xi_{x^{n}%
}^{\prime}\right\} \\
&  =\text{Tr}\left\{  \Omega\right\} \\
&  \geq1-\epsilon^{\prime\prime},
\end{align*}
which we proved before in (\ref{eq:Omega-high-trace})\ (with $\epsilon
^{\prime\prime}\equiv2\epsilon+2\sqrt{\epsilon}$). This concludes the proof of
the achievability part of the measurement compression theorem with feedback.

\subsection{Converse theorem for measurement compression}

\label{sec:IC-converse}This section provides a proof of a version of the
converse theorem, which states that the only achievable
rates $R$ and $S$ of classical communication and common randomness
consumption, respectively, are in the rate region given in
Theorem~\ref{thm:meas-comp}. We note that Winter proved a strong version of
the converse theorem \cite{W01}, which states that the error probability
converge exponentially to one as $n$ becomes large. Winter's strong converse
implies that the boundary of the rate region in Theorem~\ref{thm:meas-comp} is
a very sharp dividing line. Here, for the sake of simplicity, we stick to the proof of the ``weak'' converse, which only bounds the error probability away from zero.  The reader can consult  Section~IV\ of
Ref.~\cite{W01}\ for details of Winter's strong converse proof.

The converse theorem states that the
\textquotedblleft single-letter\textquotedblright\ quantities in the rate
region in Theorem~\ref{thm:meas-comp}\ are optimal. A nice consequence is that 
there is no need
to evaluate an intractable regularization of the associated region, as is often the case
for many coding theorems in quantum Shannon theory \cite{W11}. The theorem truly
provides a complete understanding of the measurement compression task from an
information-theoretic perspective.

We now prove the weak converse. Figure~\ref{fig:IC} depicts the most general
protocol for measurement compression with feedback, and it proves to be useful
here to consider a purification of the original input state. The protocol
begins with the reference and Alice possessing the joint system $R^{n}A^{n}$
and Alice sharing the common randomness $M$ with Bob. She then performs a
simulation of the measurement, outputting a random variable $L$ and another
random variable $X^{\prime n}$ that acts as the measurement output on her
side. She sends $L$ to Bob, and Bob produces $\hat{X}^{n}$ from $L$ and the
common randomness $M$. If the protocol is any good for measurement compression
with feedback, then the resulting state $\omega^{R^{n}\hat{X}^{n}X^{\prime n}%
}$ should be $\epsilon$-close in trace distance to the ideal state
$\sigma^{R^{n}X^{n}\overline{X}^{n}}$ (the state resulting from the ideal
protocol in Figure~\ref{fig:ideal-IC}), where $\overline{X}^{n}$ is a copy of
the variable $X^{n}$:%
\begin{equation}
\left\Vert \omega^{R^{n}\hat{X}^{n}X^{\prime n}}-\sigma^{R^{n}X^{n}%
\overline{X}^{n}}\right\Vert _{1}\leq\epsilon.
\label{eq:IC-converse-assumption}%
\end{equation}

We now prove the first lower bound on the classical communication rate $R$:%
\begin{align*}
nR  &  \geq H\left(  L\right) \\
&  \geq I\left(  L;MR^{n}\right) \\
&  =I\left(  LM;R^{n}\right)  +I\left(  L;M\right)  -I\left(  R^{n};M\right)
\\
&  \geq I\left(  LM;R^{n}\right) \\
&  \geq I(\hat{X}^{n};R^{n})_{\omega}\\
&  \geq I(X^{n};R^{n})_{\sigma}-n\epsilon^{\prime}\\
&  =nI\left(  X;R\right)  -n\epsilon^{\prime}.
\end{align*}
The first inequality follows because the entropy of a uniform random variable
is larger than the entropy of any other random variable. The second inequality
follows because $I\left(  L;MR^{n}\right)  =H\left(  L\right)  -H\left(
L|MR^{n}\right)  $ and $H\left(  L|MR^{n}\right)  \geq0$ for a classical
variable $L$. The first equality is an easily verified identity for mutual
information. The third inequality follows because the common randomness $M$ is
not correlated with the reference $R^{n}$ (and hence $I\left(  R^{n};M\right)
=0$) and because $I\left(  L;M\right)  \geq0$. The fourth inequality is from
quantum data processing (Bob processes $L$ and $M$ to get $\hat{X}^{n}$). The
fifth inequality is from (\ref{eq:IC-converse-assumption}) and continuity of
quantum mutual information (the Alicki-Fannes' inequality~\cite{AF04}), where
$\epsilon^{\prime}$ is some function $f(\epsilon)$ such that $\lim
_{\epsilon\rightarrow0}f\left(  \epsilon\right)  =0$. The final equality
follows because the ideal state $\sigma$ is a tensor-power state, and thus the
mutual information $I(X^{n};R^{n})_{\sigma}$ is additive.

A proof for the lower bound on the sum rate $R+S$ goes as follows:%
\begin{align*}
n\left(  R+S\right)   &  \geq H\left(  LM\right) \\
&  \geq I(X^{\prime n};LM)\\
&  \geq I(X^{\prime n};\hat{X}^{n})_{\omega}\\
&  \geq I\left(  \overline{X}^{n}; X^{n}\right)  _{\sigma}-n\epsilon^{\prime
}\\
&  =H\left(  X^{n} \right)  -n\epsilon^{\prime}\\
&  =nH\left(  X\right)  -n\epsilon^{\prime}.
\end{align*}
The first two inequalities follow for the same reasons as the first two above
(we are assuming that copies of $L$ and $M$ are available since they are
classical). The third inequality is quantum data processing. The fourth
inequality follows from (\ref{eq:IC-converse-assumption}) and continuity of
entropy. The first equality follows because the mutual information between a
variable and a copy of it is equal to its entropy. The final equality follows
because the entropy is additive for a tensor power state.

Optimality of the bound $R+S\geq H\left(  X\right)  $ for negative $S$ follows
by considering a protocol whereby Alice uses classical communication alone in
order to simulate the measurement output $X^{n}$ and generate common
randomness $M$ with Bob. The converse in this case proceeds as follows:%
\begin{align*}
nR  &  \geq H\left(  L\right) \\
&  =I\left(  \overline{L};L\right) \\
&  \geq I(X^{\prime n}M^{\prime};\hat{X}^{n}M)\\
&  \geq I(\overline{X}^{n}\overline{M};X^{n}M)-n\epsilon^{\prime}\\
&  =I\left(  \overline{X}^{n};X^{n}\right)  +I\left(  \overline{M};M\right)
-n\epsilon^{\prime}\\
&  =nH\left(  X\right)  +n|S|-n\epsilon^{\prime}.
\end{align*}
The second inequality follows because Bob and Alice have to process $L$ and
its copy $\overline{L}$ in order to recover the approximate $\hat{X}^{n} M$
and $X^{\prime n}M^{\prime}$, respectively. The third inequality follows
because these systems should be close to the ideal ones for a good protocol
(and applying continuity of entropy). The next equalities follow because the
information quantities factor as above for the ideal state.

\subsection{Extension to quantum instruments}

We now briefly review Winter's argument for extending the above protocol from
POVMs to quantum instruments. A quantum instrument is the most general model
for quantum measurement that includes both a classical output and a
post-measurement quantum state \cite{DL70,D76,Ozawa1984}. Our goal is now to
simulate the action of a given quantum instrument on many copies of an input
state $\rho$ using as few resources as possible. The simulation should be such
that Bob possesses the classical output at the end of the protocol (as in the
case of POVM compression), and, as an additional requirement, Alice possesses
the quantum output.

In the present setting, we can conveniently treat a quantum instrument as a
completely positive, trace-preserving (CPTP) map $\mathcal{N}_{\text{instr}}$
of the form
\begin{equation}
\mathcal{N}_{\text{instr}}\left(  \rho\right)  \equiv\sum_{x}\mathcal{N}%
_{x}\left(  \rho\right)  \otimes\left\vert x\right\rangle \left\langle
x\right\vert , \label{eq:q-instr}%
\end{equation}
where each $\mathcal{N}_{x}$ is a completely positive, trace-non-increasing
map of the form%
\[
\mathcal{N}_{x}\left(  \rho\right)  \equiv\sum_{y}N_{x,y}\rho N_{x,y}^{\dag},
\]
such that%
\[
\sum_{y}N_{x,y}^{\dag}N_{x,y}\leq I.
\]
The simulation, implemented by a sequence of maps $\widetilde{\mathcal{N}%
_{\text{instr}}^{n}}$, is defined to be faithful if the following condition holds:

\begin{definition}
[Faithful instrument simulation]A sequence of maps $\widetilde{\mathcal{N}%
_{\emph{instr}}^{n}}$ provides a faithful simulation of the quantum instrument
$\mathcal{N}_{\emph{instr}}$ on the state $\rho$ if for all $\epsilon>0$ and
sufficiently large $n$, the action of the approximation channel on many copies
of a purification $\phi_{\rho}$ of $\rho$ is indistinguishable from the true
quantum instrument, up to a factor of $\epsilon$:%
\begin{equation}
\left\Vert \left(  \text{\emph{id}}\otimes\mathcal{N}_{\operatorname{instr}%
}^{\otimes n}\right)  \left(  \phi_{\rho}^{\otimes n}\right)
-(\text{\emph{id}}\otimes\widetilde{\mathcal{N}_{\operatorname{instr}}^{n}%
})\left(  \phi_{\rho}^{\otimes n}\right)  \right\Vert _{1}\leq\epsilon.
\label{eq:CP-map-simul-condition}%
\end{equation}

\end{definition}

In this case, we have the following theorem:

\begin{theorem}
[Instrument simulation]\label{thm:instrument-simulation}Let $\rho$ be a source
state and $\mathcal{N}_{\operatorname{instr}}$ an instrument to simulate on
this state. A protocol for a faithful feedback simulation of $\mathcal{N}%
_{\operatorname{instr}}$ with classical communication rate $R$
and common randomness rate $S$ exists if and only if%
\begin{align*}
R  &  \geq I\left(  X;R\right)  ,\\
R + S  &  \geq H\left(  X\right)  ,
\end{align*}
where the entropies are with respect to a state of the following form:%
\begin{equation}
\sum_{x}\left\vert x\right\rangle \left\langle x\right\vert ^{X}%
\otimes\text{\emph{Tr}}_{A}\left\{  \left(  I^{R}\otimes\mathcal{N}_{x}%
^{A}\right)  (\phi^{RA})\right\}  ,
\end{equation}
and $\phi^{RA}$ is some purification of the state $\rho$. The simulation is
such that Alice possesses the quantum output of the channel and Bob possesses
the classical output. \label{thm:instr-comp}
\end{theorem}

\begin{proof}
We just prove achievability because the converse theorem from the previous
section applies to this case as well. We start by considering the case in
which every map $\mathcal{N}_{x}$ can be written as
\begin{equation}
\mathcal{N}_{x}\left(  \rho\right)  =N_{x}\rho N_{x}^{\dag}.
\label{eq:single-kraus-instr}%
\end{equation}
(The general case, stated in Section~V-G of Ref.~\cite{DHW08}\ though lacking
a formal proof, will also be addressed.)

We construct an approximation instrument using the protocol in the
achievability proof from Section~\ref{sec:MC-achieve}. Let us set%
\[
\Lambda_{x}=N_{x}^{\dag}N_{x},
\]
and construct the operators $\Gamma_{x^{n}}^{\left(  m\right)  }$ from
$\Lambda_{x}$ and the state $\rho$ as in (\ref{eq:POVM-def-meas-comp}), such
that they satisfy all of the properties that we had before. Define the
distribution $p_{X}\left(  x\right)  $ and the states $\hat{\rho}_{x}$ as we
did before:%
\begin{align*}
p_{X}\left(  x\right)   &  =\text{Tr}\left\{  \Lambda_{x}\rho\right\}  ,\\
\hat{\rho}_{x}  &  =\frac{1}{p_{X}\left(  x\right)  }\sqrt{\rho}\Lambda
_{x}\sqrt{\rho}.
\end{align*}
We construct the approximation instrument $\widetilde{\mathcal{N}%
_{\operatorname{instr}}^{n}}$ from $\Gamma_{x^{n}}^{\left(  m\right)  }$,
$p_{X}\left(  x\right)  $, $\hat{\rho}_{x}$, and the Kraus operators~$N_{x}$.
First, consider that the approximation instrument $\widetilde{\mathcal{N}%
_{\operatorname{instr}}^{n}}$ will be a convex combination of some other
instruments:%
\begin{equation}
\widetilde{\mathcal{N}_{\operatorname{instr}}^{n}}\left(  \sigma\right)
\equiv\frac{1}{\left\vert \mathcal{M}\right\vert }\sum_{m}\mathcal{E}%
_{\operatorname{instr}}^{\left(  m\right)  }\left(  \sigma\right)  ,
\label{eq:approx-channel}%
\end{equation}
where%
\[
\mathcal{E}_{\operatorname{instr}}^{\left(  m\right)  }\left(  \sigma\right)
\equiv\sum_{x^{n}}F_{x^{n}}^{\left(  m\right)  }\sigma F_{x^{n}}^{\left(
m\right)  \dag}\otimes\left\vert x^{n}\right\rangle \left\langle
x^{n}\right\vert ,\ \ \ \ \ \ \ \ \ \ \sum_{x^{n}}F_{x^{n}}^{\left(  m\right)
\dag}F_{x^{n}}^{\left(  m\right)  }\leq I.
\]
We now construct the operators $F_{x^{n}}^{\left(  m\right)  }$. Define the
conditional distribution $p_{\widetilde{X^{n}}|M}\left(  x^{n}|m\right)  $ as
follows:%
\[
p_{\widetilde{X^{n}}|M}\left(  x^{n}|m\right)  \equiv\frac{1}{\left\vert
\mathcal{L}\right\vert }\left\vert \left\{  l:x^{n}\left(  l,m\right)
=x^{n}\right\}  \right\vert ,
\]
as in (\ref{eq:POVM-def-meas-comp}), and let $p_{M}\left(  m\right)
=1/\left\vert \mathcal{M}\right\vert $, so that the marginal distribution
$p_{\widetilde{X^{n}}}\left(  x^{n}\right)  $ is as follows:%
\begin{align*}
p_{\widetilde{X^{n}}}\left(  x^{n}\right)   &  =\frac{1}{\left\vert
\mathcal{M}\right\vert }\sum_{m}p_{\widetilde{X^{n}}|M}\left(  x^{n}|m\right)
\\
&  =\frac{1}{\left\vert \mathcal{L}\right\vert \left\vert \mathcal{M}%
\right\vert }\left\vert \left\{  l,m:x^{n}\left(  l,m\right)  =x^{n}\right\}
\right\vert .
\end{align*}
Take a left polar decomposition of the operator $N_{x}\sqrt{\rho}$
and use it to define the unitary operator $U_{x}$:%
\begin{equation}
N_{x}\sqrt{\rho}=U_{x}\sqrt{\sqrt{\rho}N_{x}^{\dag}N_{x}\sqrt{\rho}}%
=U_{x}\sqrt{p_{X}\left(  x\right)  \hat{\rho}_{x}}.
\label{eq:polar-decomp-kraus}%
\end{equation}
Let $U_{x^{n}}$ be as follows:%
\[
U_{x^{n}}\equiv U_{x_{1}}\otimes\cdots\otimes U_{x_{n}}.
\]
We define the Kraus operators $F_{x^{n}}^{\left(  m\right)  }$ for the
instruments $\mathcal{E}_{\operatorname{instr}}^{\left(  m\right)  }$ as
follows:%
\begin{equation}
F_{x^{n}}^{\left(  m\right)  }\equiv U_{x^{n}}\sqrt{p_{\widetilde{X^{n}}%
|M}\left(  x^{n}|m\right)  \ \frac{S}{1+\epsilon}\ \xi_{x^{n}}}\left(
\omega^{-1/2}\right)  . \label{eq:kraus-approx}%
\end{equation}
One can check that these define completely positive trace-non-increasing
instruments $\mathcal{E}_{\operatorname{instr}}^{\left(  m\right)  }$---this
follows from the fact that the operators $\Gamma_{x^{n}}^{\left(  m\right)  }$
in (\ref{eq:POVM-def-meas-comp}) form a sub-POVM. The instruments
$\mathcal{E}_{\operatorname{instr}}^{\left(  m\right)  }$ in turn form the
instrument $\widetilde{\mathcal{N}^{n}}_{\operatorname{instr}}$ via the
relation in (\ref{eq:approx-channel}).

We can now check that this construction satisfies the condition in
(\ref{eq:CP-map-simul-condition}) for a faithful simulation. Consider a
purification%
\[
\phi_{\rho}^{\otimes n}=\left(  I\otimes\sqrt{\omega}\right)  \left\vert
I\right\rangle \left\langle I\right\vert \left(  I\otimes\sqrt{\omega}\right)
,
\]
where $\left\vert I\right\rangle $ is the vector obtained by \textquotedblleft
flipping the bra\textquotedblright\ of the identity channel $\sum_{x^{n}%
}\left\vert x^{n}\right\rangle \left\langle x^{n}\right\vert $:%
\[
\left\vert I\right\rangle \equiv\sum_{x^{n}}\left\vert x^{n}\right\rangle
\left\vert x^{n}\right\rangle .
\]
We have the following bound on the instrument simulation performance:%
\begin{align*}
&  \left\Vert \left(  \text{id}\otimes\mathcal{N}^{\otimes n}%
_{\operatorname{instr}}\right)  \left(  \phi_{\rho}^{\otimes n}\right)
-(\text{id}\otimes\widetilde{\mathcal{N}^{n}_{\operatorname{instr}} })\left(
\phi_{\rho}^{\otimes n}\right)  \right\Vert _{1}\\
&  =\sum_{x^{n}}\left\Vert \left(  I\otimes N_{x^{n}}\sqrt{\omega}\right)
\left\vert I\right\rangle \left\langle I\right\vert \left(  I\otimes
\sqrt{\omega}N_{x^{n}}^{\dag}\right)  -\frac{1}{\left\vert \mathcal{M}
\right\vert }\sum_{m}\left(  I\otimes F_{x^{n}}^{\left(  m\right)  }
\sqrt{\omega}\right)  \left\vert I\right\rangle \left\langle I\right\vert
\left(  I\otimes\sqrt{\omega}F_{x^{n}}^{\left(  m\right)  \dag}\right)
\right\Vert _{1}\\
&  =\sum_{x^{n}}\left\Vert
\begin{array}
[c]{l}%
p_{X^{n}}\left(  x^{n}\right)  \left(  I\otimes U_{x^{n}}\sqrt{\hat{\rho
}_{x^{n}}}\right)  \left\vert I\right\rangle \left\langle I\right\vert \left(
I\otimes\sqrt{\hat{\rho}_{x^{n}}}U_{x^{n}}^{\dag}\right)  -\\
\ \ \ \ \ \ \ \ \ \ \ \ \ \ \ \ \ \ \ \ p_{\widetilde{X^{n}}}\left(
x^{n}\right)  \ \frac{S}{1+\epsilon}\ \left(  I\otimes U_{x^{n}}\sqrt
{\xi_{x^{n}}}\right)  \left\vert I\right\rangle \left\langle I\right\vert
\left(  I\otimes\sqrt{\xi_{x^{n}}}U_{x^{n}}^{\dag}\right)
\end{array}
\right\Vert _{1}\\
&  =\sum_{x^{n}}\left\Vert p_{X^{n}}\left(  x^{n}\right)  \left(
I\otimes\sqrt{\hat{\rho}_{x^{n}}}\right)  \left\vert I\right\rangle
\left\langle I\right\vert \left(  I\otimes\sqrt{\hat{\rho}_{x^{n}}}\right)
-p_{\widetilde{X^{n}}}\left(  x^{n}\right)  \ \frac{S}{1+\epsilon}\ \left(
I\otimes\sqrt{\xi_{x^{n}}}\right)  \left\vert I\right\rangle \left\langle
I\right\vert \left(  I\otimes\sqrt{\xi_{x^{n}}}\right)  \right\Vert _{1}\\
&  \leq\sum_{x^{n}}\left\Vert p_{X^{n}}\left(  x^{n}\right)  \left(
I\otimes\sqrt{\hat{\rho}_{x^{n}}}\right)  \left\vert I\right\rangle
\left\langle I\right\vert \left(  I\otimes\sqrt{\hat{\rho}_{x^{n}}}\right)
-p_{X^{n}}\left(  x^{n}\right)  \left(  I\otimes\sqrt{\xi_{x^{n}}}\right)
\left\vert I\right\rangle \left\langle I\right\vert \left(  I\otimes\sqrt
{\xi_{x^{n}}}\right)  \right\Vert _{1}\\
&  \ \ \ \ +\sum_{x^{n}}\left\Vert p_{X^{n}}\left(  x^{n}\right)  \left(
I\otimes\sqrt{\xi_{x^{n}}}\right)  \left\vert I\right\rangle \left\langle
I\right\vert \left(  I\otimes\sqrt{\xi_{x^{n}}}\right)  -p_{\widetilde{X^{n}}%
}\left(  x^{n}\right)  \frac{S}{1+\epsilon}\left(  I\otimes\sqrt{\xi_{x^{n}}%
}\right)  \left\vert I\right\rangle \left\langle I\right\vert \left(
I\otimes\sqrt{\xi_{x^{n}}}\right)  \right\Vert _{1}.
\end{align*}
The first equality follows from the block structure of the instruments with
respect to the classical flags $\left\vert x^{n}\right\rangle \left\langle
x^{n} \right\vert $. The second equality follows by substituting the polar
decomposition in (\ref{eq:polar-decomp-kraus}) and the definition in
(\ref{eq:kraus-approx}). The third equality follows from the unitary
invariance of the trace norm. The last inequality is the triangle inequality.
Continuing, we have%
\begin{align*}
&  \leq\sum_{x^{n}}p_{X^{n}}\left(  x^{n}\right)  \left\Vert \left(
I\otimes\sqrt{\hat{\rho}_{x^{n}}}\right)  \left\vert I\right\rangle
\left\langle I\right\vert \left(  I\otimes\sqrt{\hat{\rho}_{x^{n}}}\right)
-\left(  I\otimes\sqrt{\xi_{x^{n}}}\right)  \left\vert I\right\rangle
\left\langle I\right\vert \left(  I\otimes\sqrt{\xi_{x^{n}}}\right)
\right\Vert _{1}\\
&  \ \ \ \ +\sum_{x^{n}}\left\vert p_{X^{n}}\left(  x^{n}\right)
-p_{\widetilde{X^{n}}}\left(  x^{n}\right)  \ \frac{S}{1+\epsilon}\right\vert
\\
&  \leq2\sqrt{2}\sum_{x^{n}}p_{X^{n}}\left(  x^{n}\right)  \sqrt[4]{\left\Vert
\hat{\rho}_{x^{n}}-\xi_{x^{n}}\right\Vert _{1}}+2\epsilon\\
&  \leq2\sqrt{2}\sqrt[4]{\sum_{x^{n}}p_{X^{n}}\left(  x^{n}\right)  \left\Vert
\hat{\rho}_{x^{n}}-\xi_{x^{n}}\right\Vert _{1}}+2\epsilon\\
&  \leq2\sqrt{2}\sqrt[4]{\epsilon+2\sqrt{\epsilon^{\prime}}+2\sqrt
{\epsilon^{\prime\prime}}}+2\epsilon.
\end{align*}
The first inequality follows by factoring out the distribution $p_{X^{n}%
}\left(  x^{n}\right)  $ and because the positive operator $(I\otimes\sqrt
{\xi_{x^{n}}})\left\vert I\right\rangle \left\langle I\right\vert
(I\otimes\sqrt{\xi_{x^{n}}})$ has trace less than one. The second inequality
follows from Winter's Lemma 14:%
\[
\left\Vert \left(  I\otimes\sqrt{\hat{\rho}_{x^{n}}}\right)  \left\vert
I\right\rangle \left\langle I\right\vert \left(  I\otimes\sqrt{\hat{\rho
}_{x^{n}}}\right)  -\left(  I\otimes\sqrt{\xi_{x^{n}}}\right)  \left\vert
I\right\rangle \left\langle I\right\vert \left(  I\otimes\sqrt{\xi_{x^{n}}%
}\right)  \right\Vert _{1}\leq\sqrt[4]{\left\Vert \hat{\rho}_{x^{n}}%
-\xi_{x^{n}}\right\Vert _{1}},
\]
and our previous bound in (\ref{eq:bound-var-dist}). The third inequality is
from concavity of the quartic-root function, and the final one follows from
our previous bounds in (\ref{eq:two-other-final-bounds-meas-comp}) and the
fact that the probability mass of the atypical set is upper bounded by
$\epsilon$.

We are now ready to consider the general case, in which one wants to simulate
an instrument of the form~(\ref{eq:q-instr}). For this purpose, we require a
slightly different coding strategy that combines ideas from
Section~\ref{sec:MC-achieve}\ and the above development.

First, consider that it is possible to implement a quantum instrument of the
form in (\ref{eq:q-instr}) by tracing over an auxiliary register $Y$:%
\[
\mathcal{N}_{\text{instr}}\left(  \rho\right)  =\text{Tr}_{Y}\left\{
\sum_{x,y}N_{x,y}\rho N_{x,y}^{\dag}\otimes\left\vert x\right\rangle
\left\langle x\right\vert ^{X}\otimes\left\vert y\right\rangle \left\langle
y\right\vert ^{Y}\right\}  .
\]
So, Alice and Bob will simulate the following instrument%
\begin{equation}
\sum_{x,y}N_{x,y}\rho N_{x,y}^{\dag}\otimes\left\vert x\right\rangle
\left\langle x\right\vert ^{X}\otimes\left\vert y\right\rangle \left\langle
y\right\vert ^{Y}, \label{eq:instr-sim}%
\end{equation}
in such a way that Bob does not receive the outcome $y$, and thus they
effectively implement the instrument $\mathcal{N}_{\text{instr}}$. The idea is
that they will exploit a code with the following structure:

\begin{enumerate}
\item Alice communicates $nI\left(  X;R\right)  $ bits of classical
communication to Bob (enough for him to reconstruct the $x$ output).

\item Alice keeps $nI\left(  Y;R|X\right)  $ bits of the output to herself.

\item Alice exploits $nH\left(  X|R\right)  $ bits of common randomness shared
with Bob in the simulation.

\item Alice uses $nH\left(  Y|XR\right)  $ bits of local, uniform randomness
(not shared with Bob).
\end{enumerate}

The entropies are with respect to the following classical-quantum state:%
\[
\sum_{x,y}\text{Tr}_{A}\left\{  \left(  I^{R}\otimes\left(  N_{x,y}^{\dag
}N_{x,y}\right)  ^{A}\right)  \left(  \phi_{\rho}^{RA}\right)  \right\}
\otimes\left\vert x\right\rangle \left\langle x\right\vert ^{X}\otimes
\left\vert y\right\rangle \left\langle y\right\vert ^{Y},
\]
and observe that $I\left(  X;R\right)  $ and $H\left(  X|R\right)  $ are
invariant with respect to the choice of Kraus operators $\left\{
N_{x,y}\right\}  $ for each map $\mathcal{N}_{x}$.

More precisely, the measurements used in the simulation are chosen randomly as
in the proof in Section~\ref{sec:MC-achieve}, with the following
modifications. Choose $\left\vert \mathcal{L}_{1}\right\vert \left\vert
\mathcal{M}_{1}\right\vert $ codewords $x^{n}\left(  l_{1},m_{1}\right)  $
independently and randomly according to a pruned version of the distribution%
\[
p_{X}\left(  x\right)  \equiv\text{Tr}\left\{  \mathcal{N}_{x}\left(
\rho\right)  \right\}  =\sum_{y}\text{Tr}\left\{  N_{x,y}^{\dag}N_{x,y}%
\rho\right\}  ,
\]
with%
\begin{align*}
\left\vert \mathcal{L}_{1}\right\vert  &  \approx2^{nI\left(  X;R\right)  },\\
\left\vert \mathcal{M}_{1}\right\vert  &  \approx2^{nH\left(  X|R\right)  }.
\end{align*}
For each pair $\left(  l_{1},m_{1}\right)  $, choose $\left\vert
\mathcal{L}_{2}\right\vert \left\vert \mathcal{M}_{2}\right\vert $ codewords
$y^{n}\left(  l_{2},m_{2}|l_{1},m_{1}\right)  $ independently and randomly
according to a distribution:%
\[
p_{Y^{\prime n}|X^{n}}\left(  y^{n}|x^{n}\left(  l_{1},m_{1}\right)  \right)
,
\]
which is a pruned version of the conditional distribution%
\[
p_{Y|X}\left(  y|x\right)  =\frac{1}{p_{X}\left(  x\right)  }\text{Tr}\left\{
N_{x,y}^{\dag}N_{x,y}\rho\right\}  ,
\]
where%
\begin{align*}
\left\vert \mathcal{L}_{2}\right\vert  &  \approx2^{nI\left(  Y;R|X\right)
},\\
\left\vert \mathcal{M}_{2}\right\vert  &  \approx2^{nH\left(  Y|XR\right)  }.
\end{align*}

After choosing these codewords, we have a codebook $\left\{  x^{n}\left(
l_{1},m_{1}\right)  ,y^{n}\left(  l_{2},m_{2}|l_{1},m_{1}\right)  \right\}  $.
Divide all of these codewords into $\left\vert \mathcal{M}_{1}\right\vert
\left\vert \mathcal{M}_{2}\right\vert $ sets of the form $\left\{
x^{n}\left(  l_{1},m_{1}\right)  ,y^{n}\left(  l_{2},m_{2}|l_{1},m_{1}\right)
\right\}  _{l_{1},l_{2}}$. In order to have a faithful simulation, we require
several conditions analogous to (\ref{eq:event-m}) and (\ref{eq:condition-2})
to hold (except that the first average similar to that in (\ref{eq:event-m})
is over just $l_{1}$ and there is another over both $l_{1}$ and $l_{2}$, and
the other operators like $\hat{C}$\ in (\ref{eq:condition-2}) are with respect
to both $l_{1}$ and $m_{1}$ and all of $l_{1}$, $l_{2}$, $m_{1}$, and $m_{2}%
$). Choosing the sizes of the sets as we do above and applying the Operator
Chernoff Bound several times guarantees that there exists a choice of the
codebook $\left\{  x^{n}\left(  l_{1},m_{1}\right)  ,y^{n}\left(  l_{2}%
,m_{2}|l_{1},m_{1}\right)  \right\}  $ such that these conditions hold. By the
development at the end of Section~\ref{sec:MC-achieve} and the result for
instruments of the special form~(\ref{eq:single-kraus-instr}), it follows that
these conditions lead to a faithful simulation.

The simulation then operates by having the variable $m_{1}$ be common
randomness shared with Bob, $m_{2}$ as additional local, uniform randomness
that Alice uses for picking the measurement, and all of the measurements have
outcomes $l_{1}$ and $l_{2}$. After performing the measurement simulation,
Alice sends the outcome $l_{1}$ to Bob, which he can subsequently use to
reconstruct the codeword $x^{n}\left(  l_{1},m_{1}\right)  $ by combining with
his share $m_{1}$ of the common randomness. The proof as we had it before goes
through---the only difference is in constructing the codebook in such a way
that the sequences $x^{n}$ and $y^{n}$ are separated out. The simulated
instrument has a form like that in (\ref{eq:instr-sim}), and if Alice discards
$y^{n}$, it follows, by
applying the monotonicity of trace distance to the condition in
(\ref{eq:CP-map-simul-condition}),
that Alice and Bob simulate the original instrument.
\end{proof}

As a closing note for this section, we would like to mention that the quantity
$I(X;R)$, appearing in Theorem~\ref{thm:instr-comp}, has a long
history.
Since $I(X;R)$ measures the amount of data \emph{created} by the quantum
measurement, contrarily to the shared randomness that exists before the
measurement itself, it seems natural to consider it as a measure of the
\emph{information gain} produced by the quantum measurement. In this
connection, the 1971 paper of Groenewold was the first to put forward the
problem of measuring the information gain in \textquotedblleft
quantal\textquotedblright\ measurements by means of information-theoretic
quantities~\cite{G71}. Groenewold considered the following quantity
(reformulated according to our notation):
\[
G(\rho,\mathcal{N}_{\operatorname{instr}}):=H(\rho)-\sum_{x}p_{X}(x)H\left\{
\mathcal{N}_{x}(\rho)\,/\,p_{X}(x)\right\}  ,
\]
and he conjectured its positivity for von Neumann-L\"{u}ders measurements. Keep in mind that, at that time, the theory of quantum instruments was
in its infancy, and the von Neumann-L\"{u}ders state reduction postulate,
according to which the initial state is projected onto the eigenspace
corresponding to the observed outcome, was the only model of state reduction
usually considered. Subsequently, Groenewold's conjecture was proved by
Lindblad~\cite{L72}. As the theory of quantum instruments
advanced~\cite{Ozawa1984}, quantum instruments with negative Groenewold's
information gain appeared to be the rule, rather than the exception, until
Ozawa finally settled the problem by proving that $G(\rho,\mathcal{N}%
_{\operatorname{instr}})$ is nonnegative for all states $\rho$ if and only if
the quantum instrument has the special form in~(\ref{eq:single-kraus-instr}%
)~\cite{O86}.

The point is that, for quantum instruments of the form
in~(\ref{eq:single-kraus-instr}), Groenewold's information gain $G(\rho
,\mathcal{N}_{\operatorname{instr}})$ is equal to $I(X;R)$~\cite{BHH08}. This
is a consequence of the fact that, for any matrix $K$, $K^{\dag}K$ and
$KK^{\dag}$ have the same eigenvalues (i.e., the squares of the singular
values of $K$), so that
\[
H\left\{  \mathcal{N}_{x}(\rho) \, / \, p_{X}(x)\right\}  \equiv H(N_{x}\rho
N_{x}^{\dag}\, / \, p_{X}(x) )=H(\sqrt{\rho}N_{x}^{\dag}N_{x}\sqrt{\rho} \,
/\, p_{X}(x)).
\]
This coincidence retroactively strengthens the interpretation of $G(\rho,\mathcal{N}_{\operatorname{instr}})$ as the information
gain due to a quantum measurement, at least in the special case of instruments satisfying (\ref{eq:single-kraus-instr}). In those cases, due to Winter's measurement compression theorem,  $G(\rho,\mathcal{N}_{\operatorname{instr}})$ truly is the rate at which the instrument generates information. More generally, however, $I(X;R)$ is the better measure of information gain both because it is nonnegative and because it \emph{always} has the full strength of Winter's theorem behind it.

\subsubsection{Application to channels}

As already noticed in~\cite{W01}, with Theorem~\ref{thm:instr-comp} at hand,
it is easy to consider the case in which one wants to simulate the action of
some CPTP\ map $\mathcal{N}$\ on many copies of the state $\rho$. The idea is
that, for every Kraus representation \cite{K83} of the map $\mathcal{N}$ as
\begin{equation}
\mathcal{N}\left(  \sigma\right)  \equiv\sum_{x}N_{x}\sigma N_{x}^{\dag},
\label{eq:q-channel-for-sim}%
\end{equation}
where $N_{x}$ are a set of Kraus operators satisfying%
\[
\sum_{x}N_{x}^{\dag}N_{x}=I,
\]
one can apply Theorem~\ref{thm:instr-comp} and simulate the corresponding
quantum instrument
\[
\mathcal{N}_{\operatorname{instr}}\left(  \rho\right)  =\sum_{x} N_{x}\rho
N_{x}^{\dag}\otimes\left\vert x\right\rangle \left\langle x\right\vert .
\]
Then, any protocol faithfully simulating the above instrument automatically
leads, by monotonicity of trace distance, to a faithful simulation of the
channel $\mathcal{N}$, in the sense that it provides a sequence of maps
$\widetilde{\mathcal{N}^{n}}$ such that:
\[
\left\Vert \left(  \text{id}\otimes\mathcal{N}^{\otimes n}\right)  \left(
\phi_{\rho}^{\otimes n}\right)  -(\text{id}\otimes\widetilde{\mathcal{N}^{n}%
})\left(  \phi_{\rho}^{\otimes n}\right)  \right\Vert _{1}\leq\epsilon,
\]
for any $\epsilon> 0$ and sufficiently large $n$.

An important thing to stress is that the rates obtained in this way
\emph{depend} on the particular Kraus representation used to construct the
instrument $\mathcal{N}_{\operatorname{instr}}$. The rates of consumption of
classical resources can hence be minimized over all possible Kraus
representations of a given channel. However, such an optimization turns out to
be difficult in general, as the following example shows.

Let us consider the case of a channel, which can be written as a mixture of
unitaries, i.e.,
\begin{equation}
\mathcal{N}(\rho)=\sum_{x}p(x)U_{x}\rho U_{x}^{\dag}, \label{eq:rand-uni}%
\end{equation}
where $U_{x}^{\dag}U_{x}=I$. Such a channel can be simulated without the need
for classical communication. This follows simply from the fact that the quantum
instrument constructed from~(\ref{eq:rand-uni}) corresponds to measuring the
POVM $\Lambda_{x}=p(x)I$, whose outcomes are completely random and
uncorrelated with the reference, so that $I(X;R)=0$. In fact, the converse is
also true: if a given channel admits a Kraus decomposition for which
$I(X;R)=0$, then its action on the state $\rho$ can be written as a mixture of
unitaries as in~(\ref{eq:rand-uni}) \cite{B06}. In order to show this, suppose
that we find a Kraus decomposition $\mathcal{N}(\rho)=\sum_{x}N_{x}\rho
N_{x}^{\dag}$ such that the quantum mutual information $I(X;R)=0$, where it is
calculated with respect to the following classical-quantum state
\[
\sum_{x}\left\vert x\right\rangle \left\langle x\right\vert ^{X}%
\otimes\operatorname{Tr}_{A}\left\{  (I^{R}\otimes N_{x}^{\dag}N_{x}^{A}%
)\phi^{RA}\right\}  ,
\]
and $\phi^{RA}$ is a purification of $\rho$. Adopting the same notation used
at the beginning of Section~\ref{sec:MC-achieve}, we know that $I(X;R)=0$ if
and only if the sub-normalized states $\theta_{x}^{R}=\sqrt{\rho}(N_{x}^{\dag
}N_{x})^{T}\sqrt{\rho}$ are all proportional to $\rho$. This is possible if
and only if the operators $(N_{x}^{\dag}N_{x})^{T}$ are all proportional to
the identity (on the support of $\rho$), thus proving the claim.

Hence, as the above example shows, to minimize the rate of classical
communication needed to simulate a quantum channel constitutes a task of
complexity comparable to that of deciding whether a given channel possesses a
random-unitary Kraus decomposition or not, for which numerical methods are
known~\cite{AS08} but a general analytical solution has yet to be found.

\section{Non-feedback measurement compression}

\label{sec:MC-non-feedback}We now discuss an extension of Winter's measurement
compression theorem in which the sender is not required to obtain the outcome
of the measurement simulation (known as a \textquotedblleft
non-feedback\textquotedblright\ simulation). Achieving a feedback simulation is more demanding than one without feedback, so we should expect the non-feedback problem to show some reduction in the resources required.  To get a sense of where the improvement comes from will require considering a more general type of POVM decomposition than that in (\ref{eq:POVM-decomposition}). Suppose that it is possible to decompose a
POVM\ $\left\{  \Lambda_{x}\right\}  $ in
terms of a random selection according to a random variable $M$, an
\textquotedblleft internal\textquotedblright\ POVM $\{\Gamma_{w}^{\left(
m\right)  }\}$ with outcomes $w$, and a classical post-processing map
$p_{X|W}\left(  x|w\right)  $ \cite{MDM90,buscemi:082109}:%
\begin{equation}
\Lambda_{x}=\sum_{m,w}p_{M}\left(  m\right)  \Gamma_{w}^{\left(  m\right)
}p_{X|W}\left(  x|w\right)  . \label{eq:non-feedback-decomposition}%
\end{equation}
In that case, Alice and Bob could proceed with a protocol along the
following lines:\ they use Winter's measurement compression protocol to
simulate the POVM\ $\{\sum_{m}p_{M}\left(  m\right)  \Gamma_{w}^{\left(
m\right)  }\}_{w}$ and Bob locally simulates the classical postprocessing map
$p_{X|W}\left(  x|w\right)  $. (This is essentially how a \textquotedblleft
non-feedback\textquotedblright\ simulation will proceed, but there are some
details to be worked out.)

We should compare the performance of the above protocol against one that
exploits a feedback simulation for $\left\{  \Lambda_{x}\right\}  $. The
classical communication cost will increase from $I\left(  X;R\right)  $ to
$I\left(  W;R\right)  $ (the data-processing inequality $I\left(
W;R\right)  \geq I\left(  X;R\right)  $ holds because $W$ is
\textquotedblleft closer\textquotedblright\ to $R$ than is $X$), but the
common randomness cost will be cheaper because the non-feedback protocol
requires only $nI\left(  W;X|R\right)  $ bits of common randomness rather than
$nH\left(  X|R\right)  $ bits (essentially because Bob can find a clever way
to simulate the local map $p_{X|W}\left(  x|w\right)  $). Thus, if the savings
in common randomness consumption are larger than the increase in classical
communication cost, then there is an advantage to performing a non-feedback
simulation. In general, decomposing a POVM\ in many different ways according
to (\ref{eq:non-feedback-decomposition}) leads to a non-trivial curve
characterizing the trade-off between classical communication and common randomness.

In this connection, it is important to remark that the decomposition
(sometimes referred to as a {\it refinement}) of a POVM
according to the post-processing relation:
\begin{equation}
\Lambda_{x}=\sum_{w} \Xi_{w}
p_{X|W}\left(  x|w\right)  \label{eq:non-feedback-decomposition-simpler},
\end{equation}
of which (\ref{eq:non-feedback-decomposition}) is a special case,
is different from the convex decomposition described
in~(\ref{eq:POVM-decomposition}). In particular, while the conditions
for a POVM to be extremal (i.e., not non-trivially decomposable) with
respect to~(\ref{eq:POVM-decomposition}) are known to be rather
involved~\cite{DPP05}, it turns out that POVMs which are extremal with
respect to~(\ref{eq:non-feedback-decomposition-simpler}) can be neatly
characterized as those (and only those) whose elements are all
rank-one operators~\cite{MDM90}. Hence, if the POVM that Alice and Bob
want to simulate is rank-one (i.e., all its elements are rank-one
operators), then there is nothing to gain from implementing a
non-feedback simulation instead of a feedback simulation.
Notice, however, that the two
decompositions~(\ref{eq:POVM-decomposition})
and~(\ref{eq:non-feedback-decomposition-simpler}) are completely independent:
POVMs which are extremal with respect to~(\ref{eq:POVM-decomposition})
need not also be extremal with respect
to~(\ref{eq:non-feedback-decomposition-simpler}), and
vice versa~\cite{buscemi:082109}. This is the reason why there is plenty of room for
non-trivial trade-off relations between classical communication and
common randomness if the POVM to be simulated is not rank-one.

Theorem~\ref{thm:non-feedback-meas-comp}\ below gives a full characterization
of the trade-off for a nonfeedback measurement compression protocol, in the
sense that the protocol summarized above has a matching single-letter converse
proof for its optimality. Thus, we can claim to have a complete understanding
of this task from an information-theoretic perspective.

We should mention that some of the above ideas regarding non-feedback
simulation are already present in prior works \cite{DHW08,C08,BDHSW09}, and
indeed, these works are what led us to pursue a non-feedback measurement
compression protocol. In Ref.~\cite{DHW08}, Devetak \textit{et al}.~observed
in their remarks around Eqs.~(43-45)\ of their paper that a protocol in which
the sender also receives the outcomes of the simulation is optimal, but
\textquotedblleft examples are known in which less randomness is
necessary\textquotedblright\ for protocols that do not have this
restriction. They did not state any explicit examples, however, nor did they
state that there would be a general theorem characterizing the trade-off in
the non-feedback case. Cuff's theorem \cite{C08}\ regarding the trade-off
between classical communication and common randomness for a non-feedback
reverse Shannon theorem is a special case of
Theorem~\ref{thm:non-feedback-meas-comp} below, essentially because a noisy
classical channel is a special case of a quantum measurement and thus the
simulation task is a special case. Ref.~\cite{BDHSW09}\ characterized the
trade-off between quantum communication and entanglement for a non-feedback
simulation of a quantum channel. Thus,
Theorem~\ref{thm:non-feedback-meas-comp}\ below \textquotedblleft sits in
between\textquotedblright\ the communication tasks considered in
Ref.~\cite{C08}\ and Ref.~\cite{BDHSW09}. We should also remark that
Ref.~\cite{BDHSW09}\ stated that it is possible to reduce the common
randomness cost in the non-feedback reverse Shannon theorem either with
randomness recycling or by derandomizing some of it, and we should be able to
employ these approaches in a non-feedback measurement compression protocol.
Though, our approach below is to modify Winter's original protocol directly by
changing the structure of the code.

\subsection{Non-feedback measurement compression theorem}

\begin{theorem}
[Non-feedback measurement compression]\label{thm:non-feedback-meas-comp}Let
$\rho$ be a source state and $\mathcal{N}$ a quantum instrument to simulate on
this state:%
\[
\mathcal{N}\left(  \rho\right)  =\sum_{x}\mathcal{N}_{x}\left(  \rho\right)
\otimes\left\vert x\right\rangle \left\langle x\right\vert ^{X}.
\]
A protocol for faithful non-feedback simulation of the quantum instrument 
with classical communication rate $R$ and common randomness rate
$S$ exists if and only if $R$ and $S$ are in the rate region given by the union of
the following regions:%
\begin{align*}
R  &  \geq I\left(  W;R\right)  ,\\
R+S  &  \geq I\left(  W;XR\right)  ,
\end{align*}
where the entropies are with respect to a state of the following form:%
\begin{equation}
\sum_{x,w}p_{X|W}\left(  x|w\right)  \left\vert w\right\rangle \left\langle
w\right\vert ^{W}\otimes\left\vert x\right\rangle \left\langle x\right\vert
^{X}\otimes\text{\emph{Tr}}_{A}\left\{  \left(  I^{R}\otimes\mathcal{M}%
_{w}^{A}\right)  \left(  \phi_{\rho}^{RA}\right)  \right\}  ,
\label{eq:IC-state-feedback}%
\end{equation}
$\phi_{\rho}^{RA}$ is some purification of the state $\rho$, and the union is
with respect to all decompositions of the original instrument $\mathcal{N}$ of
the form:%
\begin{equation}
\mathcal{N}\left(  \rho\right)  =\sum_{x,w}p_{X|W}\left(  x|w\right)
\mathcal{M}_{w}\left(  \rho\right)  \otimes\left\vert x\right\rangle
\left\langle x\right\vert ^{X}. \label{eq:Markov-decomp-instr}%
\end{equation}
Observe that the systems $R$, $W$, and $X$ in (\ref{eq:IC-state-feedback})
form a quantum Markov chain:$\ R-W-X$.
\end{theorem}

The information quantity $I(W;XR)$ appearing in the above theorem generalizes
Wyner's well-known ``common information'' between dependent random variables \cite{W75}.

\subsection{Achievability for non-feedback measurement compression}

We now prove the achievability part of the above theorem. Suppose for
simplicity that we are just trying to simulate the POVM\ $\Lambda=\left\{
\Lambda_{x}\right\}  $ where each $\Lambda_{x}$ is a positive operator such
that $\Lambda_{x}=\sum_{w}p_{X|W}\left(  x|w\right)  M_{w}$ and each $M_{w}$
is a positive operator. The case for a general quantum instrument follows by
considering this case and by extending it similarly to how we extended
POVM\ compression to instrument compression in
Theorem~\ref{thm:instrument-simulation}. So, the relevant overall
classical-quantum state to consider when building codes for a non-feedback
simulation is as follows:%
\[
\sum_{w,x}p_{X|W}\left(  x|w\right)  \text{Tr}_{A}\left\{  M_{w}^{A}\phi
_{\rho}^{RA}\right\}  \otimes\left\vert w\right\rangle \left\langle
w\right\vert ^{W}\otimes\left\vert x\right\rangle \left\langle x\right\vert
^{X},
\]
which simplifies to%
\[
\sum_{w,x}p_{X|W}\left(  x|w\right)  \ \sqrt{\rho}M_{w}\sqrt{\rho}%
\otimes\left\vert w\right\rangle \left\langle w\right\vert ^{W}\otimes
\left\vert x\right\rangle \left\langle x\right\vert ^{X},
\]
after realizing that Tr$_{A}\left\{  M_{w}^{A}\phi_{\rho}^{RA}\right\}
=\sqrt{\rho}M_{w}\sqrt{\rho}$ (in the above and what follows, we ignore the
transpose in the eigenbasis of $\rho$ on $M_{w}$ because it is irrelevant for
the result). Let $\tau$ denote the state obtained by tracing over the $W$
register of the above state:%
\[
\tau\equiv\sum_{w}\sqrt{\rho}M_{w}\sqrt{\rho}\otimes\sigma_{w}^{X},
\]
where the classical state $\sigma_{w}^{X}$ is as follows:%
\[
\sigma_{w}^{X}\equiv\sum_{x}p_{X|W}\left(  x|w\right)  \ \left\vert
x\right\rangle \left\langle x\right\vert ^{X}.
\]
Consider the following ensemble:%
\begin{align*}
p_{W}\left(  w\right)   &  \equiv\text{Tr}\left\{  M_{w}\rho\right\}  ,\\
\hat{\rho}_{w}  &  \equiv\frac{1}{p_{W}\left(  w\right)  }\sqrt{\rho}%
M_{w}\sqrt{\rho}.
\end{align*}
Observe that $\rho$ is the expected state of this ensemble:%
\[
\sum_{w}p_{W}\left(  w\right)  \hat{\rho}_{w}=\sum_{w}\sqrt{\rho}M_{w}%
\sqrt{\rho}=\rho.
\]
Also, the state $\tau$ is as follows:%
\[
\tau=\sum_{w}p_{W}\left(  w\right)  \hat{\rho}_{w}\otimes\sigma_{w}^{X}.
\]

Our approach is similar to Winter's approach detailed in
Section~\ref{sec:MC-achieve}: choose $\left\vert \mathcal{L}\right\vert
\left\vert \mathcal{M}\right\vert $ codewords $w^{n}\left(  l,m\right)
$\ according to the pruned version of the distribution $p_{W^{n}}\left(
w^{n}\right)  $. As long as%
\begin{align*}
\left\vert \mathcal{L}\right\vert  &  \approx2^{nI\left(  W;R\right)  },\\
\left\vert \mathcal{M}\right\vert  &  \approx2^{nI\left(  W;X|R\right)  },\\
\left\vert \mathcal{L}\right\vert \left\vert \mathcal{M}\right\vert  &
\approx2^{nI\left(  W;XR\right)  },
\end{align*}
the operator Chernoff bound (Lemma~\ref{lemma:operator-chernoff}) guarantees
that there exists a choice of the codewords $w^{n}\left(  l,m\right)  $\ such
that the following conditions are true:%
\begin{align}
\frac{1}{\left\vert \mathcal{L}\right\vert }\sum_{l}\hat{\rho}_{w^{n}\left(
l,m\right)  }^{\prime}  &  \in\left[  \left(  1\pm\epsilon\right)
\rho^{\prime n}\right]  ,\label{eq:l-avg-Chernoff}\\
\frac{1}{\left\vert \mathcal{L}\right\vert \left\vert \mathcal{M}\right\vert
}\sum_{l,m}\kappa_{w^{n}\left(  l,m\right)  }  &  \in\left[  \left(
1\pm\epsilon\right)  \tau^{\prime n}\right]  , \label{eq:lm-avg-Chernoff}%
\end{align}
where each $\hat{\rho}_{w^{n}}^{\prime}$ is a typical projected version of
$\hat{\rho}_{w^{n}}\equiv\hat{\rho}_{w_{1}}\otimes\cdots\otimes\hat{\rho
}_{w_{n}}$:%
\[
\hat{\rho}_{w^{n}}^{\prime}\equiv\Pi\ \Pi_{\rho,\delta}^{n}\ \Pi_{\hat{\rho
}_{w^{n}},\delta}\ \hat{\rho}_{w^{n}}\ \Pi_{\hat{\rho}_{w^{n}},\delta}%
\ \Pi_{\rho,\delta}^{n}\ \Pi.
\]
In the above, $\Pi_{\hat{\rho}_{w^{n}},\delta}$ is the conditionally typical
projector for $\hat{\rho}_{w^{n}}$, $\Pi_{\rho,\delta}^{n}$ is the average
typical projector for $\rho^{\otimes n}$, and $\Pi$ is the eigenvalue cutoff
projector as before. We define each $\kappa_{w^{n}\left(  l,m\right)  }$ as a
typical projected version of the state $\hat{\rho}_{w^{n}\left(  l,m\right)
}\otimes\sigma_{w^{n}\left(  l,m\right)  }$:%
\[
\kappa_{w^{n}}\equiv\Pi^{\prime}\ \Pi_{\tau,\delta}^{n}\ \left(  \Pi
_{\hat{\rho}_{w^{n}},\delta}\otimes\Pi_{\sigma_{w^{n}},\delta}\right)
\ \hat{\rho}_{w^{n}}\otimes\sigma_{w^{n}}\ \left(  \Pi_{\hat{\rho}_{w^{n}%
},\delta}\otimes\Pi_{\sigma_{w^{n}},\delta}\right)  \ \Pi_{\tau,\delta}%
^{n}\ \Pi^{\prime}.
\]
In the above, $\Pi_{\hat{\rho}_{w^{n}},\delta}\otimes\Pi_{\sigma_{w^{n}%
},\delta}$ is the conditionally typical projector for $\hat{\rho}_{w^{n}%
}\otimes\sigma_{w^{n}}$, $\Pi_{\tau,\delta}^{n}$ is an average typical
projector for $\tau$, and $\Pi^{\prime}$ is another eigenvalue cutoff
projector. The states $\rho^{\prime n}$ and $\tau^{\prime n}$ are the
expectations of the states $\hat{\rho}_{w^{n}}^{\prime}$ and $\kappa_{w^{n}}$,
respectively, with respect to the pruned version of the distribution
$p_{W^{n}}\left(  w^{n}\right)  $. Recall that the operator Chernoff bound
guarantees with high probability that the sample averages $\frac{1}{\left\vert
\mathcal{L}\right\vert }\sum_{l}\hat{\rho}_{w^{n}\left(  l,m\right)  }%
^{\prime}$ and $\frac{1}{\left\vert \mathcal{L}\right\vert \left\vert
\mathcal{M}\right\vert }\sum_{l,m}\kappa_{w^{n}\left(  l,m\right)  }$ are
within $\epsilon$ (in the operator interval sense) of their true expectations
$\rho^{\prime n}$ and $\tau^{\prime n}$, respectively.\ Thus there exist
particular values of the $w^{n}\left(  l,m\right)  $ such that the above
conditions are all true. We can use the condition in (\ref{eq:l-avg-Chernoff})
to guarantee that the following defines a legitimate POVM\ (just as in
Winter's approach in Section~\ref{sec:MC-achieve}):%
\[
\Upsilon_{l}^{\left(  m\right)  }\equiv\frac{S}{1+\epsilon}\frac{1}{\left\vert
\mathcal{L}\right\vert }\omega^{-1/2}\,\hat{\rho}_{w^{n}\left(  l,m\right)
}^{\prime}\,\omega^{-1/2},
\]
where $S$ is the mass of the typical set corresponding to the distribution
$p_{W^{n}}(w^{n})$ and $\omega=\rho^{\otimes n}$. Also, observe that the
following states are close in trace distance for sufficiently large $n$, due
to quantum typicality and the Gentle Operator Lemma:%
\begin{align}
\left\Vert \hat{\rho}_{w^{n}}^{\prime}-\hat{\rho}_{w^{n}}\right\Vert _{1}  &
\leq f_{1}\left(  \epsilon\right)  ,\label{eq:nonfeedback-triangle-1}\\
\left\Vert \kappa_{w^{n}}-\hat{\rho}_{w^{n}}\otimes\sigma_{w^{n}}\right\Vert
_{1}  &  \leq f_{2}\left(  \epsilon\right)  .
\label{eq:nonfeedback-triangle-2}%
\end{align}
Here and in what follows, $f_{i}\left(  \epsilon\right)  $ is some polynomial
in $\epsilon$ so that $\lim_{\epsilon\rightarrow0}f_{i}\left(  \epsilon
\right)  =0$.

We use the conditions in (\ref{eq:l-avg-Chernoff}) and
(\ref{eq:lm-avg-Chernoff}) to guarantee that the simulation is faithful. The
protocol proceeds as follows: Alice and Bob use the common randomness $M$ to
select a POVM\ $\Upsilon^{\left(  m\right)  }$. Alice performs a measurement
and gets the outcome $l$ (corresponding to the operator $\Upsilon_{l}^{\left(
m\right)  }$). She sends the index $l$ to Bob, who then prepares the classical
state $\sigma_{w^{n}\left(  l,m\right)  }$ based on the common randomness~$m$
and the measurement outcome $l$. Consider the following chain of equalities:%
\begin{align*}
&  \sum_{x^{n}}\text{Tr}_{A^{n}}\left\{  \left(  \text{id}\otimes
\Lambda_{x^{n}}\right)  \left(  \phi_{\rho}^{\otimes n}\right)  \right\}
\otimes\left\vert x^{n}\right\rangle \left\langle x^{n}\right\vert \\
&  =\sum_{w^{n},x^{n}}\text{Tr}_{A^{n}}\left\{  \left(  \text{id}\otimes
M_{w^{n}}\right)  \left(  \phi_{\rho}^{\otimes n}\right)  \right\}  \otimes
p_{X^{n}|W^{n}}\left(  x^{n}|w^{n}\right)  \left\vert x^{n}\right\rangle
\left\langle x^{n}\right\vert \\
&  =\sum_{w^{n},x^{n}}\sqrt{\omega}M_{w^{n}}\sqrt{\omega}\otimes
p_{X^{n}|W^{n}}\left(  x^{n}|w^{n}\right)  \left\vert x^{n}\right\rangle
\left\langle x^{n}\right\vert \\
&  =\tau^{\otimes n},
\end{align*}%
\[
\sum_{m,l}\frac{1}{\left\vert \mathcal{M}\right\vert }\text{Tr}_{A^{n}%
}\left\{  \left(  \text{id}\otimes\Upsilon_{l}^{\left(  m\right)  }\right)
\left(  \phi_{\rho}^{\otimes n}\right)  \right\}  \otimes\sigma_{w^{n}\left(
l,m\right)  }=\frac{1}{\left\vert \mathcal{M}\right\vert \left\vert
\mathcal{L}\right\vert }\sum_{m,l}\frac{S}{1+\epsilon}\hat{\rho}_{w^{n}\left(
l,m\right)  }^{\prime}\otimes\sigma_{w^{n}\left(  l,m\right)  }.
\]
By exploiting (\ref{eq:lm-avg-Chernoff}), that $\left\Vert \tau^{\otimes
n}-\tau^{\prime n}\right\Vert _{1}\leq f_{3}\left(  \epsilon\right)  $, and
that $\left\Vert \rho^{\otimes n}-\rho^{\prime n}\right\Vert _{1}\leq
f_{4}\left(  \epsilon\right)  $, we have that%
\[
\left\Vert \tau^{\otimes n}-\frac{1}{\left\vert \mathcal{M}\right\vert
\left\vert \mathcal{L}\right\vert }\sum_{l,m}\kappa_{w^{n}\left(  l,m\right)
}\right\Vert _{1}\leq f_{5}\left(  \epsilon\right)  .
\]
Also, we have that%
\[
\left\Vert \frac{1}{\left\vert \mathcal{M}\right\vert \left\vert
\mathcal{L}\right\vert }\sum_{m,l}\frac{S}{1+\epsilon}\hat{\rho}_{w^{n}\left(
l,m\right)  }^{\prime}\otimes\sigma_{w^{n}\left(  l,m\right)  }-\frac
{1}{\left\vert \mathcal{M}\right\vert \left\vert \mathcal{L}\right\vert }%
\sum_{m,l}\hat{\rho}_{w^{n}\left(  l,m\right)  }^{\prime}\otimes\sigma
_{w^{n}\left(  l,m\right)  }\right\Vert _{1}\leq f_{6}\left(  \epsilon\right)
.
\]
From (\ref{eq:nonfeedback-triangle-1}) and (\ref{eq:nonfeedback-triangle-2})
and convexity of trace distance, we have that%
\[
\left\Vert \frac{1}{\left\vert \mathcal{M}\right\vert \left\vert
\mathcal{L}\right\vert }\sum_{m,l}\kappa_{w^{n}\left(  l,m\right)  }-\frac
{1}{\left\vert \mathcal{M}\right\vert \left\vert \mathcal{L}\right\vert }%
\sum_{m,l}\hat{\rho}_{w^{n}\left(  l,m\right)  }^{\prime}\otimes\sigma
_{w^{n}\left(  l,m\right)  }\right\Vert _{1}\leq f_{7}\left(  \epsilon\right)
.
\]
Putting all of these together with the triangle inequality gives an upper
bound on the trace distance between the ideal output of the measurement and
the state resulting from the simulation:%
\[
\left\Vert \sum_{x^{n}}\text{Tr}_{A^{n}}\left\{  \left(  \text{id}%
\otimes\Lambda_{x^{n}}\right)  \left(  \phi_{\rho}^{\otimes n}\right)
\right\}  \otimes\left\vert x^{n}\right\rangle \left\langle x^{n}\right\vert
-\sum_{m,l}\frac{1}{\left\vert \mathcal{M}\right\vert }\text{Tr}_{A^{n}%
}\left\{  \left(  \text{id}\otimes\Upsilon_{l}^{\left(  m\right)  }\right)
\left(  \phi_{\rho}^{\otimes n}\right)  \right\}  \otimes\sigma_{w^{n}\left(
l,m\right)  }\right\Vert _{1}\leq f_{8}\left(  \epsilon\right)  .
\]
The case for a general quantum instrument follows by similar reasoning as that
in the proof of Theorem~\ref{thm:instrument-simulation}.

\subsection{Converse for non-feedback measurement compression}

We now prove the converse part of Theorem~\ref{thm:non-feedback-meas-comp}.
Our proof is similar to Cuff's converse proof for the non-feedback version of
the classical reverse Shannon theorem~\cite{C08}. Figure~\ref{fig:IC} can
serve as a depiction of the most general protocol for a non-feedback
simulation, if we ignore the decoding on Alice's side to produce $X^{\prime
n}$. The non-feedback protocol begins with Alice and the reference sharing
many copies of the state $\phi_{\rho}^{RA}$ and Alice sharing common
randomness $M$ with Bob. She then chooses a quantum instrument $\Upsilon
^{\left(  m\right)  }$ based on the common randomness $M$ and performs it on
her systems $A^{n}$. The measurement returns outcome $L$, and the overall
state is as follows:%
\[
\sum_{l,m}\frac{1}{\left\vert \mathcal{M}\right\vert }(\Upsilon_{l}^{\left(
m\right)  })^{A^{n}}(\left(  \phi_{\rho}^{RA}\right)  ^{\otimes n}%
)\otimes\left\vert l\right\rangle \left\langle l\right\vert ^{L}%
\otimes\left\vert m\right\rangle \left\langle m\right\vert ^{M},
\]
where $\Upsilon_{l}^{\left(  m\right)  }$ is a completely positive, trace
non-increasing map. Alice sends the register $L$ to Bob. Based on $L$ and $M$,
he performs some stochastic map $p_{\hat{X}^{n}|L,M}\left(  \hat{x}%
^{n}|l,m\right)  $ to give his estimate $\hat{x}^{n}$ of the measurement
outcome. The resulting state is as follows:%
\[
\omega^{R^{n}A^{n}LM\hat{X}^{n}}\equiv\sum_{l,m,\hat{x}^{n}}\frac
{1}{\left\vert \mathcal{M}\right\vert }p_{\hat{X}^{n}|L,M}\left(  \hat{x}%
^{n}|l,m\right)  (\Upsilon_{l}^{\left(  m\right)  })^{A^{n}}(\left(
\phi_{\rho}^{RA}\right)  ^{\otimes n})\otimes\left\vert l\right\rangle
\left\langle l\right\vert ^{L}\otimes\left\vert m\right\rangle \left\langle
m\right\vert ^{M}\otimes\left\vert \hat{x}^{n}\right\rangle \left\langle
\hat{x}^{n}\right\vert ^{\hat{X}^{n}}.
\]
The following condition should hold for all $\epsilon>0$ and sufficiently
large $n$ for a faithful non-feedback simulation:%
\[
\left\Vert \omega^{R^{n}\hat{X}^{n}}-\sum_{x^{n}}\text{Tr}_{A^{n}}\left\{
\left(  I\otimes\mathcal{N}_{x^{n}}\right)  \left(  \phi_{\rho}^{RA}\right)
^{\otimes n}\right\}  \otimes\left\vert x^{n}\right\rangle \left\langle
x^{n}\right\vert ^{X^{n}}\right\Vert _{1}\leq\epsilon
\]

We prove the first bound as follows:%
\begin{align*}
nR  &  \geq H\left(  L\right)  _{\omega}\\
&  \geq I\left(  L;MR^{n}\right)  _{\omega}\\
&  =I\left(  LM;R^{n}\right)  _{\omega}+I\left(  L;M\right)  _{\omega
}-I\left(  M;R^{n}\right)  _{\omega}\\
&  \geq I\left(  LM;R^{n}\right)  _{\omega}\\
&  =H\left(  R^{n}\right)  _{\omega}-H\left(  R^{n}|LM\right)  _{\omega}\\
&  \geq\sum_{k}\left[  H\left(  R_{k}\right)  _{\omega}-H\left(
R_{k}|LM\right)  _{\omega}\right] \\
&  =\sum_{k}I\left(  LM;R_{k}\right)  _{\omega}\\
&  =nI\left(  LM;R|K\right)  _{\sigma}\\
&  \geq nI\left(  LM;R|K\right)  _{\sigma}+nI\left(  R;K\right)  _{\sigma
}-n\epsilon^{\prime}\\
&  =nI\left(  LMK;R\right)  _{\sigma}-n\epsilon^{\prime}.
\end{align*}
The first two inequalities follow for reasons similar to the first few steps
of our previous converse. The first equality is an identity for quantum mutual
information. The third inequality follows because there are no correlations
between $R^{n}$ and $M$ so that $I\left(  M;R^{n}\right)  _{\omega}=0$. The
second equality is an identity for quantum mutual information. The fourth
inequality follows from subadditivity of quantum entropy:%
\[
H\left(  R^{n}|LM\right)  _{\omega}\leq\sum_{k}H\left(  R_{k}|LM\right)
_{\omega},
\]
and because the state on $R^{n}$ is a tensor-power state so that%
\[
H\left(  R^{n}\right)  _{\omega}=\sum_{k}H\left(  R_{k}\right)  _{\omega}.
\]
The third equality is another identity. The fourth equality comes about by
defining the state $\sigma$ as follows:%
\begin{multline}
\sigma^{RALM\hat{X}K}\equiv\sum_{l,m,k,\hat{x}}\frac{1}{n\left\vert
\mathcal{M}\right\vert }p_{\hat{X}_{k}|L,M}\left(  \hat{x}|l,m\right)
\text{Tr}_{R_{1}^{k-1}R_{k+1}^{n}A_{1}^{k-1}A_{k+1}^{n}}\left\{  (\Upsilon
_{l}^{\left(  m\right)  })^{A^{n}}(\left(  \phi_{\rho}^{RA}\right)  ^{\otimes
n})\right\}  \otimes\\
\left\vert l\right\rangle \left\langle l\right\vert ^{L}\otimes\left\vert
m\right\rangle \left\langle m\right\vert ^{M}\otimes\left\vert \hat
{x}\right\rangle \left\langle \hat{x}\right\vert ^{\hat{X}}\otimes\left\vert
k\right\rangle \left\langle k\right\vert ^{K}, \label{eq:helper-state-no-fdbk}%
\end{multline}
and exploiting the fact that $K$ is a uniform classical random variable, with
distribution $1/n$, determining which systems $R_{k}A_{k}\hat{X}_{k}$ to
select. (The notation $\text{Tr}_{R_{i}^{j}}$ with $i\leq j$ indicates to
trace over systems $R_{i}\cdots R_{j}$.) From the fact that the measurement
simulation is faithful, we can apply the Alicki-Fannes' inequality to conclude
that%
\begin{equation}
I\left(  R\hat{X};K\right)  _{\sigma}=\left\vert I\left(  R\hat{X};K\right)
_{\sigma}-I\left(  RX;K\right)  _{\tau}\right\vert \leq\epsilon^{\prime},
\label{eq:no-fdbk-K-info-bnd}%
\end{equation}
where $\tau$ is a state like $\sigma$ but resulting from the tensor-power
state for ideal measurement compression (and due to its IID\ structure, it has
no correlations with any particular system $k$ so that $I\left(  RX;K\right)
_{\tau}=0$). The above also implies that%
\[
I\left(  R;K\right)  _{\sigma}\leq\epsilon^{\prime},
\]
by quantum data processing. The final equality is an application of the chain
rule for quantum mutual information. The state $\sigma$ for the final
information term has the form in (\ref{eq:Markov-decomp-instr}) with $LMK=W$,
the distribution $p_{X|W}\left(  x|w\right)  $ as%
\[
p_{\hat{X}_{k}|L,M}\left(  \hat{x}|l,m\right)  ,
\]
and the completely positive, trace non-increasing maps $\mathcal{M}_{w}$
defined by%
\[
\varrho^{A}\rightarrow\frac{1}{n\left\vert \mathcal{M}\right\vert }%
\text{Tr}_{A_{1}^{k-1}A_{k+1}^{n}}\left\{  (\Upsilon_{l}^{\left(  m\right)
})^{A^{n}}(\left(  \phi_{\rho}^{A}\right)  ^{\otimes k-1}\otimes\varrho
^{A}\otimes\left(  \phi_{\rho}^{A}\right)  ^{\otimes n-k})\right\}  .
\]
Also, observe that $R-\left(  LMK\right)  -\hat{X}$ forms a quantum Markov chain.

We now prove the second bound:%
\begin{align*}
n\left(  R+S\right)   &  \geq H\left(  LM\right)  _{\omega}\\
&  \geq I(LM;\hat{X}^{n}R^{n})_{\omega}\\
&  =H(\hat{X}^{n}R^{n})_{\omega}-H(\hat{X}^{n}R^{n}|LM)_{\omega}\\
&  \geq\sum_{k}\left[  H(\hat{X}_{k}R_{k})_{\omega}-H(\hat{X}_{k}%
R_{k}|LM)_{\omega}\right]  -n\epsilon^{\prime}\\
&  =\sum_{k}I(LM;\hat{X}_{k}R_{k})_{\omega}-n\epsilon^{\prime}\\
&  =nI(LM;\hat{X}R|K)_{\sigma}-n\epsilon^{\prime}\\
&  \geq nI(LM;\hat{X}R|K)_{\sigma}+nI(K;\hat{X}R)_{\sigma}-n2\epsilon^{\prime
}\\
&  =nI(LMK;\hat{X}R)_{\sigma}-n2\epsilon^{\prime}.
\end{align*}
The first two inequalities follow from similar reasons as our previous
inequalities. The first equality is an identity. The third inequality follows
from subadditivity of entropy:%
\[
H(\hat{X}^{n}R^{n}|LM)_{\omega}\leq\sum_{k}H(\hat{X}_{k}R_{k}|LM)_{\omega},
\]
and from the fact that the measurement simulation is faithful so that%
\[
\left\vert H(\hat{X}^{n}R^{n})_{\omega}-\sum_{k}H(\hat{X}_{k}R_{k})_{\omega
}\right\vert \leq n\epsilon^{\prime},
\]
where we have applied Lemma~\ref{lem:entropy-for-close-IID} below. The second
equality is an identity. The third equality follows by considering the state
$\sigma$ as defined in (\ref{eq:helper-state-no-fdbk}). The fourth inequality
follows from (\ref{eq:no-fdbk-K-info-bnd}). The final equality is the chain
rule for quantum mutual information. Similarly as stated above, the state
$\sigma$ has the form in (\ref{eq:Markov-decomp-instr}).

\begin{lemma}
\label{lem:entropy-for-close-IID}Suppose that a state $\rho^{A^{n}}$ is
$\epsilon$-close in trace distance to an IID\ state $\left(  \sigma
^{A}\right)  ^{\otimes n}$:%
\begin{equation}
\left\Vert \rho^{A^{n}}-\left(  \sigma^{A}\right)  ^{\otimes n}\right\Vert
_{1}\leq\epsilon. \label{eq:close-to-IID}%
\end{equation}
Then the entropy $H\left(  A^{n}\right)  _{\rho}$ of $\rho^{A^{n}}$ and the
entropy $\sum_{k}H\left(  A_{k}\right)  _{\rho}$ are close in the following
sense:%
\[
\left\vert H\left(  A^{n}\right)  _{\rho}-\sum_{k}H\left(  A_{k}\right)
_{\rho}\right\vert \leq2n\epsilon\log\left\vert A\right\vert +\left(
n+1\right)  H_{2}\left(  \epsilon\right)  .
\]

\end{lemma}

\begin{proof}
Apply the Fannes-Audenart inequality to\ (\ref{eq:close-to-IID})
to obtain%
\begin{equation}
\left\vert H\left(  A^{n}\right)  _{\rho}-H\left(  A^{n}\right)
_{\sigma^{\otimes n}}\right\vert \leq\epsilon n\log\left\vert A\right\vert
+H_{2}\left(  \epsilon\right)  . \label{eq:entropy-close-IID-1}%
\end{equation}
The following inequality also follows by applying monotonicity of trace
distance to (\ref{eq:close-to-IID}):%
\[
\left\Vert \rho^{A_{k}}-\sigma^{A}\right\Vert _{1}\leq\epsilon,
\]
which then gives that%
\[
\left\vert H\left(  A_{k}\right)  _{\rho}-H\left(  A\right)  _{\sigma
}\right\vert \leq\epsilon\log\left\vert A\right\vert +H_{2}\left(
\epsilon\right)  ,
\]
by again applying the Fannes-Audenart inequality. Summing these over all $k$
then gives that%
\begin{equation}
\sum_{k=1}^{n}\left\vert H\left(  A\right)  _{\sigma}-H\left(  A_{k}\right)
_{\rho}\right\vert \leq n\epsilon\log\left\vert A\right\vert +nH_{2}\left(
\epsilon\right)  . \label{eq:entropy-close-IID-2}%
\end{equation}
Applying the triangle inequality to (\ref{eq:entropy-close-IID-1}) and
(\ref{eq:entropy-close-IID-2}) gives the desired result:%
\begin{align*}
2n\epsilon\log\left\vert A\right\vert +\left(  n+1\right)  H_{2}\left(
\epsilon\right)   &  \geq\left\vert H\left(  A^{n}\right)  _{\rho}-H\left(
A^{n}\right)  _{\sigma^{\otimes n}}\right\vert +\sum_{k=1}^{n}\left\vert
H\left(  A\right)  _{\sigma}-H\left(  A_{k}\right)  _{\rho}\right\vert \\
&  \geq\left\vert H\left(  A^{n}\right)  _{\rho}-H\left(  A^{n}\right)
_{\sigma^{\otimes n}}+\sum_{k=1}^{n}\left[  H\left(  A\right)  _{\sigma
}-H\left(  A_{k}\right)  _{\rho}\right]  \right\vert \\
&  =\left\vert H\left(  A^{n}\right)  _{\rho}-H\left(  A^{n}\right)
_{\sigma^{\otimes n}}+H\left(  A^{n}\right)  _{\sigma^{\otimes n}}-\sum
_{k=1}^{n}H\left(  A_{k}\right)  _{\rho}\right\vert \\
&  =\left\vert H\left(  A^{n}\right)  _{\rho}-\sum_{k=1}^{n}H\left(
A_{k}\right)  _{\rho}\right\vert .
\end{align*}

\end{proof}

\section{Classical data compression with quantum side information}

We now turn to the third protocol of this paper:\ classical data compression
with quantum side information. We discuss this protocol in detail because it
is a step along the way to constructing our protocol for measurement
compression with quantum side information (and we have a particular way that
we construct this latter protocol). Devetak and Winter first proved
achievability and optimality of a protocol for this task \cite{DW02}. They
proved this result by appealing to\ Winter's proof of the classical capacity
theorem \cite{itit1999winter}\ and a standard recursive code construction
argument of Csisz\`{a}r and K\"{o}rner \cite{CK81}. Renes \textit{et
al}.~later gave a proof of this protocol by exploiting two-universal hash
functions and a square-root measurement \cite{RB08,RR12} (the first paper
proved the IID\ version and the latter the \textquotedblleft
one-shot\textquotedblright\ case). Renes \textit{et al}.~further explored a
connection between this protocol and privacy amplification\ by considering
entropic uncertainty relations \cite{Renes08062011,RR11}.

Our development here contains a review of this information processing task and
the statement of the theorem, in addition to providing novel proofs of both the achievability
part and the converse that are direct quantum generalizations of the
well-known approaches in Refs.~\cite{CT91,el2010lecture} for the Slepian-Wolf
problem \cite{SW73}. The encoder in our achievability proof bears some
similarities with those in Refs.~\cite{CT91,el2010lecture,RR12}---the protocol
has the sender first hash the received sequence and send the hash along to the
receiver. The receiver then employs a sequential quantum decoder---he searches
sequentially among all the possible quantum states that are consistent with
the hash in order to determine the sequence emitted by the source. The main
tool employed in the error analysis is Sen's non-commutative union bound
\cite{S11}. A potential advantage of a sequential decoding approach is that it
might lead to physical implementations of these protocols for small block sizes, 
along the lines discussed in Ref.~\cite{GTW12}.

\subsection{Information processing task for CDC with QSI}

We now discuss the general information processing task. Consider an ensemble
$\left\{  p_{X}\left(  x\right)  ,\rho_{x}\right\}  $. Suppose that a source issues a 
random sequence $X^{n}$ to Alice, distributed
according to the IID$\ $distribution $p_{X^{n}}\left(  x^{n}\right)  $, while also issuing
the correlated quantum state $\rho_{X^{n}}$ to Bob. Their joint
state is described by the ensemble $\left\{  p_{X^{n}}\left(  x^{n}\right)  ,\rho_{x^{n}%
}\right\}  $, or equivalently, by a classical-quantum state of the following
form:%
\[
\sum_{x^{n}}p_{X^{n}}\left(  x^{n}\right)  \left\vert x^{n}\right\rangle
\left\langle x^{n}\right\vert ^{X^{n}}\otimes\rho_{x^{n}}^{B^{n}}.
\]
The goal is for Alice to communicate the particular sequence $x^{n}$ that the
source issues, by using as few bits as possible.

\begin{figure}[ptb]
\begin{center}
\includegraphics[
natheight=3.933200in,
natwidth=4.253100in,
height=2.3644in,
width=2.5374in
]{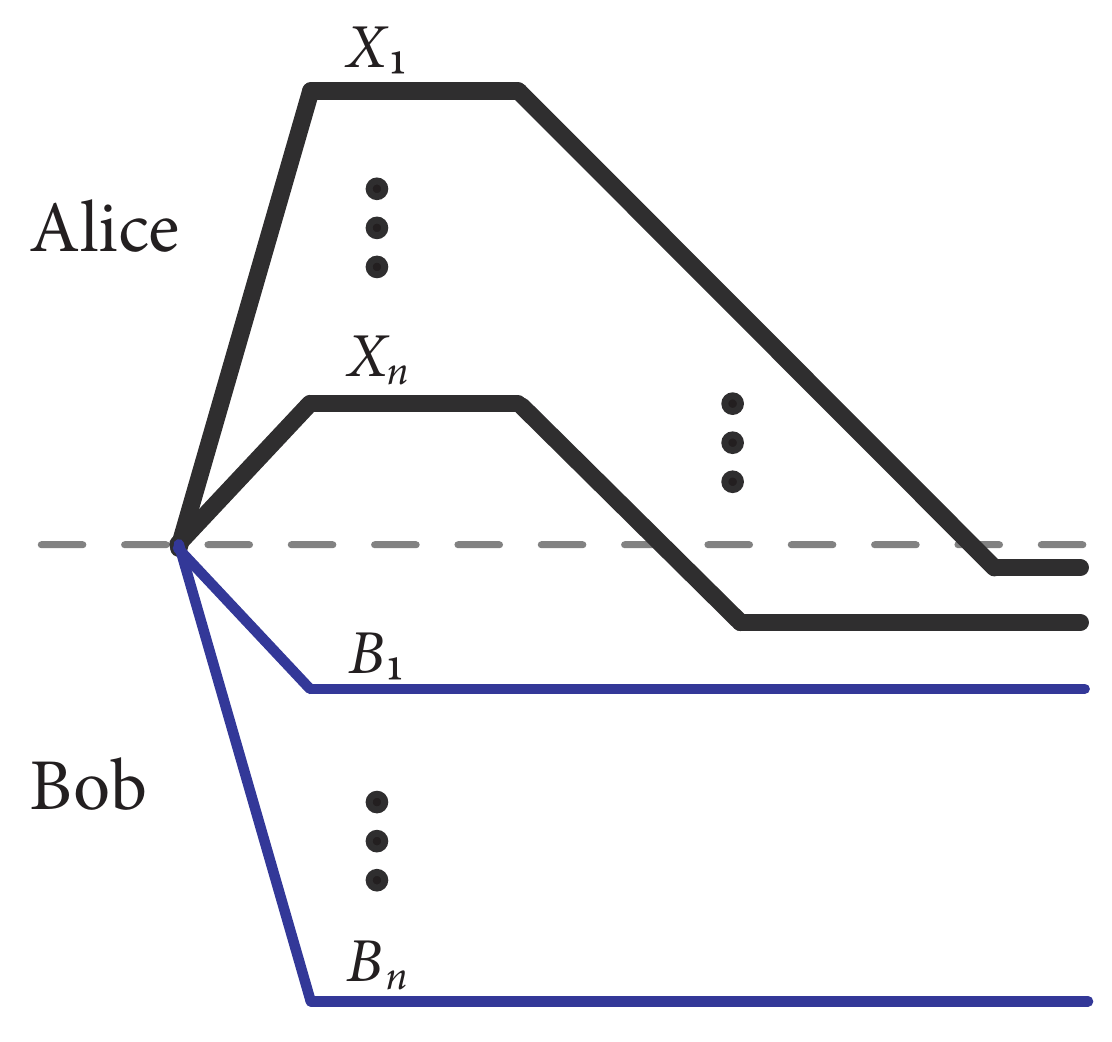}
\end{center}
\caption{\textbf{Ideal protocol for classical data compression with quantum
side information.} In this protocol, we assume that a quantum information
source distributes many copies of a classical-quantum state to Alice and Bob,
such that Alice receives the classical part and Bob receives the quantum part.
The goal is for Alice to communicate the classical sequence received from the
source to Bob. In an ideal case, she would simply transmit this sequence to
Bob. Though, it is possible to obtain a significant savings in communication
by allowing for an asymptotically vanishing error and for Bob to infer
something about the classical sequence from his correlated quantum states.}%
\label{fig:ideal-CDC-QSI}%
\end{figure}

One potential strategy is to exploit Shannon compression---just compress the
sequence to $nH\left(  X\right)  $ bits, keeping only the typical set
according to the distribution $p_{X^{n}}\left(  x^{n}\right)  $. But they can
actually do much better in general if Bob exploits his quantum side
information in the form of the correlated state $\rho_{x^{n}}$.

The most general protocol has Alice hash her sequence $x^{n}$ to some variable
$L\in\mathcal{L}$ (this is just some many-to-one mapping $f:\mathcal{X}%
^{n}\rightarrow\mathcal{L}$). She transmits the variable $L$ to Bob over a
noiseless classical channel using $\log_{2}\left\vert \mathcal{L}\right\vert $
bits. Bob then exploits the hashed variable $L$ and his quantum side
information $\rho_{x^{n}}$ to distinguish between all of the possible states
that are consistent with the hash~$L$ (i.e., his action will be some quantum
measurement depending on the hash $L$). The output of his measurement is some
approximation sequence $\hat{X}^{n}$. The protocol has one parameter that
characterizes its quality. We demand that the state $\sigma^{X^{n}\hat{X}%
^{n}B^{n}}$ after Alice and Bob's actions should be close in trace distance to
an ideal state $\rho^{X^{n}\overline{X}^{n}B^{n}}$, where $\overline{X}^{n}$
is a copy of $X^{n}$ (this would be the ideal output if Alice were to just
send a copy of the variable $X^{n}$ to Bob):%
\begin{equation}
\left\Vert \rho^{X^{n}\overline{X}^{n}B^{n}}-\sigma^{X^{n}\hat{X}^{n}B^{n}%
}\right\Vert _{1}\leq\epsilon. \label{eq:cond-CDC-QSI}%
\end{equation}
The above specifies an $\left(  n,R,\epsilon\right)  $ code for this task,
where $R\equiv\log_{2}\left\vert \mathcal{L}\right\vert /n$.

\begin{figure}[ptb]
\begin{center}
\includegraphics[
natheight=4.139900in,
natwidth=7.166700in,
height=2.5114in,
width=3.5405in
]{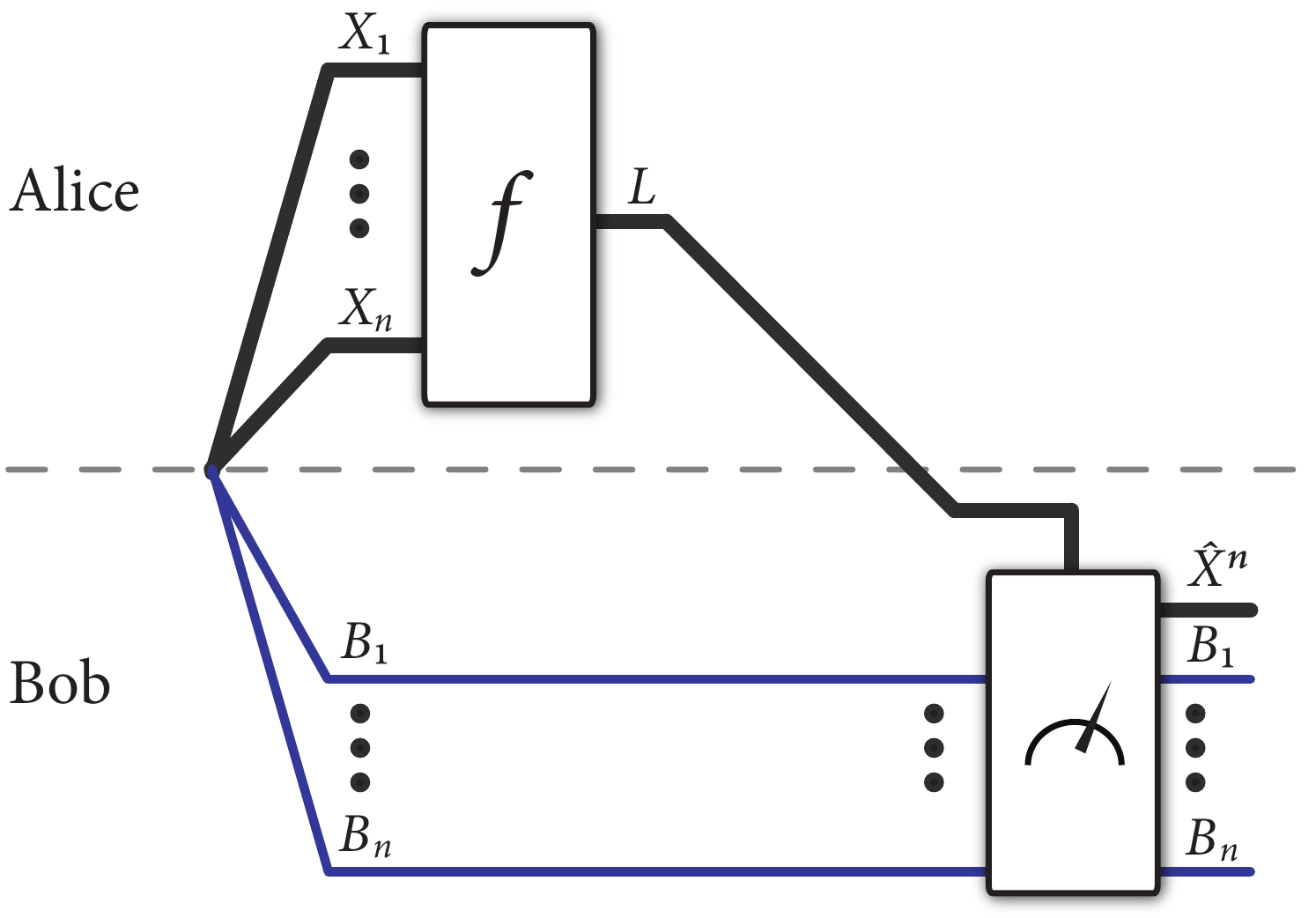}
\end{center}
\caption{\textbf{Classical data compression with QSI\ protocol.} The protocol
begins with the source distributing a random classical sequence to Alice and a
correlated quantum state to Bob. Alice begins by hashing the sequence to a
variable $L$ with some hash function $f$. She then transmits the variable $L$
to Bob, using $\log_{2}\left\vert \mathcal{L}\right\vert $ noiseless classical
bit channels. Bob receives the hash, and he then enumerates all of the
sequences $x^{n}$ that are consistent with the hash (so that $f\left(
x^{n}\right)  =l$). He performs a \textquotedblleft quantum
scan\textquotedblright\ over all of the quantum states $\rho_{x^{\prime n}}$
that are consistent with the received hash. This quantum scan amounts to a
sequence of binary quantum measurements, effectively asking, \textquotedblleft
Does my quantum state correspond to the first sequence consistent with the
hash? To the second? etc.\textquotedblright\ After receiving the answer
\textquotedblleft yes\textquotedblright\ to one of these questions, he
declares the \textquotedblleft yes\textquotedblright\ sequence to be the one
sent from Alice. This strategy has asymptotically vanishing error as long as
the size of the hash is at least $nH\left(  X|B\right)  $ bits.}%
\label{fig:actual-CDC-QSI}%
\end{figure}

A rate $R$ is achievable if there exists an $\left(  n,R,\epsilon\right)  $
code for all $\epsilon>0$ and sufficiently large $n$.

\subsection{Classical data compression with QSI theorem}

\begin{theorem}
[Classical data compression with quantum side information]\label{thm:cdc-qsi}%
Suppose that%
\[
\sum_{x}p_{X}\left(  x\right)  \left\vert x\right\rangle \left\langle
x\right\vert ^{X}\otimes\rho_{x}^{B}%
\]
is a classical-quantum state that characterizes a classical-quantum source.
Then the conditional von Neumann entropy $H\left(  X|B\right)  $ is the
smallest possible achievable rate for classical data compression with quantum
side information for this source:%
\[
\inf\left\{  R\ |\ R\text{ is achievable}\right\}  =H\left(  X|B\right)  .
\]

\end{theorem}

\subsection{Achievability proof for CDC with QSI}

\label{sec:CDC-QSI-achieve}The resource inequality for this communication task
is as follows:%
\[
\left\langle \rho^{XB}\right\rangle +H\left(  X|B\right)  \left[  c\rightarrow
c\right]  \geq\left\langle \rho^{XX_{B}B}\right\rangle ,
\]
the meaning being that if Alice and Bob share many copies of the state
$\rho^{XB}$ and she communicates at a rate $H\left(  X|B\right)  $ to Bob,
then they can construct the state $\rho^{XX_{B}B}$, so that Bob has a copy of
the variable $X$.

The strategy for achievability is for Alice to hash her sequence $X^{n}$ to
some variable $L$. Bob then receives the variable $L$ after Alice communicates
it to him. He then \textquotedblleft scans\textquotedblright\ over all of the
quantum states $\rho_{x^{n}}$ that are consistent with the hash and such that
the sequence $x^{n}$ is typical (the strategy essentially disregards the
atypical sequences $x^{n}$ since their total probability mass is
asymptotically negligible). He can accomplish this \textquotedblleft
scan\textquotedblright\ by performing a sequential decoding strategy
\cite{GLM10,S11}, which consists of binary tests of the form,
\textquotedblleft Is this state consistent with the hash? Or
this one? etc.\textquotedblright\ He performs these tests until he receives a
\textquotedblleft yes\textquotedblright\ answer in one of his measurements.

The intuition for why $H\left(  X|B\right)  $ should be the ultimate rate of
communication is that there are $\approx2^{nH\left(  X\right)  }$ sequences of
the source to account for (the typical ones). From the HSW\ theorem
\cite{Hol98,PhysRevA.56.131}, we know that the maximal number of sequences
that Bob can distinguish is $\approx2^{nI\left(  X;B\right)  }$. Thus, if
Alice divides the source sequences into $\approx2^{nH\left(  X\right)
}/2^{nI\left(  X;B\right)  }=2^{nH\left(  X|B\right)  }$ groups and sends the
label of the group, then Bob should be able to determine which sequence
$x^{n}$ is the one that the source issued.

\textbf{Detailed Strategy.} More formally, the encoding strategy is as
follows. Alice and Bob are allowed to have an agreed-upon hash function
$f:\mathcal{X}^{n}\rightarrow\mathcal{L}$, selected at random from a
two-universal family. A hash function $f$ has a collision if two differing
sequences $x_{1}^{n},x_{2}^{n}\in\mathcal{X}^{n}$ hash to the same value:%
\[
x_{1}^{n}\neq x_{2}^{n}\ \ \ \Longrightarrow\ \ \ f\left(  x_{1}^{n}\right)
=f\left(  x_{2}^{n}\right)  .
\]
A two-universal family has the property that the probability of a collision is
the same as that for a uniformly random function (where the probability is
with respect to the random choice of the hash function):%
\begin{equation}
x_{1}^{n}\neq x_{2}^{n}\ \ \ \Longrightarrow\ \ \ \Pr_{f}\left\{  f\left(
x_{1}^{n}\right)  =f\left(  x_{2}^{n}\right)  \right\}  \leq\frac
{1}{\left\vert \mathcal{L}\right\vert }=2^{-nR}. \label{eq:two-univ-prop}%
\end{equation}
Such a strategy is equivalent to the \textquotedblleft random
binning\textquotedblright\ strategy often discussed in information theory
texts \cite{CT91,el2010lecture}.

Upon receiving the hash value $l$, Bob performs a sequence of binary
measurements $\left\{  \Pi_{x^{n}},I-\Pi_{x^{n}}\right\}  $ for all the
sequences $x^{n}$ that are consistent with the hash value (so that $f\left(
x^{n}\right)  =l$) and such that they are strongly typical ($x^{n}\in
T_{\delta}^{X^{n}}$---see Appendix~\ref{sec:typ-review} for details). We
define the set $\mathcal{A}\left(  f,l\right)  $ to capture these sequences:%
\begin{equation}
\mathcal{A}\left(  f,l\right)  \equiv\left\{  x^{n}:f\left(  x^{n}\right)
=l,\ \ x^{n}\in T_{\delta}^{X^{n}}\right\}  . \label{eq:seq-cons-w-hash}%
\end{equation}
The projector $\Pi_{x^{n}}$ is a strong conditionally typical projector (see
Appendix~\ref{sec:typ-review}), with the property that%
\[
\text{Tr}\left\{  \Pi_{x^{n}}\ \rho_{x^{n}}\right\}  \geq1-\epsilon,
\]
for all $\epsilon>0$ and sufficiently large $n$. From the above property, we
would expect these measurements to perform well in identifying the actual
state transmitted.

\textbf{Error Analysis.} We define the error probability as follows:%
\[
\Pr\left\{  \text{\textquotedblleft error @ decoder\textquotedblright%
}\right\}  =\sum_{x^{n}}p_{X^{n}}\left(  x^{n}\right)  \ \Pr\left\{
\text{\textquotedblleft error @ decoder\textquotedblright\ }|\ x^{n}\right\}
.
\]
It is then clear that we can focus on the typical sequences $x^{n}$ because
the above error probability is equal to%
\begin{align}
&  \sum_{x^{n}\in T_{\delta}^{X^{n}}}p_{X^{n}}\left(  x^{n}\right)
\ \Pr\left\{  \text{\textquotedblleft error @ decoder\textquotedblright%
\ }|\ x^{n}\right\}  +\sum_{x^{n}\notin T_{\delta}^{X^{n}}}p_{X^{n}}\left(
x^{n}\right)  \ \Pr\left\{  \text{\textquotedblleft error @
decoder\textquotedblright\ }|\ x^{n}\right\} \nonumber\\
&  \leq\sum_{x^{n}\in T_{\delta}^{X^{n}}}p_{X^{n}}\left(  x^{n}\right)
\ \Pr\left\{  \text{\textquotedblleft error @ decoder\textquotedblright%
\ }|\ x^{n}\right\}  +\epsilon. \label{eq:CDCQSI-err-prob-1}%
\end{align}

Now we consider the error term $\Pr\left\{  \text{\textquotedblleft error @
decoder\textquotedblright\ }|\ x^{n}\right\}  $. Let $a_{1}^{\left(  l\right)
}$, \ldots, $a_{\left\vert \mathcal{A}\right\vert }^{\left(  l\right)  }$
enumerate all of the sequences in the set $\mathcal{A}\left(  f,l\right)  $
defined in (\ref{eq:seq-cons-w-hash}) (those sequences consistent with the
hash $l$). Let $a_{m}^{\left(  l\right)  }$ be the actual sequence $x^{n}$
that the source issues. Then the probability for a correct decoding for Bob is
as follows:%
\[
\text{Tr}\left\{  \Pi_{a_{m}^{\left(  l\right)  }}\ \hat{\Pi}_{a_{m-1}%
^{\left(  l\right)  }}\cdots\hat{\Pi}_{a_{1}^{\left(  l\right)  }}%
\ \rho_{a_{m}^{\left(  l\right)  }}\ \hat{\Pi}_{a_{1}^{\left(  l\right)  }%
}\cdots\hat{\Pi}_{a_{m-1}^{\left(  l\right)  }}\ \Pi_{a_{m}^{\left(  l\right)
}}\right\}  ,
\]
where $\hat{\Pi}_{x^{n}}\equiv I-\Pi_{x^{n}}$, so that the binary tests give a
response of \textquotedblleft no\textquotedblright\ until the test for
$a_{m}^{\left(  l\right)  }$ gives a response of \textquotedblleft
yes.\textquotedblright\ The probability for incorrectly decoding with this
strategy is%
\[
1-\text{Tr}\left\{  \Pi_{a_{m}^{\left(  l\right)  }}\ \hat{\Pi}_{a_{m-1}%
^{\left(  l\right)  }}\cdots\hat{\Pi}_{a_{1}^{\left(  l\right)  }}%
\ \rho_{a_{m}^{\left(  l\right)  }}\ \hat{\Pi}_{a_{1}^{\left(  l\right)  }%
}\cdots\hat{\Pi}_{a_{m-1}^{\left(  l\right)  }}\ \Pi_{a_{m}^{\left(  l\right)
}}\right\}  ,
\]
so that we can write the error probability in (\ref{eq:CDCQSI-err-prob-1}) as%
\begin{equation}
\sum_{l\in\mathcal{L}}\sum_{a_{m}^{\left(  l\right)  }\in\mathcal{A}\left(
f,l\right)  }p_{X^{n}}\left(  a_{m}^{\left(  l\right)  }\right)  \ \left[
1-\text{Tr}\left\{  \Pi_{a_{m}^{\left(  l\right)  }}\ \hat{\Pi}_{a_{m-1}%
^{\left(  l\right)  }}\cdots\hat{\Pi}_{a_{1}^{\left(  l\right)  }}%
\ \rho_{a_{m}^{\left(  l\right)  }}\ \hat{\Pi}_{a_{1}^{\left(  l\right)  }%
}\cdots\hat{\Pi}_{a_{m-1}^{\left(  l\right)  }}\ \Pi_{a_{m}^{\left(  l\right)
}}\right\}  \right]  . \label{eq:CDCQSI-err-prob-2}%
\end{equation}
We can rewrite this error probability as%
\[
\sum_{l\in\mathcal{L}}\sum_{a_{m}^{\left(  l\right)  }\in\mathcal{A}\left(
f,l\right)  }p_{X^{n}}\left(  a_{m}^{\left(  l\right)  }\right)
\ \text{Tr}\left\{  \left(  I-\Theta_{a_{m}^{\left(  l\right)  }}\right)
\rho_{a_{m}^{\left(  l\right)  }}\right\}  ,
\]
where we define the POVM\ element $\Theta_{a_{m}^{\left(  l\right)  }}$\ as%
\begin{equation}
\Theta_{a_{m}^{\left(  l\right)  }}\equiv\hat{\Pi}_{a_{1}^{\left(  l\right)
}}\cdots\hat{\Pi}_{a_{m-1}^{\left(  l\right)  }}\ \Pi_{a_{m}^{\left(
l\right)  }}\ \hat{\Pi}_{a_{m-1}^{\left(  l\right)  }}\cdots\hat{\Pi}%
_{a_{1}^{\left(  l\right)  }}. \label{eq:POVM-element-CDC-QSI}%
\end{equation}
Using the facts that (see Appendix~\ref{sec:typ-review})%
\[
1=\text{Tr}\left\{  \rho_{a_{m}^{\left(  l\right)  }}\right\}  =\text{Tr}%
\left\{  \Pi\ \rho_{a_{m}^{\left(  l\right)  }}\right\}  +\text{Tr}\left\{
\left(  I-\Pi\right)  \rho_{a_{m}^{\left(  l\right)  }}\right\}  \leq
\text{Tr}\left\{  \Pi\ \rho_{a_{m}^{\left(  l\right)  }}\right\}  +\epsilon,
\]
where $\Pi$ is the typical projector for the average state $\sum_{x}%
p_{X}\left(  x\right)  \rho_{x}$, and%
\begin{align*}
&  \text{Tr}\left\{  \Pi_{a_{m}^{\left(  l\right)  }}\ \hat{\Pi}%
_{a_{m-1}^{\left(  l\right)  }}\cdots\hat{\Pi}_{a_{1}^{\left(  l\right)  }%
}\ \rho_{a_{m}^{\left(  l\right)  }}\ \hat{\Pi}_{a_{1}^{\left(  l\right)  }%
}\cdots\hat{\Pi}_{a_{m-1}^{\left(  l\right)  }}\ \Pi_{a_{m}^{\left(  l\right)
}}\right\} \\
&  =\text{Tr}\left\{  \hat{\Pi}_{a_{1}^{\left(  l\right)  }}\cdots\hat{\Pi
}_{a_{m-1}^{\left(  l\right)  }}\ \Pi_{a_{m}^{\left(  l\right)  }}\ \hat{\Pi
}_{a_{m-1}^{\left(  l\right)  }}\cdots\hat{\Pi}_{a_{1}^{\left(  l\right)  }%
}\ \rho_{a_{m}^{\left(  l\right)  }}\right\} \\
&  \geq\text{Tr}\left\{  \hat{\Pi}_{a_{1}^{\left(  l\right)  }}\cdots\hat{\Pi
}_{a_{m-1}^{\left(  l\right)  }}\ \Pi_{a_{m}^{\left(  l\right)  }}\ \hat{\Pi
}_{a_{m-1}^{\left(  l\right)  }}\cdots\hat{\Pi}_{a_{1}^{\left(  l\right)  }%
}\ \Pi\ \rho_{a_{m}^{\left(  l\right)  }}\ \Pi\right\}  -\left\Vert
\rho_{a_{m}^{\left(  l\right)  }}-\Pi\ \rho_{a_{m}^{\left(  l\right)  }}%
\ \Pi\right\Vert _{1}\\
&  \geq\text{Tr}\left\{  \hat{\Pi}_{a_{1}^{\left(  l\right)  }}\cdots\hat{\Pi
}_{a_{m-1}^{\left(  l\right)  }}\ \Pi_{a_{m}^{\left(  l\right)  }}\ \hat{\Pi
}_{a_{m-1}^{\left(  l\right)  }}\cdots\hat{\Pi}_{a_{1}^{\left(  l\right)  }%
}\ \Pi\ \rho_{a_{m}^{\left(  l\right)  }}\ \Pi\right\}  -2\sqrt{\epsilon},
\end{align*}
we can bound the expression in (\ref{eq:CDCQSI-err-prob-2}) from above by
\begin{equation}
\sum_{l\in\mathcal{L}}\sum_{a_{m}^{\left(  l\right)  }\in\mathcal{A}\left(
f,l\right)  }p_{X^{n}}\left(  a_{m}^{\left(  l\right)  }\right)  \ \left[
\text{Tr}\left\{  \Pi\ \rho_{a_{m}^{\left(  l\right)  }}\ \Pi\right\}
-\text{Tr}\left\{  \Pi_{a_{m}^{\left(  l\right)  }}\ \hat{\Pi}_{a_{m-1}%
^{\left(  l\right)  }}\cdots\hat{\Pi}_{a_{1}^{\left(  l\right)  }}\ \Pi
\ \rho_{a_{m}^{\left(  l\right)  }}\ \Pi\ \hat{\Pi}_{a_{1}^{\left(  l\right)
}}\cdots\hat{\Pi}_{a_{m-1}^{\left(  l\right)  }}\ \Pi_{a_{m}^{\left(
l\right)  }}\right\}  \right]  , \label{eq:avg-typ-bound}%
\end{equation}
(with the other terms $\epsilon+2\sqrt{\epsilon}$ omitted for simplicity). We
now apply Sen's non-commutative union bound \cite{S11}%
\ (Lemma~\ref{lem-non-com-union-bound} in Appendix~\ref{sec:useful-lemmas})
along with concavity of square-root to obtain the following upper bound:%
\[
2\sqrt{\sum_{l\in\mathcal{L}}\sum_{a_{m}^{\left(  l\right)  }\in
\mathcal{A}\left(  f,l\right)  }p_{X^{n}}\left(  a_{m}^{\left(  l\right)
}\right)  \ \left[  \text{Tr}\left\{  \left(  I-\Pi_{a_{m}^{\left(  l\right)
}}\right)  \Pi\ \rho_{a_{m}^{\left(  l\right)  }}\ \Pi\right\}  +\sum
_{i=1}^{m-1}\text{Tr}\left\{  \Pi_{a_{i}^{\left(  l\right)  }}\ \Pi
\ \rho_{a_{m}^{\left(  l\right)  }}\ \Pi\right\}  \right]  }.
\]
For the first term in the square-root, we have that%
\begin{align}
\text{Tr}\left\{  \left(  I-\Pi_{a_{m}^{\left(  l\right)  }}\right)  \Pi
\ \rho_{a_{m}^{\left(  l\right)  }}\ \Pi\right\}   &  \leq\text{Tr}\left\{
\left(  I-\Pi_{a_{m}^{\left(  l\right)  }}\right)  \rho_{a_{m}^{\left(
l\right)  }}\right\}  +\left\Vert \rho_{a_{m}^{\left(  l\right)  }}-\Pi
\ \rho_{a_{m}^{\left(  l\right)  }}\ \Pi\right\Vert _{1}\nonumber\\
&  \leq\epsilon+2\sqrt{\epsilon}. \label{eq:CDC-QSI-first-err-bnd}%
\end{align}
For the second term in the square-root, we have%
\begin{align*}
&  \sum_{l\in\mathcal{L}}\sum_{a_{m}^{\left(  l\right)  }\in\mathcal{A}\left(
f,l\right)  }p_{X^{n}}\left(  a_{m}^{\left(  l\right)  }\right)  \sum
_{i=1}^{m-1}\text{Tr}\left\{  \Pi_{a_{i}^{\left(  l\right)  }}\ \Pi
\ \rho_{a_{m}^{\left(  l\right)  }}\ \Pi\right\} \\
&  \leq\sum_{l\in\mathcal{L}}\sum_{a_{m}^{\left(  l\right)  }\in
\mathcal{A}\left(  f,l\right)  }p_{X^{n}}\left(  a_{m}^{\left(  l\right)
}\right)  \sum_{a_{i}^{\left(  l\right)  }\in\mathcal{A}\left(  f,l\right)
\ :\ i\neq m}\text{Tr}\left\{  \Pi_{a_{i}^{\left(  l\right)  }}\ \Pi
\ \rho_{a_{m}^{\left(  l\right)  }}\ \Pi\right\} \\
&  =\sum_{x^{n}\in T_{\delta}^{X^{n}}}p_{X^{n}}\left(  x^{n}\right)
\sum_{x^{\prime n}\in T_{\delta}^{X^{n}}\ :\ x^{\prime n}\neq x^{n}%
}\mathcal{I}\left(  f\left(  x^{\prime n}\right)  =f\left(  x^{n}\right)
\right)  \ \text{Tr}\left\{  \Pi_{x^{\prime n}}\ \Pi\ \rho_{x^{n}}%
\ \Pi\right\} \\
&  \leq\sum_{x^{n}}p_{X^{n}}\left(  x^{n}\right)  \sum_{x^{\prime n}\in
T_{\delta}^{X^{n}}\ :\ x^{\prime n}\neq x^{n}}\mathcal{I}\left(  f\left(
x^{\prime n}\right)  =f\left(  x^{n}\right)  \right)  \ \text{Tr}\left\{
\Pi_{x^{\prime n}}\ \Pi\ \rho_{x^{n}}\ \Pi\right\}  .
\end{align*}
The first inequality follows by summing over all the indices not equal to $m$.
The equality follows by introducing the indicator function $\mathcal{I}\left(
f\left(  x^{\prime n}\right)  =f\left(  x^{n}\right)  \right)  $, and the last
inequality follows by summing over all sequences $x^{n}$.

We now analyze the expectation of the error probability, with respect to the
random hash function $f$. (We can imagine that this expectation was there from
the beginning of the analysis, and apply concavity of square-root to bring it
over this second term).\ This leads to%
\begin{align*}
&  \mathbb{E}_{f}\left\{  \sum_{x^{n}}p_{X^{n}}\left(  x^{n}\right)
\sum_{x^{\prime n}\in T_{\delta}^{X^{n}}\ :\ x^{\prime n}\neq x^{n}%
}\mathcal{I}\left(  f\left(  x^{\prime n}\right)  =f\left(  x^{n}\right)
\right)  \ \text{Tr}\left\{  \Pi_{x^{\prime n}}\ \Pi\ \rho_{x^{n}}%
\ \Pi\right\}  \right\} \\
&  =\sum_{x^{n}}p_{X^{n}}\left(  x^{n}\right)  \sum_{x^{\prime n}\in
T_{\delta}^{X^{n}}\ :\ x^{\prime n}\neq x^{n}}\mathbb{E}_{f}\left\{
\mathcal{I}\left(  f\left(  x^{\prime n}\right)  =f\left(  x^{n}\right)
\right)  \right\}  \ \text{Tr}\left\{  \Pi_{x^{\prime n}}\ \Pi\ \rho_{x^{n}%
}\ \Pi\right\} \\
&  =\sum_{x^{n}}p_{X^{n}}\left(  x^{n}\right)  \sum_{x^{\prime n}\in
T_{\delta}^{X^{n}}\ :\ x^{\prime n}\neq x^{n}}\Pr_{f}\left\{  f\left(
x^{\prime n}\right)  =f\left(  x^{n}\right)  \right\}  \ \text{Tr}\left\{
\Pi_{x^{\prime n}}\ \Pi\ \rho_{x^{n}}\ \Pi\right\} \\
&  \leq2^{-nR}\sum_{x^{n}}p_{X^{n}}\left(  x^{n}\right)  \sum_{x^{\prime n}\in
T_{\delta}^{X^{n}}}\text{Tr}\left\{  \Pi_{x^{\prime n}}\ \Pi\ \rho_{x^{n}%
}\ \Pi\right\}  ,
\end{align*}
where the inequality follows from the two-universal property in
(\ref{eq:two-univ-prop}), the fact that $R=\log_{2}\left\vert \mathcal{L}%
\right\vert /n$, and by summing over all sequences $x^{\prime n}$ in the
typical set. Continuing, we have%
\begin{align*}
&  =2^{-nR}\sum_{x^{\prime n}\in T_{\delta}^{X^{n}}}\text{Tr}\left\{
\Pi_{x^{\prime n}}\ \Pi\ \left(  \sum_{x^{n}}p_{X^{n}}\left(  x^{n}\right)
\rho_{x^{n}}\right)  \ \Pi\right\} \\
&  =2^{-nR}\sum_{x^{\prime n}\in T_{\delta}^{X^{n}}}\text{Tr}\left\{
\Pi_{x^{\prime n}}\ \Pi\ \rho^{\otimes n}\ \Pi\right\} \\
&  \leq2^{-nR}\ 2^{-n\left[  H\left(  B\right)  -\delta\right]  }%
\sum_{x^{\prime n}\in T_{\delta}^{X^{n}}}\text{Tr}\left\{  \Pi_{x^{\prime n}%
}\ \Pi\right\} \\
&  \leq2^{-nR}\ 2^{-n\left[  H\left(  B\right)  -\delta\right]  }%
\sum_{x^{\prime n}\in T_{\delta}^{X^{n}}}\text{Tr}\left\{  \Pi_{x^{\prime n}%
}\right\} \\
&  \leq2^{-nR}\ 2^{-n\left[  H\left(  B\right)  -\delta\right]  }\ 2^{n\left[
H\left(  B|X\right)  +\delta\right]  }\ 2^{n\left[  H\left(  X\right)
+\delta\right]  }\\
&  =2^{-n\left[  R-H\left(  X|B\right)  -3\delta\right]  }.
\end{align*}
The first inequality follows from the operator inequality $\Pi\ \rho^{\otimes
n}\ \Pi\leq2^{-n\left[  H\left(  B\right)  -\delta\right]  }\ \Pi$, and the
second from $\Pi\leq I$. The final inequality follows from the bounds
Tr$\left\{  \Pi_{x^{\prime n}}\right\}  \leq2^{n\left[  H\left(  B|X\right)
+\delta\right]  }$ and $\left\vert T_{\delta}^{X^{n}}\right\vert
\leq2^{n\left[  H\left(  X\right)  +\delta\right]  }$. Collecting everything,
the overall error probability is bounded by%
\begin{equation}
\epsilon^{\prime}\equiv\epsilon+2\sqrt{\epsilon}+2\sqrt{\epsilon
+2\sqrt{\epsilon}+2^{-n\left[  R-H\left(  X|B\right)  -3\delta\right]  }}.
\label{eq:CDC-QSI-total-bound}%
\end{equation}

Since the expectation of the above error probability is small (where the
expectation is with respect to the random choice of hash function), there
exists some particular hash function from the family such that the above
inequality is true. Thus, as long as $R=H\left(  X|B\right)  +4\delta$, we can
guarantee that the error probability of the scheme is arbitrarily small.

We now argue that it is possible to make the state after the decoding be
arbitrarily close to the initial state, so that the condition in
(\ref{eq:cond-CDC-QSI}) holds. After recovering the sequence $x^{n}%
=a_{m}^{\left(  l\right)  }$ issued by the source, Bob can place it in a
classical register, and the post-measurement state from sequential decoding
has the following form:%
\[
\frac{1}{\text{Tr}\left\{  \Theta_{a_{m}^{\left(  l\right)  }}\rho
_{a_{m}^{\left(  l\right)  }}\right\}  }\Pi_{a_{m}^{\left(  l\right)  }}%
\ \hat{\Pi}_{a_{m-1}^{\left(  l\right)  }}\cdots\hat{\Pi}_{a_{1}^{\left(
l\right)  }}\ \rho_{a_{m}^{\left(  l\right)  }}\ \hat{\Pi}_{a_{1}^{\left(
l\right)  }}\cdots\hat{\Pi}_{a_{m-1}^{\left(  l\right)  }}\ \Pi_{a_{m}%
^{\left(  l\right)  }},
\]
with $\Theta_{a_{m}^{\left(  l\right)  }}$ defined in
(\ref{eq:POVM-element-CDC-QSI}) and assuming a correct decoding. The operators
$\Pi_{a_{m}^{\left(  l\right)  }}\ \hat{\Pi}_{a_{m-1}^{\left(  l\right)  }%
}\cdots\hat{\Pi}_{a_{1}^{\left(  l\right)  }}$ and $\sqrt{\Theta
_{a_{m}^{\left(  l\right)  }}}$ are related by a left polar decomposition:%
\[
\Pi_{a_{m}^{\left(  l\right)  }}\ \hat{\Pi}_{a_{m-1}^{\left(  l\right)  }%
}\cdots\hat{\Pi}_{a_{1}^{\left(  l\right)  }}=U_{a_{m}^{\left(  l\right)  }%
}\sqrt{\Theta_{a_{m}^{\left(  l\right)  }}},
\]
for some unitary $U_{a_{m}^{\left(  l\right)  }}$. So after Bob recovers the
sequence $a_{m}^{\left(  l\right)  }$, he applies the unitary $U_{a_{m}%
^{\left(  l\right)  }}^{\dag}$, and the state becomes as follows:%
\[
\frac{1}{\text{Tr}\left\{  \Theta_{a_{m}^{\left(  l\right)  }}\rho
_{a_{m}^{\left(  l\right)  }}\right\}  }\sqrt{\Theta_{a_{m}^{\left(  l\right)
}}}\ \rho_{a_{m}^{\left(  l\right)  }}\ \sqrt{\Theta_{a_{m}^{\left(  l\right)
}}}.
\]
We can now show that the condition in (\ref{eq:cond-CDC-QSI}) holds for this
decoding procedure (including the unitaries $U_{a_{m}^{\left(  l\right)  }%
}^{\dag}$). Consider that for all typical sequences $x^{n}\in T_{\delta
}^{X^{n}}$, the trace distance between the initial state and the
post-measurement state has the following bound:%
\[
\left\Vert \rho_{x^{n}}-\text{Tr}\left\{  \Theta_{x^{n}}\rho_{x^{n}}\right\}
^{-1}\sqrt{\Theta_{x^{n}}}\ \rho_{x^{n}}\ \sqrt{\Theta_{x^{n}}}\right\Vert
_{1}\leq2\sqrt{\text{Tr}\left\{  \left(  I-\Theta_{x^{n}}\right)  \rho_{x^{n}%
}\right\}  },
\]
which follows from the Gentle Measurement Lemma (Lemma~9.4.1 of
Ref.~\cite{W11}). Combining this bound with the fact that there is no
measurement when $x^{n}\notin T_{\delta}^{X^{n}}$ and defining $\Theta_{x^{n}%
}=I$ in this case, averaging over $p_{X^{n}}\left(  x^{n}\right)  $, and then
applying our bound in (\ref{eq:CDC-QSI-total-bound}) and concavity of
square-root, we obtain the following upper bound:%
\[
\sum_{x^{n}}p_{X^{n}}\left(  x^{n}\right)  \left\Vert \rho_{x^{n}}%
-\text{Tr}\left\{  \Theta_{x^{n}}\rho_{x^{n}}\right\}  ^{-1}\sqrt
{\Theta_{x^{n}}}\ \rho_{x^{n}}\ \sqrt{\Theta_{x^{n}}}\right\Vert _{1}%
\leq2\sqrt{\epsilon^{\prime}}.
\]
It then follows that the condition in (\ref{eq:cond-CDC-QSI}) is satisfied for
this protocol.

\subsection{Converse theorem for CDC with QSI}

This section provides a simple proof of the converse theorem for CQC-QSI. The
converse demonstrates that the single-letter rate in Theorem~\ref{thm:cdc-qsi}%
\ is optimal. An inspection of the proof reveals a close similarity with the
Slepian-Wolf converse in Ref.~\cite{el2010lecture}.

In the most general protocol for this task, Alice receives the sequence
$X^{n}$ from the source. She then hashes it to a random variable $L$ where
$f\left(  X^{n}\right)  =L$ and sends it over to Bob via some noiseless
classical bit channels. Bob receives $L$, processes it and $B^{n}$ to obtain
an estimate $\hat{X}^{n}$ of $X^{n}$. If the protocol is any good for this
task, then the actual state $\omega^{X^{n}\hat{X}^{n}B^{n}}$ at the end should
be $\epsilon$-close in trace distance to the ideal state $\sigma
^{X^{n}\overline{X}^{n}B^{n}}$, where $\overline{X}^{n}$ is a copy of the
variable $X^{n}$:%
\begin{equation}
\left\Vert \omega^{X^{n}\hat{X}^{n}B^{n}}-\sigma^{X^{n}\overline{X}^{n}B^{n}%
}\right\Vert _{1}\leq\epsilon. \label{eq:CDC-QSI-good-protocol}%
\end{equation}
A proof of the converse goes as follows:%
\begin{align*}
nR  &  \geq H\left(  L\right) \\
&  \geq H\left(  L|B^{n}\right) \\
&  =I\left(  X^{n};L|B^{n}\right)  +H\left(  L|B^{n}X^{n}\right) \\
&  \geq I\left(  X^{n};L|B^{n}\right) \\
&  =H\left(  X^{n}|B^{n}\right)  -H\left(  X^{n}|LB^{n}\right) \\
&  \geq H\left(  X^{n}|B^{n}\right)  _{\omega}-H(X^{n}|\hat{X}^{n})_{\omega}\\
&  \geq H\left(  X^{n}|B^{n}\right)  _{\sigma}-n\epsilon^{\prime}\\
&  =nH\left(  X|B\right)  -n\epsilon^{\prime}.
\end{align*}
The first two inequalities follow for reasons similar to those in the previous
converse in Section~\ref{sec:IC-converse}. The first equality is an entropy
identity, and the third inequality follows because $H\left(  L|B^{n}%
X^{n}\right)  \geq0$ for a classical variable $L$. The second equality is
another entropy identity. The fourth inequality follows from quantum data
processing of $L$ and $B^{n}$ to obtain the estimate $\hat{X}^{n}$. The final
inequality follows from the condition in (\ref{eq:CDC-QSI-good-protocol}),
continuity of entropy (Alicki-Fannes' inequality \cite{AF04}), and the fact
that $H(X^{n}|\overline{X}^{n})_{\sigma}=0$ since $\overline{X}^{n}$ is copy
of $X^{n}$, with $\epsilon^{\prime}$ being some function $g\left(
\epsilon\right)  $ such that $\lim_{\epsilon\rightarrow0}g\left(
\epsilon\right)  =0$. The final equality follows because the entropy is
additive on a tensor-power state.

\section{Measurement compression with quantum side information}

\label{sec:MC-QSI-feedback}We now discuss another new protocol:\ measurement
compression with quantum side information (MC-QSI). The information processing
task for this protocol is similar to that in measurement compression
(Section~\ref{sec:meas-comp}), with the exception that they are to perform the
protocol on the $A$ system of some bipartite state, and Bob is allowed to use
his system $B$ in order to reduce the communication resources needed for the
simulation. The protocol discussed in this section is a \textquotedblleft
feedback\textquotedblright\ simulation, in which the sender also obtains the
outcome of the measurement simulation. After reviewing the information
processing task, we state the MC-QSI\ with feedback theorem and prove
achievability of the protocol and its converse. Finally, we discuss several
applications of it.

\subsection{Information processing task for MC-QSI}

The information processing task in this case is a straightforward extension of
that for measurement compression with feedback. As such, we leave the
discussion of it to the captions of Figures~\ref{fig:ideal-ICQSI}\ and
\ref{fig:actual-ICQSI}. One point to observe from the figures is that in the ideal implementation of the measurement, the side information in system $B$ is left untouched. As a result, the measurement compression protocol will be permitted to use system $B$, but only in ways that do not significantly disturb it.

\begin{figure}[ptb]
\begin{center}
\includegraphics[
natheight=4.060300in,
natwidth=4.873200in,
height=2.9767in,
width=2.5374in
]{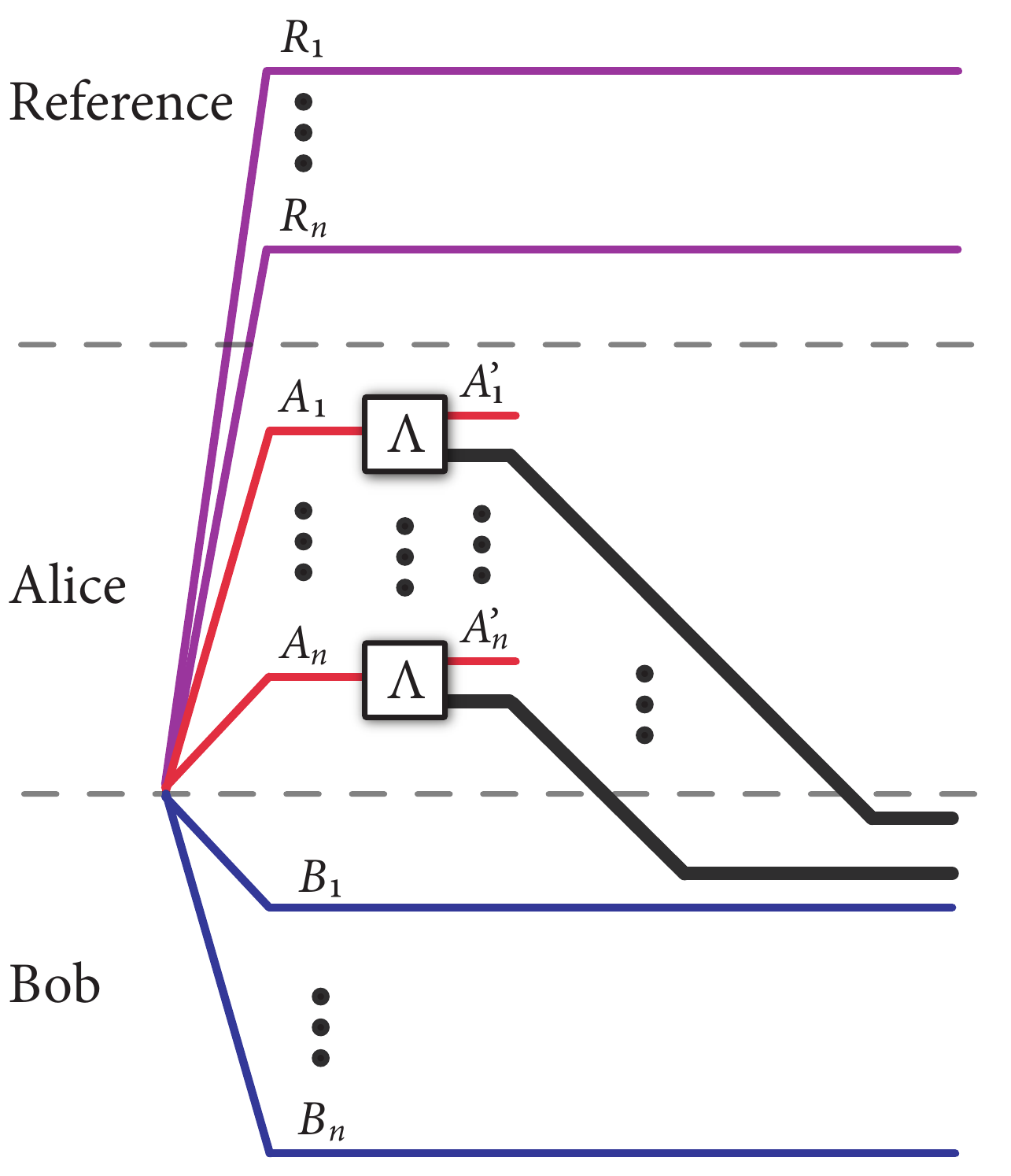}
\end{center}
\caption{\textbf{Ideal measurement compression with quantum side information.}
The ideal protocol to which we should compare performance of any actual
protocol. The sender and receiver share many copies of some bipartite state
$\rho^{AB}$. Alice performs the measurement $\Lambda\equiv\left\{  \Lambda
_{x}\right\}  $ locally on each of her shares and sends the outcomes to Bob. A
simulation of this measurement would have the sender and receiver operate
according to some procedure that is statistically indistinguishable from this
ideal case.}%
\label{fig:ideal-ICQSI}%
\end{figure}

\begin{figure}[ptb]
\begin{center}
\includegraphics[
width=2.9386in
]{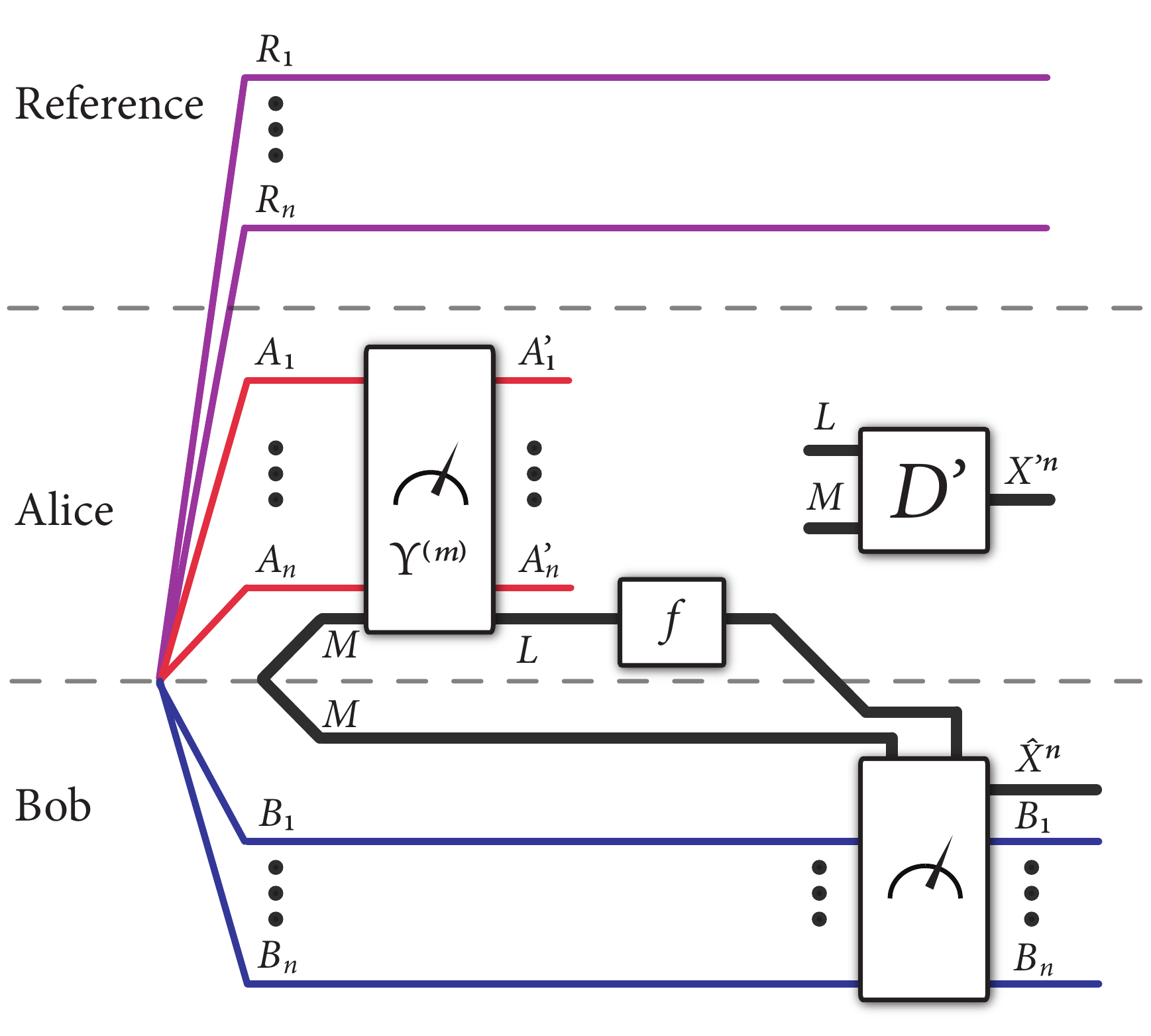}
\end{center}
\caption{\textbf{Measurement compression with quantum side information
protocol.} The figure depicts the most general protocol for this task when
both the sender and receiver are to obtain the outcome of the measurement
simulation. Assuming that Alice and Bob share many copies of some bipartite
state $\rho^{AB}$ and have common randomness $M$ available, Alice simulates
the measurement $\Lambda^{\otimes n}$ by performing some POVM\ conditional on
the value of the common randomness. Rather than send the full output $L$\ of
the measurement to Bob, Alice hashes it to $f\left(  L\right)  $ using some
hash function $f$, and she sends the hash $f\left(  L\right)  $ to Bob. Bob
performs a measurement on his systems $B^{n}$, conditional on the hash
$f\left(  L\right)  $ and his share of the common randomness $M$. From this
measurement, he can recover the full value of $L$ and then reconstruct the
sequence $x^{n}$ using $L$ and $M$. The protocol is also a \textquotedblleft
feedback\textquotedblright\ simulation, such that Alice recovers the outcome
of the simulation by processing the classical registers $L$ and $M$. The
protocol performs well if the output of this simulated measurement is
statistically indistinguishable from the output of the true measurement
$\Lambda^{\otimes n}$ (from the perspective of someone holding the reference
systems and the measurement outcomes).}%
\label{fig:actual-ICQSI}%
\end{figure}

\subsection{Measurement compression with quantum side information theorem}
\label{sec:mc-qsi}

\begin{theorem}
[Measurement compression with QSI]\label{thm:ICQSI}Let $\rho^{AB}$ be a source
state shared between a sender $A$ and a receiver $B$, and let $\Lambda$ be a
POVM\ to simulate on the $A$ system of this state. A protocol for faithful
feedback simulation of the POVM\ with classical communication
rate $R$ and common randomness rate $S$ exists if and only if the following
inequalities hold:%
\begin{align*}
R  &  \geq I\left(  X;R|B\right)  ,\\
R+S  &  \geq H\left(  X|B\right)  ,
\end{align*}
where the entropies are with respect to a state of the following form:%
\[
\sum_{x}\left\vert x\right\rangle \left\langle x\right\vert ^{X}%
\otimes\text{\emph{Tr}}_{A}\left\{  \left(  I^{R}\otimes\Lambda_{x}%
^{A}\right)  \phi^{RAB}\right\}  ,
\]
and $\phi^{RAB}$ is some purification of the state $\rho^{AB}$.
\end{theorem}

The achievable rate region closely resembles Figure~\ref{fig:IC-rate-region}, except that all of the information quantities should
be conditioned on the system $B$ since, in the new task, $B$ is available as quantum side information.

It is instructive to see how the second example of Section~\ref{sec:meas-comp-example}%
\ changes if quantum side information is available. Suppose that Bob now
possesses the purification of the maximally mixed state, so that Alice and Bob
share a Bell state before communication begins. This means that there is no
purification system $R$ because the state on $A$ and $B$ is already pure. In
this case, the state on $X$ and $B$ after the measurement in
(\ref{eq:example-POVM}) is as follows:%
\[
\frac{1}{4}\left(  \left\vert 0\right\rangle \left\langle 0\right\vert
^{X}\otimes\left\vert 0\right\rangle \left\langle 0\right\vert ^{B}+\left\vert
1\right\rangle \left\langle 1\right\vert ^{X}\otimes\left\vert 1\right\rangle
\left\langle 1\right\vert ^{B}+\left\vert 2\right\rangle \left\langle
2\right\vert ^{X}\otimes\left\vert +\right\rangle \left\langle +\right\vert
^{B}+\left\vert 3\right\rangle \left\langle 3\right\vert ^{X}\otimes\left\vert
-\right\rangle \left\langle -\right\vert ^{B}\right)  .
\]
The conditional mutual information $I\left(  X;R|B\right)  $ is zero because
the reference system is trivial. The conditional entropy $H\left(
X|RB\right)  =H\left(  X|B\right)  $ is equal to one bit. A simple
interpretation of this result is that the measurement in
(\ref{eq:example-POVM}) just requires one bit of common randomness in order to
pick the $X$ or $Z$ Pauli measurement at random. Bob then performs the
selected measurement locally, and the effect is the same as if Alice were to
perform it on her share of the state because the state is maximally entangled.

\subsection{Achievability Proof for MC-QSI}

The resource inequality corresponding to MC-QSI is as follows:%
\[
\left\langle \rho^{AB}\right\rangle +I\left(  X;R|B\right)  \left[
c\rightarrow c\right]  +H\left(  X|RB\right)  \left[  cc\right]
\geq\left\langle \Lambda^{A}\left(  \rho^{AB}\right)  \right\rangle .
\]
The meaning of this resource inequality is that the sender and receiver can
simulate the action of the POVM\ $\Lambda^{\otimes n}$ on $n$ copies of the
state $\rho^{AB}$, by exploiting $nI\left(  X;R|B\right)  $ bits of classical
communication and $nH\left(  X|RB\right)  $ bits of common randomness, and the
simulation becomes exact as $n$ becomes large.

One might think that it would be possible to concatenate the protocols of
measurement compression and CDC-QSI according to the rules of the resource
calculus \cite{DHW08} in order to have a protocol for MC-QSI. The scheme that
we develop below certainly does exploit features of both protocols, but a
direct concatenation is not possible because Alice and Bob need to exploit the
same codebook for both the measurement compression part and the CDC-QSI part
of the protocol. We note that this is similar to the way that the protocol for
channel simulation with quantum side information operates \cite{LD09}.

The basic strategy for MC-QSI is as follows. Alice simulates the measurement
on the $A^{n}$ systems of the IID\ state $\left(  \rho^{AB}\right)  ^{\otimes
n}$, with the systems $R^{n} B^{n}$ acting as a purification of $A^{n}$, by
first selecting the variable $m$ according to the common randomness shared
with Bob, and then by performing a POVM\ $\{\Upsilon_{l}^{\left(  m\right)
}\}$ chosen according to a codebook $\mathcal{C}\equiv\left\{  x^{n}\left(
l,m\right)  \right\}  $. (The codebook is of the form discussed in
Section~\ref{sec:MC-achieve}.) Alice and Bob both know the codebook
$\mathcal{C}$ used in the measurement compression strategy. Bob shares the
common randomness variable $m$ with Alice, and thus he already has this as
side information to help in determining the variable $l$. Alice hashes the
variable~$l$ according to some hash function $f$ and sends the hash. Bob
receives the hash $k\equiv f\left(  l\right)  $, and then he \textquotedblleft
scans\textquotedblright\ over all of the post-measurement states
(corresponding to codewords $x^{n}\left(  l^{\prime},m\right)  $) that are
consistent with the hash $k$ and his common randomness value $m$. We define
the set $\mathcal{A}\left(  f,k,m\right)  $ to denote the set of all such
codewords:%
\begin{equation}
\mathcal{A}\left(  f,k,m\right)  \equiv\left\{  x^{n}\left(  l,m\right)
:f\left(  l\right)  =k,\ x^{n}\left(  l,m\right)  \in\mathcal{C}\right\}  .
\label{eq:ICQSI-hash-set}%
\end{equation}
Observe that this set cannot be any larger than $\mathcal{L}$ (the set of all
possible $l$):%
\[
\left\vert \mathcal{A}\left(  f,k,m\right)  \right\vert \leq\left\vert
\mathcal{L}\right\vert =2^{n\left[  I\left(  X;RB\right)  +3\delta\right]  }.
\]

The intuition behind the protocol is that measurement compression proceeds as
before using $nI\left(  X;BR\right)  $ bits for the outcome of the measurement
and $nH\left(  X|RB\right)  $ for the common randomness, because the systems
$RB$ act as a purification for $A$. But in this case, Bob has the quantum
systems $B^{n}$ available and should be able to determine $nI\left(
X;B\right)  $ bits about $X^{n}$ by performing a collective measurement on his
systems (following from the HSW theorem \cite{Hol98,PhysRevA.56.131}). So, Alice should
only need to send the difference of these amounts, $n\left(  I\left(
X;BR\right)  -I\left(  X;B\right)  \right)  =nI\left(  X;R|B\right)  $, to Bob.

\textbf{Detailed Strategy.} The encoding strategy for this scenario is as
follows. Alice and Bob are allowed to have an agreed-upon hash function
$f:\mathcal{L}\rightarrow\mathcal{K}$, selected at random from a two-universal
family (as described in Section~\ref{sec:CDC-QSI-achieve}). Alice's message to
Bob will be an element of $\mathcal{K}$. The collision
probability for some $l\neq l^{\prime}$ in this case is as follows:%
\[
\Pr_{f}\left\{  f\left(  l\right)  =f\left(  l^{\prime}\right)  \right\}
\leq\frac{1}{\left\vert \mathcal{K}\right\vert }=2^{-nR}.
\]

Upon receiving the hash value $k$ and having a particular value $m$ for the
common randomness, Bob performs a sequence of binary measurements $\left\{
\Pi_{x^{n}\left(  l,m\right)  },I-\Pi_{x^{n}\left(  l,m\right)  }\right\}  $
for all the codewords $x^{n}\left(  l,m\right)  \in\mathcal{C}$ that are
consistent with the hash value (so that $f\left(  l\right)  =k$). Note that
$\Pi_{x^{n}\left(  l,m\right)  }$ is a conditionally typical projector for the
tensor-product state $\rho_{x^{n}\left(  l,m\right)  }^{B^{n}}$, the
conditional state on Bob's system after performing the ideal measurement
$\Lambda^{\otimes n}$ and receiving the outcome $x^{n}\left(  l,m\right)  $.
Recall from (\ref{eq:ICQSI-hash-set}) that $\mathcal{A}\left(  f,k,m\right)  $
is the set of all such codewords. In the following, we will show that, by
choosing
\begin{align}
\left\vert \mathcal{L}\right\vert  &  = 2^{n\left[  I\left(  X;RB\right)
+3\delta\right]  },\label{eq_L4}\\
\left\vert \mathcal{M}\right\vert  &  = 2^{n\left[  H\left(  X|RB\right)
+\delta\right]  },\label{eq_M4}\\
\left\vert \mathcal{K}\right\vert  &  = 2^{n\left[  I\left(  X;R|B\right)
+11\delta\right]  },
\end{align}
the error probability will approach zero as $n$ goes to infinity.

\textbf{Error Analysis.} The error probability for this decoder is then as
follows:%
\begin{equation}
\Pr\left\{  \text{\textquotedblleft error @ decoder\textquotedblright%
}\right\}  =\frac{1}{\left\vert \mathcal{M}\right\vert }\sum_{l,m}q\left(
x^{n}\left(  l,m\right)  \right)  \ \Pr\left\{  \text{\textquotedblleft error
@ decoder\textquotedblright\ }|\ l,m\right\}  , \label{eq:ICQSI-err-prob}%
\end{equation}
where%
\[
q\left(  x^{n}\left(  l,m\right)  \right)  =\text{Tr}\left\{  \left(
(\Upsilon_{l}^{\left(  m\right)  })^{A^{n}}\otimes I^{B^{n}}\right)  \left(
\rho^{AB}\right)  ^{\otimes n}\right\}
\]
is the probability of receiving outcome $l$ when performing the simulated
measurement in (\ref{eq:POVM-comp-meas-ops}). The post-measurement states on $B^n$ for
the POVM $\{ \Upsilon_l^{(m)} \}$ are as follows:%
\[
\widetilde{\rho}_{x^{n}\left(  l,m\right)  }^{B^{n}}\equiv\frac{1}{q\left(
x^{n}\left(  l,m\right)  \right)  }\text{Tr}_{A^{n}}\left\{  (\Upsilon
_{l}^{\left(  m\right)  })^{A^{n}}\left(  \rho^{AB}\right)  ^{\otimes
n}\right\}  .
\]
Note that the probability masses $q\left(  x^{n}\left(  l,m\right)  \right)  $
and $q\left(  x^{n}\left(  l^{\prime},m^{\prime}\right)  \right)  $ and the
states $\widetilde{\rho}_{x^{n}\left(  l,m\right)  }^{B^{n}}$ and
$\widetilde{\rho}_{x^{n}\left(  l^{\prime},m^{\prime}\right)  }^{B^{n}}$ are
equivalent, respectively, if two different codewords have the same value
(i.e., if $l\neq l^{\prime}$ or $m\neq m^{\prime}$ but $x^{n}\left(
l,m\right)  =x^{n}\left(  l^{\prime},m^{\prime}\right)  $, then $q\left(
x^{n}\left(  l,m\right)  \right)  =q\left(  x^{n}\left(  l^{\prime},m^{\prime
}\right)  \right)  $ and $\widetilde{\rho}_{x^{n}\left(  l,m\right)  }^{B^{n}%
}=\widetilde{\rho}_{x^{n}\left(  l^{\prime},m^{\prime}\right)  }^{B^{n}}%
$ --- this is due to the way that we choose the measurement operators
$\Upsilon_{l}^{\left(  m\right)  }$ in (\ref{eq:POVM-comp-meas-ops}) for the
measurement simulation).

Now we consider the error term $\Pr\left\{  \text{\textquotedblleft error @
decoder\textquotedblright\ }|\ l,m\right\}  $. Let $a_{1}^{\left(  km\right)
}$, \ldots, $a_{\left\vert \mathcal{A}\right\vert }^{\left(  km\right)  }$
enumerate all of the codewords $x^{n}\left(  l,m\right)  $ in the set
$\mathcal{A}\left(  f,k,m\right)  $ defined in (\ref{eq:ICQSI-hash-set})
(those codewords $x^{n}\left(  l,m\right)  $ consistent with the hash $k$).
Let $a_{j}^{\left(  km\right)  }$ denote the actual codeword $x^{n}\left(
l,m\right)  $ produced by the simulated measurement. The probability for a
correct decoding for Bob is as follows:%
\[
\text{Tr}\left\{  \Pi_{a_{j}^{\left(  km\right)  }}\ \hat{\Pi}_{a_{j-1}%
^{\left(  km\right)  }}\cdots\hat{\Pi}_{a_{1}^{\left(  km\right)  }%
}\ \widetilde{\rho}_{a_{j}^{\left(  km\right)  }}^{B^{n}}\ \hat{\Pi}%
_{a_{1}^{\left(  km\right)  }}\cdots\hat{\Pi}_{a_{j-1}^{\left(  km\right)  }%
}\ \Pi_{a_{j}^{\left(  km\right)  }}\right\}  ,
\]
so that the binary tests give a response of \textquotedblleft
no\textquotedblright\ until the test for $a_{j}^{\left(  km\right)  }$ gives a
response of \textquotedblleft yes.\textquotedblright\ Then the probability for
incorrectly decoding is%
\[
1-\text{Tr}\left\{  \Pi_{a_{j}^{\left(  km\right)  }}\ \hat{\Pi}%
_{a_{j-1}^{\left(  km\right)  }}\cdots\hat{\Pi}_{a_{1}^{\left(  km\right)  }%
}\ \widetilde{\rho}_{a_{j}^{\left(  km\right)  }}^{B^{n}}\ \hat{\Pi}%
_{a_{1}^{\left(  km\right)  }}\cdots\hat{\Pi}_{a_{j-1}^{\left(  km\right)  }%
}\ \Pi_{a_{j}^{\left(  km\right)  }}\right\}  ,
\]
so that we can write the error probability in (\ref{eq:ICQSI-err-prob}) as
follows (for this decoding strategy):%
\begin{equation}
\sum_{\substack{m\in\mathcal{M},\\k\in\mathcal{K}}}\sum_{a_{j}^{\left(
km\right)  }\in\mathcal{A}\left(  k,f,m\right)  }\frac{1}{\left\vert
\mathcal{M}\right\vert }q\left(  a_{j}^{\left(  km\right)  }\right)  \ \left[
1-\text{Tr}\left\{  \Pi_{a_{j}^{\left(  km\right)  }}\ \hat{\Pi}%
_{a_{j-1}^{\left(  km\right)  }}\cdots\hat{\Pi}_{a_{1}^{\left(  km\right)  }%
}\ \widetilde{\rho}_{a_{j}^{\left(  km\right)  }}^{B^{n}}\ \hat{\Pi}%
_{a_{1}^{\left(  km\right)  }}\cdots\hat{\Pi}_{a_{j-1}^{\left(  km\right)  }%
}\ \Pi_{a_{j}^{\left(  km\right)  }}\right\}  \right]  .
\end{equation}
Observe that the above error probability is equal to%
\begin{equation}
\frac{1}{\left\vert \mathcal{M}\right\vert }\sum_{l,m}q\left(  x^{n}\left(
l,m\right)  \right)  \ \text{Tr}\left\{  \left(  I-\Theta_{x^{n}\left(
l,m\right)  }\right)  \ \widetilde{\rho}_{x^{n}\left(  l,m\right)  }^{B^{n}%
}\ \right\}  , \label{eq:ICQSI-err-prob-2}%
\end{equation}
if we define the POVM element$\ \Theta_{x^{n}\left(  l,m\right)  }$ as%
\[
\Theta_{x^{n}\left(  l,m\right)  }\equiv\hat{\Pi}_{a_{1}^{\left(  km\right)
}}\cdots\hat{\Pi}_{a_{j-1}^{\left(  km\right)  }}\ \Pi_{a_{j}^{\left(
km\right)  }}\ \hat{\Pi}_{a_{j-1}^{\left(  km\right)  }}\cdots\hat{\Pi}%
_{a_{1}^{\left(  km\right)  }},
\]
where we recall that $a_{j}^{\left(  km\right)  }=x^{n}\left(  l,m\right)  $.

Now, we can further express the error probability in
(\ref{eq:ICQSI-err-prob-2}) as follows, by employing an indicator function:%
\begin{align}
&  =\sum_{x^{n}\in\mathcal{X}^{n}}\frac{1}{\left\vert \mathcal{M}\right\vert
}\sum_{l,m}q\left(  x^{n}\left(  l,m\right)  \right)  \ \mathcal{I}\left(
x^{n}=x^{n}\left(  l,m\right)  \right)  \ \text{Tr}\left\{  \left(
I-\Theta_{x^{n}\left(  l,m\right)  }\right)  \ \widetilde{\rho}_{x^{n}\left(
l,m\right)  }^{B^{n}}\right\} \nonumber\\
&  \leq\sum_{x^{n}\in\mathcal{X}^{n}}\frac{1}{\left\vert \mathcal{M}%
\right\vert }\sum_{l,m}q\left(  x^{n}\left(  l,m\right)  \right)
\ \mathcal{I}\left(  x^{n}=x^{n}\left(  l,m\right)  \right)  \ \text{Tr}%
\left\{  \left(  I-\Theta_{x^{n}}^{\prime}\right)  \ \widetilde{\rho}%
_{x^{n}\left(  l,m\right)  }^{B^{n}}\right\}  ,
\label{eq:ICQSI-err-prob-up-bnd}%
\end{align}
where we define $\Theta_{x^{n}}^{\prime}$ as a POVM\ element corresponding to
a worst-case decoding over the states $\widetilde{\rho}_{x^{n}\left(
l,m\right)  }^{B^{n}}$ with the same codeword value $x^{n}$:%
\[
\Theta_{x^{n}}^{\prime}\equiv\arg\max_{\substack{\Theta_{x^{n}\left(
l,m\right)  }\ :\\x^{n}=x^{n}\left(  l,m\right)  }}\text{Tr}\left\{  \left(
I-\Theta_{x^{n}\left(  l,m\right)  }\right)  \ \widetilde{\rho}_{x^{n}\left(
l,m\right)  }^{B^{n}}\right\}  .
\]
Let $\mathcal{C}^{\prime}$ be a pruning of the original codebook $\mathcal{C}$
containing no duplicate entries (it contains only the codewords with
worst-case error probabilities as given above). Then the last line in
(\ref{eq:ICQSI-err-prob-up-bnd}) is equivalent to the following one:%
\begin{equation}
\sum_{x^{n}\in\mathcal{C}^{\prime}}\text{Tr}\left\{  \left(  I-\Theta_{x^{n}%
}^{\prime}\right)  \ \frac{1}{\left\vert \mathcal{M}\right\vert }\sum
_{l,m}q\left(  x^{n}\left(  l,m\right)  \right)  \ \mathcal{I}\left(
x^{n}=x^{n}\left(  l,m\right)  \right)  \ \widetilde{\rho}_{x^{n}\left(
l,m\right)  }^{B^{n}}\right\}  . \label{eq:ICQSI-err-prob-3}%
\end{equation}

This decoding scheme will only work well if the states $\widetilde{\rho
}_{x^{n}\left(  l,m\right)  }^{B^{n}}$ are close to the tensor product states
$\rho_{x^{n}\left(  l,m\right)  }^{B^{n}}$ that would result from the ideal
measurement. We expect that this should hold if the measurement compression
part of the protocol is successful, and we prove this in detail in what follows.

The quantity characterizing a faithful measurement simulation in
(\ref{eq:meas-comp-faithful-cond}) is equivalent to the following (one can
show this by exploiting the definitions of the measurement maps and the
post-measurement states given above, after tracing out the reference systems):%
\begin{align*}
\Delta\left(  \mathcal{C}\right)   &  \equiv\sum_{x^{n}\in\mathcal{X}^{n}%
}\left\Vert p_{X^{n}}\left(  x^{n}\right)  \rho_{x^{n}}^{B^{n}}-\frac
{1}{\left\vert \mathcal{M}\right\vert }\sum_{l,m}q\left(  x^{n}\left(
l,m\right)  \right)  \ \mathcal{I}\left(  x^{n}=x^{n}\left(  l,m\right)
\right)  \ \widetilde{\rho}_{x^{n}\left(  l,m\right)  }^{B^{n}}\right\Vert
_{1}\\
&  =\sum_{x^{n}\notin\mathcal{C}^{\prime}}\left\Vert p_{X^{n}}\left(
x^{n}\right)  \rho_{x^{n}}^{B^{n}}\right\Vert _{1}+\sum_{x^{n}\in
\mathcal{C}^{\prime}}\left\Vert p_{X^{n}}\left(  x^{n}\right)  \rho_{x^{n}%
}^{B^{n}}-\frac{1}{\left\vert \mathcal{M}\right\vert }\sum_{l,m}q\left(
x^{n}\left(  l,m\right)  \right)  \ \mathcal{I}\left(  x^{n}=x^{n}\left(
l,m\right)  \right)  \ \widetilde{\rho}_{x^{n}\left(  l,m\right)  }^{B^{n}%
}\right\Vert _{1},
\end{align*}
where $\mathcal{C}^{\prime}$ is the pruned codebook containing no duplicate
entries (observe that the indicator function $\mathcal{I}\left(  x^{n}%
=x^{n}\left(  l,m\right)  \right)  $ captures all of the duplicates). Applying
the trace inequality Tr$\left\{  \Lambda\sigma\right\}  \leq\ $Tr$\left\{
\Lambda\rho\right\}  +\left\Vert \rho-\sigma\right\Vert _{1}$ from
Lemma~\ref{lem:trace-inequality}\ to the expression in
(\ref{eq:ICQSI-err-prob-3}), we then obtain the following upper bound on it:%
\begin{multline*}
\leq\sum_{x^{n}\in\mathcal{C}^{\prime}}\text{Tr}\left\{  \left(
I-\Theta_{x^{n}}^{\prime}\right)  \ p_{X^{n}}\left(  x^{n}\right)  \rho
_{x^{n}}^{B^{n}}\right\} \\
+\sum_{x^{n}\in\mathcal{C}^{\prime}}\left\Vert p_{X^{n}}\left(  x^{n}\right)
\rho_{x^{n}}^{B^{n}}-\frac{1}{\left\vert \mathcal{M}\right\vert }\sum
_{l,m}q\left(  x^{n}\left(  l,m\right)  \right)  \ \mathcal{I}\left(
x^{n}=x^{n}\left(  l,m\right)  \right)  \ \widetilde{\rho}_{x^{n}\left(
l,m\right)  }^{B^{n}}\right\Vert _{1}\\
\leq\sum_{x^{n}\in\mathcal{C}^{\prime}}p_{X^{n}}\left(  x^{n}\right)
\text{Tr}\left\{  \left(  I-\Theta_{x^{n}}^{\prime}\right)  \ \rho_{x^{n}%
}^{B^{n}}\right\}  +\Delta\left(  \mathcal{C}\right)  .
\end{multline*}
We can now focus on bounding the term on the LHS of the last line above.
Expanding it again leads to%
\begin{align*}
&  \sum_{x^{n}\in\mathcal{C}^{\prime}}p_{X^{n}}\left(  x^{n}\right)
\text{Tr}\left\{  \left(  I-\Theta_{x^{n}}^{\prime}\right)  \ \rho_{x^{n}%
}^{B^{n}}\right\} \\
&  =\sum_{\substack{m\in\mathcal{M},\\k\in\mathcal{K}}}\ \sum_{a_{j}^{\left(
km\right)  }\in\mathcal{A}^{\prime}\left(  k,f,m\right)  }p_{X^{n}}\left(
a_{j}^{\left(  km\right)  }\right)  \text{Tr}\left\{  \left(  I-\Theta
_{a_{j}^{\left(  km\right)  }}^{\prime}\right)  \ \rho_{a_{j}^{\left(
km\right)  }}^{B^{n}}\right\} \\
&  =\sum_{\substack{m\in\mathcal{M},\\k\in\mathcal{K}}}\ \sum_{a_{j}^{\left(
km\right)  }\in\mathcal{A}^{\prime}\left(  k,f,m\right)  }p_{X^{n}}\left(
a_{j}^{\left(  km\right)  }\right)  \left[  1-\text{Tr}\left\{  \Pi
_{a_{j}^{\left(  km\right)  }}\ \hat{\Pi}_{a_{j-1}^{\left(  km\right)  }%
}\cdots\hat{\Pi}_{a_{1}^{\left(  km\right)  }}\ \rho_{a_{j}^{\left(
km\right)  }}^{B^{n}}\ \hat{\Pi}_{a_{1}^{\left(  km\right)  }}\cdots\hat{\Pi
}_{a_{j-1}^{\left(  km\right)  }}\ \Pi_{a_{j}^{\left(  km\right)  }}\right\}
\right]  ,
\end{align*}
where%
\[
\mathcal{A}^{\prime}\left(  f,k,m\right)  \equiv\left\{  x^{n}\left(
l,m\right)  :f\left(  l\right)  =k,\ \ x^{n}\left(  l,m\right)  \in
\mathcal{C}^{\prime}\right\}  .
\]
We can then insert the average state typical projector as we did before in
(\ref{eq:avg-typ-bound}), in order to bound the last line from above as%
\begin{multline*}
\sum_{\substack{m\in\mathcal{M},\\k\in\mathcal{K}}}\ \sum_{a_{j}^{\left(
km\right)  }\in\mathcal{A}^{\prime}\left(  k,f,m\right)  }p_{X^{n}}\left(
a_{j}^{\left(  km\right)  }\right)  \times\\
\left[  \text{Tr}\left\{  \Pi\ \rho_{a_{j}^{\left(  km\right)  }}^{B^{n}}%
\ \Pi\right\}  -\text{Tr}\left\{  \Pi_{a_{j}^{\left(  km\right)  }}\ \hat{\Pi
}_{a_{j-1}^{\left(  km\right)  }}\cdots\hat{\Pi}_{a_{1}^{\left(  km\right)  }%
}\ \Pi\ \rho_{a_{j}^{\left(  km\right)  }}^{B^{n}}\ \Pi\ \hat{\Pi}%
_{a_{1}^{\left(  km\right)  }}\cdots\hat{\Pi}_{a_{j-1}^{\left(  km\right)  }%
}\ \Pi_{a_{j}^{\left(  km\right)  }}\right\}  \right]  .
\end{multline*}
The error accumulated in doing so is $\epsilon+2\sqrt{\epsilon}$ as before. At
this point, we can apply Sen's non-commutative union bound
(Lemma~\ref{lem-non-com-union-bound} in Appendix~\ref{sec:useful-lemmas}) and
concavity of square root to obtain the upper bound%
\[
2\sqrt{\sum_{\substack{m\in\mathcal{M},\\k\in\mathcal{K}}}\ \sum
_{a_{j}^{\left(  km\right)  }\in\mathcal{A}^{\prime}\left(  k,f,m\right)
}p_{X^{n}}\left(  a_{j}^{\left(  km\right)  }\right)  \left[  \text{Tr}%
\left\{  \left(  I-\Pi_{a_{j}^{\left(  km\right)  }}\right)  \Pi\ \rho
_{a_{j}^{\left(  km\right)  }}^{B^{n}}\ \Pi\right\}  +\sum_{i=1}%
^{j-1}\text{Tr}\left\{  \Pi_{a_{i}^{\left(  km\right)  }}\ \Pi\ \rho
_{a_{j}^{\left(  km\right)  }}^{B^{n}}\ \Pi\right\}  \right]  }.
\]
The first term inside the square root we can bound from above by
$\epsilon+2\sqrt{\epsilon}$ as we did before in
(\ref{eq:CDC-QSI-first-err-bnd}), using properties of quantum typicality. We
continue bounding the second term as follows:%
\begin{align*}
&  \sum_{\substack{m\in\mathcal{M},\\k\in\mathcal{K}}}\ \sum_{a_{j}^{\left(
km\right)  }\in\mathcal{A}^{\prime}\left(  k,f,m\right)  }p_{X^{n}}\left(
a_{j}^{\left(  km\right)  }\right)  \sum_{i=1}^{j-1}\text{Tr}\left\{
\Pi_{a_{i}^{\left(  km\right)  }}\ \Pi\ \rho_{a_{j}^{\left(  km\right)  }%
}^{B^{n}}\ \Pi\right\} \\
&  \leq\sum_{\substack{m\in\mathcal{M},\\k\in\mathcal{K}}}\ \sum
_{a_{j}^{\left(  km\right)  }\in\mathcal{A}^{\prime}\left(  k,f,m\right)
}p_{X^{n}}\left(  a_{j}^{\left(  km\right)  }\right)  \sum_{i\neq
j\ :\ a_{i}^{\left(  km\right)  }\in\mathcal{A}\left(  k,f,m\right)
}\text{Tr}\left\{  \Pi_{a_{i}^{\left(  km\right)  }}\ \Pi\ \rho_{a_{j}%
^{\left(  km\right)  }}^{B^{n}}\ \Pi\right\} \\
&  =\sum_{l,m\ :\ x^{n}\left(  l,m\right)  \in\mathcal{C}^{\prime}}p_{X^{n}%
}\left(  x^{n}\left(  l,m\right)  \right)  \sum_{l^{\prime}\in\mathcal{L}%
\ :\ l^{\prime}\neq l}\mathcal{I}\left(  f\left(  l\right)  =f\left(
l^{\prime}\right)  \right)  \text{Tr}\left\{  \Pi_{x^{n}\left(  l^{\prime
},m\right)  }\ \Pi\ \rho_{x^{n}\left(  l,m\right)  }^{B^{n}}\ \Pi\right\}  ,
\end{align*}
where the two steps follow by including all indices in the sum not equal to
$j$ and rewriting the sum with indicator functions. We now take an expectation
with respect to the random hash (realizing that we could have done this the
whole time):%
\begin{align}
&  \mathbb{E}_{f}\left\{  \sum_{l,m\ :\ x^{n}\left(  l,m\right)
\in\mathcal{C}^{\prime}}p_{X^{n}}\left(  x^{n}\left(  l,m\right)  \right)
\sum_{l^{\prime}\in\mathcal{L}\ :\ l^{\prime}\neq l}\mathcal{I}\left(
f\left(  l\right)  =f\left(  l^{\prime}\right)  \right)  \text{Tr}\left\{
\Pi_{x^{n}\left(  l^{\prime},m\right)  }\ \Pi\ \rho_{x^{n}\left(  l,m\right)
}^{B^{n}}\ \Pi\right\}  \right\} \nonumber\\
&  =\sum_{l,m\ :\ x^{n}\left(  l,m\right)  \in\mathcal{C}^{\prime}}p_{X^{n}%
}\left(  x^{n}\left(  l,m\right)  \right)  \sum_{l^{\prime}\in\mathcal{L}%
\ :\ l^{\prime}\neq l}\mathbb{E}_{f}\left\{  \mathcal{I}\left(  f\left(
l\right)  =f\left(  l^{\prime}\right)  \right)  \right\}  \text{Tr}\left\{
\Pi_{x^{n}\left(  l^{\prime},m\right)  }\ \Pi\ \rho_{x^{n}\left(  l,m\right)
}^{B^{n}}\ \Pi\right\} \nonumber\\
&  =\sum_{l,m\ :\ x^{n}\left(  l,m\right)  \in\mathcal{C}^{\prime}}p_{X^{n}%
}\left(  x^{n}\left(  l,m\right)  \right)  \sum_{l^{\prime}\in\mathcal{L}%
\ :\ l^{\prime}\neq l}\Pr_{f}\left\{  f\left(  l\right)  =f\left(  l^{\prime
}\right)  \right\}  \text{Tr}\left\{  \Pi_{x^{n}\left(  l^{\prime},m\right)
}\ \Pi\ \rho_{x^{n}\left(  l,m\right)  }^{B^{n}}\ \Pi\right\} \nonumber\\
&  \leq2^{-nR}\ \sum_{l,m\ :\ x^{n}\left(  l,m\right)  \in\mathcal{C}^{\prime
}}p_{X^{n}}\left(  x^{n}\left(  l,m\right)  \right)  \sum_{l^{\prime}%
\in\mathcal{L}\ :\ l^{\prime}\neq l}\text{Tr}\left\{  \Pi_{x^{n}\left(
l^{\prime},m\right)  }\ \Pi\ \rho_{x^{n}\left(  l,m\right)  }^{B^{n}}%
\ \Pi\right\} \nonumber\\
&  \leq2^{-nR}\ \sum_{l,m}\sum_{l^{\prime}\in\mathcal{L}\ :\ l^{\prime}\neq
l}\text{Tr}\left\{  \Pi_{x^{n}\left(  l^{\prime},m\right)  }\ \Pi\ p_{X^{n}%
}\left(  x^{n}\left(  l,m\right)  \right)  \rho_{x^{n}\left(  l,m\right)
}^{B^{n}}\ \Pi\right\} \nonumber\\
&  \leq2^{-nR}\ 2^{-n\left[  H\left(  X\right)  -\delta\right]  }\ \sum
_{l,m}\sum_{l^{\prime}\in\mathcal{L}\ :\ l^{\prime}\neq l}\text{Tr}\left\{
\Pi_{x^{n}\left(  l^{\prime},m\right)  }\ \Pi\ \rho_{x^{n}\left(  l,m\right)
}^{B^{n}}\ \Pi\right\}  . \label{eq:almost-done-ICQSI}%
\end{align}
The first inequality follows from the two-universal hashing property. The
second inequality follows from summing over all of the codewords in
$\mathcal{C}$, not just the non-duplicate entries in $\mathcal{C}^{\prime}$.
The third inequality follows because all of the sequences in $x^{n}\left(
l,m\right)  $ are chosen to be strongly typical (recall the construction in
Section~\ref{sec:MC-achieve}) and by upper bounding their probabilities by
$2^{-n\left[  H\left(  X\right)  -\delta\right]  }$. From here, we exploit the
fact that the codewords $x^{n}\left(  l,m\right)  $ were chosen randomly as
specified in Section~\ref{sec:MC-achieve}. So we now consider $X^{n}\left(
l,m\right)  $ as random variables and take the expectation with respect to
them (realizing again that we could have done this the whole time and focusing
on the rightmost term above):%
\begin{align*}
&  \mathbb{E}_{X^{n}}\left\{  \sum_{l,m}\sum_{l^{\prime}\in\mathcal{L}%
\ :\ l^{\prime}\neq l}\text{Tr}\left\{  \Pi_{X^{n}\left(  l^{\prime},m\right)
}\ \Pi\ \rho_{X^{n}\left(  l,m\right)  }^{B^{n}}\ \Pi\right\}  \right\} \\
&  =\sum_{l,m}\sum_{l^{\prime}\in\mathcal{L}\ :\ l^{\prime}\neq l}%
\text{Tr}\left\{  \mathbb{E}_{X^{n}}\left\{  \Pi_{X^{n}\left(  l^{\prime
},m\right)  }\right\}  \ \Pi\ \mathbb{E}_{X^{n}}\left\{  \rho_{X^{n}\left(
l,m\right)  }^{B^{n}}\right\}  \ \Pi\right\} \\
&  =\sum_{l,m}\sum_{l^{\prime}\in\mathcal{L}\ :\ l^{\prime}\neq l}%
\text{Tr}\left\{  \mathbb{E}_{X^{n}}\left\{  \Pi_{x^{n}\left(  l^{\prime
},m\right)  }\right\}  \ \Pi\ \sum_{x^{n}}p_{X^{\prime n}}\left(
x^{n}\right)  \rho_{x^{n}}^{B^{n}}\ \Pi\right\} \\
&  \leq\left[  1-\epsilon\right]  ^{-1}2^{-n\left[  H\left(  B\right)
-\delta\right]  }\sum_{l,m}\sum_{l^{\prime}\in\mathcal{L}\ :\ l^{\prime}\neq
l}\text{Tr}\left\{  \mathbb{E}_{X^{n}}\left\{  \Pi_{x^{n}\left(  l^{\prime
},m\right)  }\right\}  \Pi\right\}  .
\end{align*}
The first equality follows because the indices $l^{\prime}$ and $l$ are
different, implying that the random variables $X^{n}\left(  l^{\prime
},m\right)  $ and $X^{n}\left(  l,m\right)  $ are independent so that we can
distribute the expectation. The inequality follows by applying the operator
inequality (\ref{eq:prune-avg-op-ineq}) from Appendix~\ref{sec:typ-review} and
properties of quantum typicality. Continuing, we have%
\begin{align*}
&  \leq\left[  1-\epsilon\right]  ^{-1}\ 2^{-n\left[  H\left(  B\right)
-\delta\right]  }\ \sum_{l,m}\sum_{l^{\prime}\in\mathcal{L}\ :\ l^{\prime}\neq
l}\mathbb{E}_{X^{n}}\left\{  \text{Tr}\left\{  \Pi_{x^{n}\left(  l^{\prime
},m\right)  }\right\}  \right\} \\
&  \leq\left[  1-\epsilon\right]  ^{-1}\ 2^{-n\left[  H\left(  B\right)
-\delta\right]  }\ \sum_{l,m}\sum_{l^{\prime}\in\mathcal{L}\ :\ l^{\prime}\neq
l}2^{n\left[  H\left(  B|X\right)  +\delta\right]  }\\
&  \leq\left[  1-\epsilon\right]  ^{-1}\ 2^{-n\left[  H\left(  B\right)
-\delta\right]  }\ 2^{n\left[  H\left(  B|X\right)  +\delta\right]
}\ \left\vert \mathcal{L}\times\mathcal{M}\right\vert \left\vert
\mathcal{L}\right\vert \\
&  \leq\left[  1-\epsilon\right]  ^{-1}\ 2^{-n\left[  H\left(  B\right)
-\delta\right]  }\ 2^{n\left[  H\left(  B|X\right)  +\delta\right]
}\ 2^{n\left[  H\left(  X\right)  +4\delta\right]  }\ 2^{n\left[  I\left(
X;RB\right)  +3\delta\right]  }.
\end{align*}
The first inequality is from $\Pi\leq I$ and the second is from the bound
Tr$\{\Pi_{x^{n}\left(  l^{\prime},m\right)  }\}\leq2^{n\left[  H\left(
B|X\right)  +\delta\right]  }$. The final inequality follows from the
selection for the sizes of $\mathcal{L}$ and $\mathcal{M}$ in (\ref{eq_L4}%
-\ref{eq_M4}). Combining this bound with the one in
(\ref{eq:almost-done-ICQSI}), our final upper bound is%
\[
\left[  1-\epsilon\right]  ^{-1}\ 2^{-n\left[  R-I\left(  X;R|B\right)
-10\delta\right]  }.
\]
Collecting everything together, we arrive at the following upper bound on the
decoding error probability in (\ref{eq:ICQSI-err-prob}):%
\[
\epsilon^{\prime\prime\prime}\equiv\Delta\left(  \mathcal{C}\right)
+\epsilon+2\sqrt{\epsilon}+2\sqrt{\epsilon+2\sqrt{\epsilon}+\left[
1-\epsilon\right]  ^{-1}\ 2^{-n\left[  R-I\left(  X;R|B\right)  -10\delta
\right]  }}.
\]
Thus, as along as we choose $R=I\left(  X;R|B\right)  +11\delta$, the
expectation of this error with respect to the hash function and the random
choice of code vanishes in the asymptotic limit.

We now complete our achievability proof by demonstrating that there exists a
choice of the $\left\{  X^{n}\left(  l,m\right)  \right\}  $ codewords such
that the measurement simulation error and the decoding error become
arbitrarily small. Let $F$ be the event that the decoding error probability is
less than $\sqrt{\epsilon^{\prime\prime\prime}}$. Then we have the following
upper bound on the complement of this event by invoking Markov's inequality:%
\[
\Pr\left\{  F^{c}\right\}  \leq\frac{\mathbb{E}_{\mathcal{C},f}\left\{
\text{\textquotedblleft decoding\ error\textquotedblright}\right\}  }%
{\sqrt{\epsilon^{\prime\prime\prime}}}\leq\sqrt{\epsilon^{\prime\prime\prime}%
}.
\]
Thus, by choosing $\left\vert \mathcal{L}\right\vert $, $\left\vert
\mathcal{M}\right\vert $, and $\left\vert \mathcal{K}\right\vert $
appropriately, we can have all of the events $E_{m}$, $E_{0}$, and $F$ be true
for some choice of the codebook $\left\{  x^{n}\left(  l,m\right)  \right\}  $
and the hash $f$ (similar to the \textquotedblleft union
bound\textquotedblright\ argument in (\ref{eq:union-bound-arg})), so that both
the measurement simulation error and the decoding error probability are
arbitrarily small for sufficiently large $n$.

After determining the sequence $x^{n}$ resulting from the measurement
simulation, Bob can place it in a classical register. By using the fact that
the measurement simulation and the decoding are successful and employing an
argument similar to that at the end of Section~\ref{sec:CDC-QSI-achieve}, we
know that the disturbance of the state is asymptotically negligible, so that
the condition in (\ref{eq:ICQSI-good-protocol}) for a good protocol is satisfied.

A proof similar to that in Theorem~\ref{thm:instr-comp}\ implies that Alice
and Bob can exploit quantum side information and a protocol similar to the
above to simulate a general quantum instrument, in such a way that Alice
possesses the quantum output and Bob obtains the classical output. The
resulting resource inequality is stated in (\ref{eq:inst-comp-QSI-res-in}) below.

\subsection{Converse for measurement compression with QSI}

The converse proof for measurement compression with QSI\ demonstrates the
optimality of the protocol from the previous section. Specifically, it shows that the
single-letter rates in Theorem~\ref{thm:ICQSI}\ are optimal for the case of a
feedback simulation.

The most general protocol for this task has Alice combine her shares $A^{n}%
$\ of the state with her share of the common randomness $M$ and perform some
quantum operation with quantum outputs $A^{\prime n}$ and classical output
$L$. Alice then processes this variable $L$ to produce another random variable $L'$, 
which she sends to Bob
over some noiseless classical bit channels. Bob feeds $L^{\prime}$, his share
of the common randomness, and his systems $B^{n}$ into some quantum operation
with classical outputs $\hat{X}^{n}$ and quantum outputs $B^{\prime n}$. If
the protocol is any good for this task, then the actual state $\omega
^{R^{n}A^{\prime n}X^{\prime n}\hat{X}^{n}B^{\prime n}}$ should be $\epsilon
$-close in trace distance to the ideal state $\sigma^{R^{n}A^{\prime n}%
X^{n}\overline{X}^{n}B^{n}}$, where $\overline{X}^{n}$ is a copy of the
variable $X^{n}$:%
\begin{equation}
\left\Vert \omega^{R^{n}A^{\prime n}X^{\prime n}\hat{X}^{n}B^{\prime n}%
}-\sigma^{R^{n}A^{\prime n}X^{n}\overline{X}^{n}B^{n}}\right\Vert _{1}%
\leq\epsilon. \label{eq:ICQSI-good-protocol}%
\end{equation}
A proof for the first bound in Theorem~\ref{thm:ICQSI} goes as follows:%
\begin{align*}
nR  &  \geq H\left(  L^{\prime}\right) \\
&  \geq I\left(  L^{\prime};MB^{n}R^{n}\right) \\
&  =I\left(  L^{\prime}MB^{n};R^{n}\right)  +I\left(  L^{\prime}
;MB^{n}\right)  -I\left(  R^{n};MB^{n}\right) \\
&  \geq I\left(  L^{\prime}MB^{n};R^{n}\right)  -I\left(  R^{n};B^{n}\right)
\\
&  \geq I(\hat{X}^{n}B^{\prime n};R^{n})_{\omega}-I\left(  R^{n};B^{n}\right)
_{\sigma}\\
&  \geq I(X^{n}B^{n};R^{n})_{\sigma}-I\left(  R^{n};B^{n}\right)  _{\sigma
}-n\epsilon^{\prime}\\
&  =I(X^{n};B^{n}|R^{n})_{\sigma}-n\epsilon^{\prime}\\
&  =nI(X;B|R)-n\epsilon^{\prime}.
\end{align*}
The first two inequalities are straightforward (similar to steps in our
previous converse proofs). The first equality is an identity for quantum
mutual information. The third inequality follows because $I\left(  L^{\prime
};MB^{n}\right)  \geq0$ and the common randomness $M$ is uncorrelated with
systems $R^{n}$ and $B^{n}$ (so that $I\left(  R^{n};MB^{n}\right)  =I\left(
R^{n};B^{n}\right)  $). The fourth inequality follows from quantum data
processing of the systems $L^{\prime}MB^{n}$ to produce the systems $\hat
{X}^{n}B^{\prime n}$. The fifth inequality follows from the condition in
(\ref{eq:ICQSI-good-protocol}) and continuity of quantum mutual information
(the Alicki-Fannes' inequality \cite{AF04}), where $\epsilon^{\prime}$ is some
function $g\left(  \epsilon\right)  $ such that $\lim_{\epsilon\rightarrow
0}g\left(  \epsilon\right)  =0$. The second equality follows from the chain
rule for quantum mutual information: $I(X^{n}B^{n};R^{n})_{\sigma}%
=I(X^{n};R^{n}|B^{n})_{\sigma}+I(B^{n};R^{n})_{\sigma}$. The final equality
follows because conditional quantum mutual information is additive on
tensor-power states.

The argument justifying the other  bound in Theorem~\ref{thm:ICQSI}\ goes as follows:%
\begin{align*}
n\left(  R+S\right)   &  \geq H\left(  L^{\prime}M\right) \\
&  \geq H\left(  L^{\prime}M|B^{n}\right) \\
&  =I\left(  X^{\prime n};L^{\prime}M|B^{n}\right)  +H\left(  L^{\prime}
M|B^{n}X^{\prime n}\right) \\
&  \geq I\left(  X^{\prime n};L^{\prime}M|B^{n}\right) \\
&  =H\left(  X^{\prime n}|B^{n}\right)  -H\left(  X^{\prime n}|L^{\prime
}MB^{n}\right) \\
&  \geq H\left(  X^{\prime n}|B^{n}\right)  _{\omega}-H(X^{\prime n}|\hat
{X}^{n})_{\omega}\\
&  \geq H\left(  X^{n}|B^{n}\right)  _{\sigma}-n\epsilon^{\prime}\\
&  =nH\left(  X|B\right)  -n\epsilon^{\prime}.
\end{align*}
The first two inequalities are straightforward. The first equality is an
identity for quantum mutual information. The third inequality follows because
the entropy $H\left(  L^{\prime}M|B^{n}X^{\prime n}\right)  \geq0$ for
classical systems $L^{\prime}$ and $M$. The second equality is an identity for
quantum mutual information. The fourth inequality follows from quantum data
processing of the systems $L^{\prime}MB^{n}$. The last inequality follows from
the condition in (\ref{eq:ICQSI-good-protocol}), continuity of entropy, and
the fact that $H(X^{n}|\overline{X}^{n})_{\sigma}=0$ since $\overline{X}^{n}$
is a copy of $X^{n}$. The final equality follows because the entropy is
additive for tensor-power states.

Optimality of the bound $R+S\geq H\left(  X|B\right)  $ for negative $S$
follows by considering a protocol whereby Alice uses classical communication
alone in order to simulate $X^{n}$ and generate common randomness $M$ with
Bob. The converse in this case proceeds as follows:%
\begin{align*}
nR  &  \geq H\left(  L^{\prime}\right) \\
&  \geq H\left(  L^{\prime}|B^{n}\right) \\
&  \geq I\left(  X^{\prime n}M;L^{\prime}|B^{n}\right) \\
&  =H\left(  X^{\prime n}M|B^{n}\right)  -H\left(  X^{\prime}M|L^{\prime}%
B^{n}\right) \\
&  \geq H\left(  X^{\prime n}M|B^{n}\right)  _{\omega}-H(X^{\prime}M|\hat
{X}^{n}M)_{\omega}\\
&  \geq H\left(  X^{n}M|B^{n}\right)  _{\sigma}-n\epsilon^{\prime}\\
&  =nH\left(  X|B\right)  +H\left(  M\right)  -n\epsilon^{\prime}\\
&  =nH\left(  X|B\right)  +n|S|-n\epsilon^{\prime}.
\end{align*}
The fourth inequality follows because Bob has to process $L^{\prime}$ and
$B^{n}$ in order to recover the approximate $\hat{X}^{n}$ and $M$. The fifth
inequality follows because these systems should be close to the ideal ones for
a good protocol (and applying continuity of entropy). The next equalities
follow because the information quantities factor as above for the ideal state.

\subsection{Relation of MC-QSI to other protocols}

We remark on the connection between measurement compression with quantum side
information and two other protocols:\ channel simulation with quantum side
information \cite{LD09}\ and state redistribution \cite{DY08,YD09}. MC-QSI
lies somewhere in between both of these protocols---it generalizes channel
simulation with QSI but is not \textquotedblleft fully
quantum,\textquotedblright in contrast to state redistribution, which is. Channel simulation with
QSI\ is a protocol whereby a sender and receiver share many copies of a
classical-quantum state $\sum_{y}p_{Y}\left(  y\right)  \left\vert
y\right\rangle \left\langle y\right\vert ^{Y}\otimes\rho_{y}^{B}$\ distributed
to them by a source, with the sender holding the classical systems and the
receiver holding the quantum systems. The goal is for the sender and receiver
to simulate the action of a noisy classical channel $p_{X|Y}\left(
x|y\right)  $ on the sender's classical systems by using as few noiseless bit
channels and common randomness bits as possible. Luo and Devetak found that
this is possible by using a classical communication rate of $I\left(
X;Y|B\right)  $ and a common randomness rate of $H\left(  X|YB\right)  $
(compare with $I\left(  X;R|B\right)  $ and $H\left(  X|RB\right)  $ for
MC-QSI), where the entropies are with respect to a state of the following
form:%
\[
\sum_{y,x}p_{Y}\left(  y\right)  p_{X|Y}\left(  x|y\right)  \left\vert
y\right\rangle \left\langle y\right\vert ^{Y}\otimes\left\vert x\right\rangle
\left\langle x\right\vert ^{X}\otimes\rho_{y}^{B}.
\]
(They actually found the rates to be $I\left(  Y;X\right)  -I\left(
B;X\right)  $ and $H\left(  X|Y\right)  $, but combining these rates with the
fact that $I\left(  X;B|Y\right)  =0$ for a state of the above form gives the
rates we state above.) This protocol exploits aspects of CDC-QSI and the
classical reverse Shannon theorem in its proof. It has applications to rate
distortion theory with quantum side information and in devising a simpler
proof of the distillable common randomness from quantum states \cite{DW04}.
The completely classical version of this protocol has further applications to
multi-terminal problems in classical rate distortion theory \cite{L09}.
Clearly, our protocol generalizes channel simulation with QSI because a
classical channel, a classical-to-classical map, is a special case of a
quantum measurement, a quantum-to-classical map.

State redistribution is a protocol that generalizes MC-QSI to the setting
where one would like to simulate the action of a noisy quantum channel on some
bipartite state $\rho^{AB}$. That is, state redistribution leads to a quantum
reverse Shannon theorem in the presence of quantum side information which we
call QRST-QSI (the authors of Refs.~\cite{DY08,YD09}\ did not emphasize this
aspect of their protocol). Indeed, supposing that the goal is to simulate the
action of a channel $\mathcal{N}^{A\rightarrow B^{\prime}}$ on the bipartite
state, they could proceed by Alice locally performing the isometric extension
$U_{\mathcal{N}}^{A\rightarrow B^{\prime}E}$ of the channel $\mathcal{N}%
^{A\rightarrow B^{\prime}}$ on the $A$ system of the state $\rho^{AB}$.
Including the reference $R$ as a purification of $\rho^{AB}$, there are four
systems $RB^{\prime}EB$ after she does so, where Alice possesses $B^{\prime}$
and $E$ and Bob possesses $B$. Alice and Bob then operate according to the
state redistribution protocol in order for Alice to transfer the $B^{\prime}$
system to Bob (this effects the channel simulation of $\mathcal{N}%
^{A\rightarrow B^{\prime}}$ on the state $\rho^{AB}$). Transferring the state
requires some rate $Q$ of noiseless quantum communication and some rate $E$ of
noiseless entanglement, and according to the main theorem of
Refs.~\cite{DY08,YD09}, this is possible as long as%
\begin{align*}
Q  &  \geq\frac{1}{2}I\left(  B^{\prime};R|B\right)  ,\\
Q+E  &  \geq H\left(  B^{\prime}|B\right)  .
\end{align*}
Comparing the above rate region with Theorem~\ref{thm:ICQSI}\ of this paper
reveals a close analogy between noiseless quantum communication / entanglement
in QRST-QSI and noiseless classical communication / common randomness in
MC-QSI, with the factor of 1/2 above accounting for the fact that the
communication in QRST-QSI\ is quantum. Though, one should be aware that this
connection is only formally similar---in QRST-QSI, sometimes the protocol can
generate entanglement rather than consume it, depending on the channel and the
state on which the channel acts (this can never happen in MC-QSI because the
entropy $H(X|RB)$ is always positive for a classical $X$ system).

\subsection{Applications of MC-QSI}

We now discuss three applications of MC-QSI. The first application is one that
two of us announced in Ref.~\cite{HW10}, the second involves developing a
quantum reverse Shannon theorem for a quantum instrument, and the third is in
reducing the classical communication cost of the local purity distillation
protocol outlined in Ref.~\cite{KD07}.

\subsubsection{Classically-assisted state redistribution}

For the first application, the setting is that Alice and Bob share many copies
of some bipartite state $\rho^{AB}$, and we would like to know how the
resources of classical communication, quantum communication, and entanglement
can combine with the state $\rho^{AB}$ for different information processing
tasks. Let $\left\vert \psi\right\rangle ^{RAB}$ be a purification of
$\rho^{AB}$. We found a general protocol, called \textquotedblleft
classically-assisted state redistribution,\textquotedblright\ that when
combined with teleportation, super-dense coding, and entanglement distribution
can generate all of the known \textquotedblleft static\textquotedblright%
\ protocols in the literature and is furthermore optimal for these tasks
according to a multi-letter converse theorem \cite{HW10}. In the first step
of classically-assisted state redistribution, Alice and Bob employ the MC-QSI
protocol in order to implement the following resource inequality:%
\begin{equation}
\left\langle \rho^{AB}\right\rangle +I\left(  X_{B};R|B\right)  \left[
c\rightarrow c\right]  +H\left(  X_{B}|RB\right)  \left[  cc\right]
\geq\langle\overline{\Delta}^{X\rightarrow X_{A}X_{B}}\circ T^{A\rightarrow
A^{\prime}XE^{\prime}}:\rho^{AB}\rangle. \label{eq:inst-comp-QSI-res-in}%
\end{equation}
In the above, the resource on the RHS\ is a remote instrument, such that the
map $T^{A\rightarrow A^{\prime}XE^{\prime}}$ is first simulated
in such a way that Alice possesses the environment $E^\prime$ of the instrument $T^{A\rightarrow A^{\prime}X}$,
followed by a copying of the classical output $X$ to one for Alice ($X_{A}$)
and one for Bob ($X_{B}$) (the notation $\overline{\Delta}^{X\rightarrow X_{A}X_{B}}$ indicates
a classical copying channel). Let $\sigma^{A^{\prime}X_{A}X_{B}E^{\prime}B}$
denote the post-measurement state. Conditional on the classical variable $X$,
the parties then perform the state redistribution protocol \cite{DY08,YD09},
in which Alice redistributes the share $A^{\prime}$ of the post-measurement
state to Bob. The resource inequality for this task is as follows:%
\[
\langle\sigma^{A^{\prime}E^{\prime}X_{A}|BX_{B}}\rangle+\frac{1}{2}I\left(
A^{\prime};R|BX_{B}\right)  \left[  q\rightarrow q\right]  +\frac{1}{2}\left(
I\left(  A^{\prime};E^{\prime}|X_{B}\right)  -I\left(  A^{\prime}%
;B|X_{B}\right)  \right)  \left[  qq\right]  \geq\langle\sigma^{E^{\prime
}X_{A}|A^{\prime}BX_{B}}\rangle,
\]
where the vertical divider $|$\ for the states above indicates who possesses
what systems and the information quantities are all conditioned on $X$ since
this classical variable is available to both parties.\ The above resource
inequality is equivalent to the following one, after applying the identity
$I\left(  A^{\prime};R|BX_{B}\right)  =I\left(  A^{\prime};R|E^{\prime}%
X_{B}\right)  $ \cite{DY08,YD09} and moving the entanglement consumption to
the RHS\ along with a sign inversion (so that it now corresponds to an
entanglement generation rate):%
\[
\langle\sigma^{A^{\prime}X_{A}E^{\prime}|BX_{B}}\rangle+\frac{1}{2}I\left(
A^{\prime};R|E^{\prime}X_{B}\right)  \left[  q\rightarrow q\right]  \geq
\frac{1}{2}\left(  I\left(  A^{\prime};B|X_{B}\right)  -I\left(  A^{\prime
};E^{\prime}|X_{B}\right)  \right)  \left[  qq\right]  +\langle\sigma
^{E^{\prime}X_{A}|A^{\prime}BX_{B}}\rangle.
\]
Overall, we then have the following resource inequality%
\begin{multline*}
\left\langle \rho^{AB}\right\rangle +I\left(  X_{B};R|B\right)  \left[
c\rightarrow c\right]  +H\left(  X_{B}|RB\right)  \left[  cc\right]  +\frac
{1}{2}I\left(  A^{\prime};R|E^{\prime}X_{B}\right)  \left[  q\rightarrow
q\right] \\
\geq\frac{1}{2}\left(  I\left(  A^{\prime};B|X_{B}\right)  -I\left(
A^{\prime};E^{\prime}|X_{B}\right)  \right)  \left[  qq\right]  ,
\end{multline*}
if we are not concerned with the \textquotedblleft state
redistribution\textquotedblright\ aspect of the protocol and merely its
abilities for entanglement distillation. Finally, since the goal of the
protocol is entanglement distillation and not actually simulating the
measurement, we can exploit the common randomness to agree upon a particular
protocol in the ensemble of these protocols for the task of entanglement
distillation and it is not necessary to have common randomness as a resource
(it can be derandomized and this is the content of Corollary~4.8 of
Ref.~\cite{DHW08}). The final resource inequality is then%
\[
\left\langle \rho^{AB}\right\rangle +I\left(  X_{B};R|B\right)  \left[
c\rightarrow c\right]  +\frac{1}{2}I\left(  A^{\prime};R|E^{\prime}%
X_{B}\right)  \left[  q\rightarrow q\right]  \geq\frac{1}{2}\left(  I\left(
A^{\prime};B|X_{B}\right)  -I\left(  A^{\prime};E^{\prime}|X_{B}\right)
\right)  \left[  qq\right]  .
\]
Combining the above protocol with teleportation, super-dense coding, and
entanglement distribution then gives all of the known protocols on the
\textquotedblleft static branch\textquotedblright\ of quantum information theory.

\subsubsection{Quantum reverse Shannon theorem for a quantum instrument}

The quantum reverse Shannon theorem in its simplest form makes a statement
regarding the ability of noiseless quantum communication and entanglement to
simulate the action of some channel $\mathcal{N}^{A\rightarrow B^{\prime}}$ on
many copies of a state $\rho$. A simple extension of the theorem that we
discussed in the introduction is QRST-QSI, which simulates the channel on a
bipartite state $\rho^{AB}$. The resource inequality for this protocol is as
follows:%
\[
\frac{1}{2}I\left(  R;B^{\prime}|B\right)  _{\omega}\left[  q\rightarrow
q\right]  +\frac{1}{2}\left(  I\left(  B^{\prime};E\right)  _{\omega}-I\left(
B^{\prime};B\right)  _{\omega}\right)  \ \left[  qq\right]  \geq\langle
U_{\mathcal{N}}^{A\rightarrow B^{\prime}E}:\rho^{AB}\rangle,
\]
where the information quantities are with respect to a state $\omega^{RBE}$ of
the following form:%
\[
\omega^{RB^{\prime}EB}\equiv U_{\mathcal{N}}^{A\rightarrow B^{\prime}%
E}\left\vert \psi_{\rho}\right\rangle ^{RAB}.
\]
In the above, $U_{\mathcal{N}}^{A\rightarrow B^{\prime}E}$ is an isometric
extension of the channel $\mathcal{N}$ and $\left\vert \psi_{\rho
}\right\rangle ^{RAB}$ is a purification of the state $\rho^{AB}$. The
protocol employs state redistribution \cite{DY08,YD09}. In this case, if
$\frac{1}{2}\left(  I\left(  B^{\prime};E\right)  _{\omega}-I\left(
B^{\prime};B\right)  _{\omega}\right)  $ is negative, then the protocol is
generating entanglement rather than consuming it. A special case of the above
theorem is when there is no quantum side information (when $B$ is trivial), in
which case the resource inequality becomes the usual quantum reverse Shannon
theorem \cite{D06,ADHW06FQSW,BDHSW09,BCR11}:%
\[
\frac{1}{2}I\left(  R;B^{\prime}\right)  _{\omega}\left[  q\rightarrow
q\right]  +\frac{1}{2}I\left(  B^{\prime};E\right)  _{\omega}\left[
qq\right]  \geq\langle U_{\mathcal{N}}^{A\rightarrow B^{\prime}E}:\rho
^{A}\rangle.
\]

Suppose that we instead would like to simulate the action of a quantum
instrument $\mathcal{N}^{A\rightarrow XB^{\prime}}$ with classical output $X$
and quantum output $B^{\prime}$ on the bipartite state $\rho^{AB}$. A quantum
instrument is the most general model for quantum measurement that includes
both a classical output and a post-measurement quantum state \cite{D76}. A
quantum instrument always admits the following decomposition%
\[
\mathcal{N}^{A\rightarrow XB^{\prime}}\left(  \rho\right)  \equiv\sum
_{x}\left\vert x\right\rangle \left\langle x\right\vert ^{X}\otimes
\mathcal{N}_{x}^{A\rightarrow B^{\prime}}\left(  \rho\right)  ,
\]
in terms of the completely positive trace-nonincreasing maps $\mathcal{N}%
_{x}^{A\rightarrow B^{\prime}}$, such that the overall quantum map after
tracing over the classical system $X$ is a completely positive
trace-preserving map:%
\begin{align*}
\mathcal{N}^{A\rightarrow B^{\prime}}\left(  \rho\right)   &  \equiv\sum
_{x}\mathcal{N}_{x}^{A\rightarrow B^{\prime}}\left(  \rho\right)  ,\\
\text{Tr}\left\{  \mathcal{N}^{A\rightarrow B^{\prime}}\left(  \rho\right)
\right\}   &  =1.
\end{align*}
$\mathcal{N}^{A \rightarrow B'}$ has the following isometric extension:%
\[
U_{\mathcal{N}}^{A\rightarrow XX_{E}B^{\prime}E}=\sum_{x}\left\vert
x\right\rangle ^{X}\left\vert x\right\rangle ^{X_{E}}U_{\mathcal{N}_{x}%
}^{A\rightarrow B^{\prime}E},
\]
where $U_{\mathcal{N}_{x}}^{A\rightarrow B^{\prime}E}$ is an extension of the
map $\mathcal{N}_{x}^{A\rightarrow B^{\prime}}$. Thus, tracing over $X_{E}$
and $E$ recovers the action of the original instrument.

If we are interested in simulating this channel on the state $\rho^{AB}$, we
could straightforwardly apply the quantum reverse Shannon theorem to show that
the following resource inequality exists%
\begin{equation}
\frac{1}{2}I\left(  R;B^{\prime}X|B\right)  _{\omega}\left[  q\rightarrow
q\right]  +\frac{1}{2}\left(  I\left(  B^{\prime}X;E\right)  _{\omega
}-I\left(  B^{\prime}X;B\right)  _{\omega}\right)  \ \left[  qq\right]
\geq\langle U_{\mathcal{N}}^{A\rightarrow XX_{E}B^{\prime}E}:\rho^{AB}\rangle.
\label{eq:fully-q-instr-sim}%
\end{equation}
Though, we could perform this task by using less quantum communication and
entanglement if we exploit MC-QSI\ first followed by state redistribution (as
we do in the classically-assisted state redistribution protocol). The first
step of the protocol implements the following resource inequality%
\[
\left\langle \rho^{AB}\right\rangle +I\left(  X;R|B\right)  \left[
c\rightarrow c\right]  +H\left(  X|RB\right)  \left[  cc\right]  \geq
\langle\overline{\Delta}^{X\rightarrow X_{A}X_{B}}\circ\mathcal{N}%
^{A\rightarrow B^{\prime}XE}:\rho^{AB}\rangle,
\]
while the second is as follows:%
\[
\langle\sigma^{B^{\prime}X_{A}E|BX_{B}}\rangle+\frac{1}{2}I\left(  B^{\prime
};R|BX\right)  \left[  q\rightarrow q\right]  +\frac{1}{2}\left(  I\left(
B^{\prime};E|X\right)  -I\left(  B^{\prime};B|X\right)  \right)  \left[
qq\right]  \geq\langle\sigma^{X_{A}E|B^{\prime}BX_{B}}\rangle.
\]
Overall, the resource inequality for simulating the quantum instrument is as
follows:%
\begin{multline*}
\left\langle \rho^{AB}\right\rangle +I\left(  X;R|B\right)  \left[
c\rightarrow c\right]  +H\left(  X|RB\right)  \left[  cc\right]  +\frac{1}%
{2}I\left(  B^{\prime};R|BX\right)  \left[  q\rightarrow q\right] \\
+\frac{1}{2}\left(  I\left(  B^{\prime};E|X\right)  -I\left(  B^{\prime
};B|X\right)  \right)  \left[  qq\right]  \geq\langle U_{\mathcal{N}%
}^{A\rightarrow XX_{E}B^{\prime}E}:\rho^{AB}\rangle,
\end{multline*}
which is a cheaper simulation than in (\ref{eq:fully-q-instr-sim}) because we
are using classical communication and common randomness to achieve part of the
task, rather than quantum communication and entanglement for the whole
protocol. One would expect to have such a savings, since a quantum instrument
has both a classical and quantum output. We remark that this approach is very
similar to classically-assisted state redistribution from the previous
section, with the exception that we require the common randomness since the
goal is to simulate the instrument in full, rather than to distill
entanglement. A special case of the above reverse Shannon theorem occurs when
there is no quantum side information available, in which the resource
inequality reduces to%
\[
I\left(  X;R\right)  \left[  c\rightarrow c\right]  +H\left(  X|R\right)
\left[  cc\right]  +\frac{1}{2}I\left(  B^{\prime};R|X\right)  \left[
q\rightarrow q\right]  +\frac{1}{2}I\left(  B^{\prime};E|X\right)  \left[
qq\right]  \geq\langle U_{\mathcal{N}}^{A\rightarrow XX_{E}B^{\prime}E}%
:\rho^{A}\rangle.
\]

\subsubsection{Classical communication cost in local purity distillation}

We can also exploit MC-QSI to improve upon the classical communication cost in
local purity distillation \cite{KD07}. This leads to the following improvement of
Theorem~1 of Ref.~\cite{KD07}:

\begin{theorem}
The 1-way distillable local purity of the state $\rho^{AB}$ is given by
$\kappa_{\rightarrow}=\kappa_{\rightarrow}^{\ast}$, where%
\[
\kappa_{\rightarrow}^{\ast}\left(  \rho^{AB},R\right)  =\kappa\left(  \rho
^{A}\right)  +\kappa\left(  \rho^{B}\right)  +P_{\rightarrow}\left(  \rho
^{AB},R\right)  .
\]
In the above, we have the definitions%
\begin{align*}
\kappa\left(  \omega^{C}\right)   &  \equiv\log d_{C}-H\left(  C\right)
_{\omega},\\
P_{\rightarrow}\left(  \rho^{AB},R\right)   &  \equiv\lim_{k\rightarrow\infty
}\frac{1}{k}P^{\left(  1\right)  }\left(  \left(  \rho^{AB}\right)  ^{\otimes
k},kR\right)  ,
\end{align*}
and%
\begin{align*}
P^{\left(  1\right)  }\left(  \rho^{AB},R\right)   &  \equiv\max_{\Lambda
}\left\{  I\left(  Y;B\right)  _{\sigma}:I\left(  Y;E|B\right)  \leq
R\right\}  ,\\
\sigma^{YBE}  &  \equiv\left(  \mathcal{M}_{\Lambda}\otimes I^{BE}\right)
\left(  \psi^{ABE}\right)  ,
\end{align*}
where $\psi^{ABE}$ is a purification of $\rho^{AB}$, $\mathcal{M}_{\Lambda}$
is a measurement map corresponding to the POVM $\Lambda$, and the maximization
is over all POVMs mapping Alice's system $A$ to a classical system $Y$.
\end{theorem}

The improvement of Theorem~1 of Ref.~\cite{KD07} comes about by reducing the
classical communication rate from $I\left(  Y;EB\right)  $ to $I\left(
Y;E|B\right)  $ by employing the MC-QSI\ protocol in the achievability part.
The converse part of the theorem (in (19) of Ref.~\cite{KD07}) gets improved as
follows:%
\[
nR=\log d_{Y}\geq H\left(  Y\right)  \geq H\left(  Y|B^{n}\right)  \geq
I\left(  Y;E^{n}|B^{n}\right)  .
\]
It is apparent that the multi-letter nature of the converse theorem is what
led to the ability to improve the theorem, so that for any finite $k$, the
above revision of the theorem improves upon the previous one, but they are
both optimal in the regularized limit. This leads us to believe that 
even further improvements might be possible.

\subsection{Entropic uncertainty relation with QSI}

We close by relating the MC-QSI\ protocol to recent work on an entropic
uncertainty relation in the presence of quantum memory
\cite{RB09,BCCRR10,TR11,CCYZ11,FL12}. This uncertainty relation characterizes
the ability of two parties to predict the outcomes of measurements on another
system, by exploiting the quantum systems in their possession. The formal
statement of the uncertainty relation applies to a tripartite state
$\rho^{ABC}$ and is as follows:%
\begin{equation}
H\left(  X|B\right)  +H\left(  Z|C\right)  \geq\log_{2}\left(  1/c_{1}\right)
. \label{eq:EUR}%
\end{equation}
The two entropies are with respect to the following states resulting from
applying measurement maps for $\Lambda$ and $\Gamma$ to the $A$ system:%
\begin{align*}
&  \sum_{x}\left\vert x\right\rangle \left\langle x\right\vert ^{X}%
\otimes\text{Tr}_{AC}\left\{  \Lambda_{x}^{A}\ \rho^{ABC}\right\}  ,\\
&  \sum_{z}\left\vert z\right\rangle \left\langle z\right\vert ^{Z}%
\otimes\text{Tr}_{AB}\left\{  \Gamma_{z}^{A}\ \rho^{ABC}\right\}  ,
\end{align*}
and $c$ characterizes the non-commutativity of the measurements:%
\[
c_{1}\equiv\max_{x,z}\left\Vert \sqrt{\Lambda_{x}}\sqrt{\Gamma_{z}}\right\Vert
_{\infty}^{2}.
\]
This uncertainty relation is useful conceptually, but it also has operational
applications to quantum key distribution \cite{BCCRR10,TR11}, in relating data
compression to privacy amplification \cite{Renes08062011,RR12}, and in
constructing capacity-achieving quantum error correction codes\ (for certain
channels) \cite{WR12} because it is formulated in terms of entropies. Another
statement of the above entropic uncertainty relation is as follows
\cite{BCCRR10,FL12}:%
\begin{equation}
H\left(  X|B\right)  +H\left(  Z|B\right)  \geq\log_{2}\left(  1/c_{2}\right)
+H\left(  A|B\right)  , \label{eq:other-EUR}%
\end{equation}
where%
\[
c_{2}\equiv\max_{x,z}\sqrt{\text{Tr}\left\{  \Lambda_{x}\Gamma_{z}\right\}
}.
\]

Here, we show how the above uncertainty relations apply in bounding from below
the nonlocal classical resources required in two different MC-QSI protocols.
First, suppose that Alice would like to simulate the measurement $\left\{
\Lambda_{x}^{A}\right\}  $ on the state $\rho^{ABC}$ and send the outcomes to
Bob. Let $R$ be a system that purifies the state $\rho^{ABC}$. Then the MC-QSI
protocol corresponds to the following resource inequality:%
\[
\left\langle \rho^{ABC}\right\rangle +I\left(  X;RC|B\right)  \left[
c\rightarrow c\right]  +H\left(  X|RBC\right)  \left[  cc\right]
\geq\left\langle \Lambda^{A}:\rho^{ABC}\right\rangle ,
\]
where the entropies are with respect to the state%
\[
\sum_{x}\left\vert x\right\rangle \left\langle x\right\vert ^{X}%
\otimes\text{Tr}_{A}\left\{  \Lambda_{x}^{A}\ \psi^{RABC}\right\}  .
\]
The total classical cost of the above protocol is $H\left(  X|B\right)
=I\left(  X;RC|B\right)  +H\left(  X|RBC\right)  $. For the second protocol,
suppose that Alice would like to simulate the measurement $\left\{  \Gamma
_{z}^{A}\right\}  $ on the state $\rho^{ABC}$ and send the outcomes to
Charlie. Then the MC-QSI protocol in this case corresponds to the following
resource inequality:%
\[
\left\langle \rho^{ABC}\right\rangle +I\left(  Z;RB|C\right)  \left[
c\rightarrow c\right]  +H\left(  Z|RBC\right)  \left[  cc\right]
\geq\left\langle \Gamma^{A}:\rho^{ABC}\right\rangle ,
\]
where the entropies are with respect to the state%
\[
\sum_{z}\left\vert z\right\rangle \left\langle z\right\vert ^{Z}%
\otimes\text{Tr}_{A}\left\{  \Gamma_{z}^{A}\ \psi^{RABC}\right\}  .
\]
The total classical cost of the above protocol is $H\left(  Z|C\right)
=I\left(  Z;RB|C\right)  +H\left(  Z|RBC\right)  $.

Using the entropic uncertainty relation in (\ref{eq:EUR}), we can then bound
from below the total classical cost of the above protocols as follows:%
\[
I\left(  X;RC|B\right)  +H\left(  X|RBC\right)  +I\left(  Z;RB|C\right)
+H\left(  Z|RBC\right)  =H\left(  X|B\right)  +H\left(  Z|C\right)  \geq
\log_{2}\left(  1/c_{1}\right)  .
\]
We can also apply the uncertainty relation in (\ref{eq:other-EUR}) to bound
from below the total common randomness cost:%
\begin{align*}
H\left(  X|RBC\right)  +H\left(  Z|RBC\right)   &  \geq\log_{2}\left(
1/c_{2}\right)  +H\left(  A|RBC\right) \\
&  =\log_{2}\left(  1/c_{2}\right)  -H\left(  A\right)  ,
\end{align*}
where the last equality follows because the state on $RABC$ is pure. Since
this lower bound might sometimes be negative but the entropies $H\left(
X|RBC\right)  $ and $H\left(  Z|RBC\right)  $ are always positive, we can
revise the above lower bound to be as follows:%
\[
H\left(  X|RBC\right)  +H\left(  Z|RBC\right)  \geq\max\left\{  \log
_{2}\left(  1/c_{2}\right)  -H\left(  A\right)  ,0\right\}  .
\]

Given that lower bounds on the total classical cost and the total 
common randomness exist, one might be tempted to think that a lower bound on the total
information should exist as well. One might conjecture it to be of the
following form:%
\[
I\left(  X;RC|B\right)  +I\left(  Z;RB|C\right)  \geq l,
\]
where $l$ is some non-negative parameter that depends only on the measurements
and not on the state. Such a universal, state-independent lower bound
cannot hold in general, however. A simple counterexample demonstrates that the
following lower bound for strong subadditivity is the best that one might hope
for:%
\[
I\left(  X;RC|B\right)  +I\left(  Z;RB|C\right)  \geq0.
\]
Indeed, suppose that $\rho^{ABC}$ is a pure product state. Then $I\left(
X;RC|B\right)  $ is equal to zero because $R$ and $C$ have no correlations
with the measurement output $X$ on $A$, and $I\left(  Z;RB|C\right)  =0$ for a
similar reason.

\section{Non-feedback measurement compression with quantum side information}

Our final contribution concerns measurement compression with quantum side
information, in the case where the sender is not required to obtain the
outcome of the simulation, that is,  a non-feedback
simulation. We construct a protocol for this task by simply
combining elements of other protocols described earlier in the article. Moreover, we show that the protocol is optimal by proving a single-letter converse for the associated rate region.
We omit the detailed definition of the information processing task here because
it is the obvious non-feedback relaxation along the lines of Section \ref{sec:MC-non-feedback} for the definition of the MC-QSI task from Section \ref{sec:MC-QSI-feedback}.

\begin{theorem}
[Non-feedback MC-QSI]\label{thm:non-feedback-MC-QSI-inner}Let $\rho^{AB}$ be a
source state and $\mathcal{N}$ a quantum instrument to simulate on this state:%
\[
\left(  \mathcal{N}^{A\rightarrow AX}\otimes I^{B}\right)  \left(  \rho
^{AB}\right)  =\sum_{x}\left(  \mathcal{N}_{x}^{A}\otimes I^{B}\right)
\left(  \rho^{AB}\right)  \otimes\left\vert x\right\rangle \left\langle
x\right\vert ^{X}.
\]
There exists a protocol for faithful non-feedback simulation of the quantum instrument with classical communication rate $R$ and common randomness rate
$S$ if and only if $R$ and $S$ are in the union of the following
regions:%
\begin{align}
R  &  \geq I\left(  W;R|B\right)  ,\label{eq:1st-non-feed-MC-QSI}\\
R+S  &  \geq I\left(  W;XR|B\right)  ,\nonumber
\end{align}
where the entropies are with respect to a state of the following form:%
\begin{equation} \label{eqn:MC-QSI-non-feedback-state}
\sum_{x,w}p_{X|W}\left(  x|w\right)  \left\vert w\right\rangle \left\langle
w\right\vert ^{W}\otimes\left\vert x\right\rangle \left\langle x\right\vert
^{X}\otimes\text{\emph{Tr}}_{A}\left\{  \left(  I^{R}\otimes\mathcal{M}%
_{w}^{A}\otimes I^{B}\right)  \left(  \phi_{\rho}^{RAB}\right)  \right\}  ,
\end{equation}
$\phi_{\rho}^{RAB}$ is some purification of the state $\rho^{AB}$, and the
union is with respect to all decompositions of the original instrument
$\mathcal{N}$ of the form:%
\begin{equation}
\left(  \mathcal{N}^{A\rightarrow AX}\otimes I^{B}\right)  \left(  \rho
^{AB}\right)  =\sum_{x,w}p_{X|W}\left(  x|w\right)  \left(  \mathcal{M}%
_{w}^{A}\otimes I^{B}\right)  \left(  \rho^{AB}\right)  \otimes\left\vert
x\right\rangle \left\langle x\right\vert ^{X}.
\label{eq:non-feed-MC-QSI-decomp}%
\end{equation}

\end{theorem}

While demonstrating achievability of the quoted rates will consist of the routine combination of elements from other parts of the article, the converse is more subtle. In particular, a general protocol for non-feedback MC-QSI will have Bob perform an instrument on the $B^n$ system. Arguing that it is sufficient to restrict to states of the form (\ref{eqn:MC-QSI-non-feedback-state}) will involve comparing that protocol to a related simulation in which the instrument is implemented approximately by Alice.  While the modified protocol would generally require more communication than the original, for the purposes of the converse, it need not significantly increase the relevant mutual informations.

\smallskip

\begin{proof}
[Proof Sketch of Achievability]The protocol for achievability naturally
combines elements of protocols that we have considered in
Section~\ref{sec:MC-non-feedback}\ for non-feedback measurement compression
and in Section~\ref{sec:MC-QSI-feedback} for measurement compression with
quantum side information.\ The protocol begins with Alice and Bob sharing many
copies of a state $\rho^{AB}$. They would like to simulate an instrument
$\mathcal{N}^{A\rightarrow AX}$, composed of the completely positive, trace
non-increasing maps $\left\{  \mathcal{N}_{x}^{A}\right\}  $, so that they end
up with many copies of a state of the following form:%
\[
\left(  \mathcal{N}^{A\rightarrow AX}\otimes I^{B}\right)  \left(  \rho
^{AB}\right)  =\sum_{x}\left(  \mathcal{N}_{x}^{A}\otimes I^{B}\right)
\left(  \rho^{AB}\right)  \otimes\left\vert x\right\rangle \left\langle
x\right\vert ^{X}.
\]
We omit the details of the proof of the achievability part because it follows
readily from the methods detailed in Sections~\ref{sec:MC-non-feedback} and
\ref{sec:MC-QSI-feedback}. Instead, we state the achievability part as the
following resource inequality:%
\begin{equation}
\left\langle \rho^{AB}\right\rangle +I\left(  W;R|B\right)  \left[
c\rightarrow c\right]  +I\left(  W;X|RB\right)  \left[  cc\right]
\geq\left\langle \mathcal{N}^{A\rightarrow AX}\left(  \rho^{AB}\right)
\right\rangle . \label{eq:resource-ineq-MC-QSI-non-feedback}%
\end{equation}
where the information quantities are with respect to a state of the following
form:%
\begin{equation}
\sum_{x,w}p_{X|W}\left(  x|w\right)  \left\vert w\right\rangle \left\langle
w\right\vert ^{W}\otimes\left\vert x\right\rangle \left\langle x\right\vert
^{X}\otimes\text{Tr}_{A}\left\{  \left(  I^{R}\otimes\mathcal{M}_{w}%
^{A}\otimes I^{B}\right)  \left(  \phi_{\rho}^{RAB}\right)  \right\}  .
\end{equation}
In the above, $\phi_{\rho}^{RAB}$ is a purification of the state $\rho^{AB}$
and the maps $\left\{  \mathcal{M}_{w}^{A}\right\}  $ arise from a
decomposition of the original instrument into the following form:%
\[
\sum_{x}\mathcal{N}_{x}^{A}\left(  \sigma\right)  \otimes\left\vert
x\right\rangle \left\langle x\right\vert ^{X}=\sum_{x,w}p_{X|W}\left(
x|w\right)  \mathcal{M}_{w}^{A}\left(  \sigma\right)  \otimes\left\vert
x\right\rangle \left\langle x\right\vert ^{X},
\]
when acting on some arbitrary state $\sigma$. In particular, the protocol
operates by Alice and Bob performing a simulation of $\mathcal{M}_{w}^{A}$,
though Alice hashes the outcome of the simulated measurement. She sends the
hash along to Bob using noiseless classical bits channels, and he then
performs sequential decoding to search among all of the post-measurement
states that are consistent with the hash and his share of the common
randomness. This causes a negligible disturbance to the shared state in the
asymptotic limit as long as the communication rates are as in
(\ref{eq:resource-ineq-MC-QSI-non-feedback}). Finally, he simulates the
classical post-processing channel $p_{X|W}\left(  x|w\right)  $ locally,
leading to a savings in the cost of common randomness consumption.
\end{proof}

With the achievability part in hand, we now move on to the proof of the converse.

\begin{proof}
[Proof of Converse]We now prove this converse part. A modification of
Figure~\ref{fig:actual-ICQSI} (without the extra processing of $L$ and $M$ on
Alice's side) depicts the most general protocol for a non-feedback simulation
of the measurement with QSI. The protocol begins with the reference, Alice,
and Bob sharing many copies of the state $\phi_{\rho}^{RAB}$ and Alice sharing
common randomness $M$ with Bob. She then chooses a quantum instrument
$\Upsilon^{\left(  m\right)  }$ based on the common randomness $M$ and
performs it on her systems $A^{n}$. The measurement returns outcome $L$, and
the overall state is as follows:%
\[
\theta^{R^{n}A^{n}B^{n}LM}\equiv\sum_{l,m}\frac{1}{\left\vert \mathcal{M}%
\right\vert }(\Upsilon_{l}^{\left(  m\right)  })^{A^{n}}(\left(  \phi_{\rho
}^{RAB}\right)  ^{\otimes n})\otimes\left\vert l\right\rangle \left\langle
l\right\vert ^{L}\otimes\left\vert m\right\rangle \left\langle m\right\vert
^{M},
\]
where $\Upsilon_{l}^{\left(  m\right)  }$ is a completely positive, trace
non-increasing map. Alice sends the register $L$ to Bob. Based on $L$ and $M$,
he performs some quantum instrument on his systems $B^{n}$ with trace
non-increasing maps $\{\mathcal{F}_{s}^{\left(  lm\right)  }\}\ $followed by
the stochastic map $p_{\hat{X}^{n}|S,L,M}\left(  \hat{x}^{n}|s,l,m\right)  $
to give his estimate $\hat{x}^{n}$ of the measurement outcome. The resulting
state is as follows:%
\begin{multline*}
\omega^{R^{n}A^{n}LMSB^{n}\hat{X}^{n}}\equiv\sum_{l,m,s,\hat{x}^{n}}\frac
{1}{\left\vert \mathcal{M}\right\vert }p_{\hat{X}^{n}|S,L,M}\left(  \hat
{x}^{n}|s,l,m\right)  \left(  (\Upsilon_{l}^{\left(  m\right)  })^{A^{n}%
}\otimes(\mathcal{F}_{s}^{\left(  lm\right)  })^{B^{n}}\right)  (\left(
\phi_{\rho}^{RAB}\right)  ^{\otimes n})\otimes\\
\left\vert l\right\rangle \left\langle l\right\vert ^{L}\otimes\left\vert
m\right\rangle \left\langle m\right\vert ^{M}\otimes\left\vert s\right\rangle
\left\langle s\right\vert ^{S}\otimes\left\vert \hat{x}^{n}\right\rangle
\left\langle \hat{x}^{n}\right\vert ^{\hat{X}^{n}}.
\end{multline*}
The following condition should hold for all $\epsilon>0$ and sufficiently
large $n$ for a faithful non-feedback simulation:%
\begin{equation}
\left\Vert \omega^{R^{n}\hat{X}^{n}B^{n}}-\sum_{x^{n}}\text{Tr}_{A^{n}%
}\left\{  \left(  I\otimes\mathcal{N}_{x^{n}}\right)  \left(  \phi_{\rho
}^{RAB}\right)  ^{\otimes n}\right\}  \otimes\left\vert x^{n}\right\rangle
\left\langle x^{n}\right\vert ^{X^{n}}\right\Vert _{1}\leq\epsilon,
\label{eq:good-non-feed-MC-QSI-protocol}%
\end{equation}
where $\hat{X}^{n}$ is a classical register isomorphic to $X^{n}$.

We prove the first bound as follows:%
\begin{align*}
nR &  \geq H\left(  L\right)  _{\theta}\\
&  \geq I\left(  L;MB^{n}R^{n}\right)  _{\theta}\\
&  =I\left(  LMB^{n};R^{n}\right)  _{\theta}+I\left(  L;MB^{n}\right)
_{\theta}-I\left(  R^{n};MB^{n}\right)  _{\theta}\\
&  \geq I\left(  LMB^{n};R^{n}\right)  _{\theta}-I\left(  R^{n};B^{n}\right)
_{\theta}\\
&  \geq I\left(  LMSB^{n};R^{n}\right)  _{\omega}-I\left(  R^{n};B^{n}\right)
_{\omega}-n\epsilon^{\prime}\\
&  =H\left(  R^{n}|B^{n}\right)  _{\omega}-H\left(  R^{n}|LMSB^{n}\right)
_{\omega}-n\epsilon^{\prime}\\
&  \geq\sum_{k}\left[  H\left(  R_{k}|B_{k}\right)  _{\omega}-H\left(
R_{k}|LMSB_{k}\right)  _{\omega}\right]  -n2\epsilon^{\prime}\\
&  =\sum_{k}I\left(  LMS;R_{k}|B_{k}\right)  _{\omega}-n2\epsilon^{\prime}\\
&  =nI\left(  LMS;R|BK\right)  _{\sigma}-n2\epsilon^{\prime}\\
&  \geq nI\left(  LMS;R|BK\right)  _{\sigma}+nI\left(  R;K|B\right)  _{\sigma
}-n3\epsilon^{\prime}\\
&  =nI\left(  LMSK;R|B\right)  _{\sigma}-n3\epsilon^{\prime}.
\end{align*}
The first two inequalities are similar to what we had before. The first
equality is an identity for quantum mutual information. The third inequality
follows because there are no correlations between $R^{n}B^{n}$ and $M$ so that
$I\left(  MB^{n};R^{n}\right)  _{\omega}=I\left(  B^{n};R^{n}\right)
_{\omega}$. The fourth inequality follows from quantum data processing of
$LMB^{n}$ to produce $LMSB^{n}$ and from the fact that this does not change
the state too much (we apply the condition in
(\ref{eq:good-non-feed-MC-QSI-protocol}) and the Alicki-Fannes' inequality).
The second equality is an identity for quantum mutual information. The fifth
inequality follows from strong subadditivity of quantum entropy:%
\[
H\left(  R^{n}|LMSB^{n}\right)  _{\omega}\leq\sum_{k}H\left(  R_{k}%
|LMSB_{k}\right)  _{\omega},
\]
and because the state on $R^{n}B^{n}$ is close to a tensor-power state so that
by Lemma~\ref{lem:entropy-for-close-IID}, we have%
\[
H\left(  R^{n}|B^{n}\right)  _{\omega}\geq\sum_{k}H\left(  R_{k}|B_{k}\right)
_{\omega}-n\epsilon^{\prime}.
\]
The third equality is another identity. The fourth equality comes about by
defining the state $\sigma$ as follows:%
\begin{multline}
\sigma^{RALMS\hat{X}K}\equiv\sum_{l,m,k,\hat{x},s}\frac{1}{n\left\vert
\mathcal{M}\right\vert }p_{\hat{X}|LMS}\left(  \hat{x}|lms\right)\times\\
\text{Tr}_{\left(  RAB\right)  _{1}^{k-1}\left(  RAB\right)  _{k+1}^{n}%
}\left\{  \left(  (\Upsilon_{l}^{\left(  m\right)  })^{A^{n}}\otimes
(\mathcal{F}_{s}^{\left(  lmk\right)  })^{B^{n}}\right)  (\left(  \phi_{\rho
}^{RAB}\right)  ^{\otimes n})\right\}  \label{eq:helper-state-no-fdbk-MC-QSI}%
\\
\otimes\left\vert l\right\rangle \left\langle l\right\vert ^{L}\otimes
\left\vert m\right\rangle \left\langle m\right\vert ^{M}\otimes\left\vert
s\right\rangle \left\langle s\right\vert ^{S}\otimes\left\vert \hat
{x}\right\rangle \left\langle \hat{x}\right\vert ^{\hat{X}}\otimes\left\vert
k\right\rangle \left\langle k\right\vert ^{K},
\end{multline}
where the map $p_{\hat{X}|LMS}\left(  \hat{x}|lms\right)  $ is defined from
$p_{\hat{X}^{n}|LMS}\left(  \hat{x}^{n}|lms\right)  $ by keeping only the
$k^{\text{th}}$ symbol from $\hat{x}^{n}$. It also follows by exploiting the
fact that $K$ is a uniform classical random variable, with distribution $1/n$,
determining which systems $R_{k}A_{k}B_{k}\hat{X}_{k}$ to select. From the
fact that the measurement simulation is faithful, we can apply the
Alicki-Fannes' inequality to conclude that%
\begin{equation}
I\left(  R\hat{X};K|B\right)  _{\sigma}=\left\vert I\left(  R\hat
{X};K|B\right)  _{\sigma}-I\left(  RX;K|B\right)  _{\tau}\right\vert
\leq\epsilon^{\prime},\label{eq:no-fdbk-KX-info-bnd-MC-QSI}
\end{equation}
where $\tau$ is a state like $\sigma$ but resulting from the tensor-power
state for ideal measurement compression (and due to its IID\ structure, it has
no correlations with any particular system $k$ so that $I\left(
RX;K|B\right)  _{\tau}=0$). The same reasoning along with strong subadditivity
also implies that%
\begin{equation}
I\left(  R;K|B\right)  _{\sigma}\leq\epsilon^{\prime}%
.\label{eq:no-fdbk-K-info-bnd-MC-QSI}%
\end{equation}
The final equality is an application of the chain rule for quantum mutual
information. The state$~\sigma$ for the final information term has the form:%
\begin{equation}
\mathcal{M}^{AB\rightarrow ABX}\left(  \phi^{RAB}\right)  =\sum_{x,w}%
p_{X|W}\left(  x|w\right)  \left\vert x\right\rangle \left\langle x\right\vert
^{X}\otimes\mathcal{M}_{w}^{AB}\left(  \phi^{RAB}\right)
,\label{eq:non-feed-MC-QSI-joint-map}%
\end{equation}
with $LMSK=W$ and the completely positive, trace non-increasing maps
$\mathcal{M}_{w}^{AB}$ defined by%
\[
\varrho^{AB}\mapsto\frac{1}{n\left\vert \mathcal{M}\right\vert }%
\text{Tr}_{\left(  AB\right)  _{1}^{k-1}\left(  AB\right)  _{k+1}^{n}}\left\{
\left(  (\Upsilon_{l}^{\left(  m\right)  })^{A^{n}}\otimes(\mathcal{F}%
_{s}^{\left(  lm\right)  })^{B^{n}}\right)  (\left(  \phi_{\rho}^{AB}\right)
^{\otimes k-1}\otimes\varrho^{AB}\otimes\left(  \phi_{\rho}^{AB}\right)
^{\otimes n-k})\right\}  .
\]

At this point, we have proved the first inequality in
(\ref{eq:1st-non-feed-MC-QSI}) for a state of the form in
(\ref{eq:non-feed-MC-QSI-joint-map})\ where the map $\mathcal{M}_{w}^{AB}$
acts on the joint system $AB$. We now show that it is possible to construct
from $\mathcal{M}_{w}^{AB}$ a map acting only on the system $A$ (as stated in
the theorem) causing only a negligible change to the information quantity in
(\ref{eq:1st-non-feed-MC-QSI}). The idea behind this is a simple application
of Uhlmann's theorem. First, consider that the following inequality holds from
the condition in (\ref{eq:good-non-feed-MC-QSI-protocol}) and from monotonicity of trace distance:%
\[
\left\Vert \sum_{x,w}p_{X|W}\left(  x|w\right)  \text{Tr}_{A}\left\{
\mathcal{M}_{w}^{AB}\left(  \phi^{RAB}\right)  \right\}  -\sum_{x}%
\text{Tr}_{A}\left\{  \mathcal{N}_{x}\left(  \phi^{RAB}\right)  \right\}
\right\Vert _{1}\leq\epsilon
\]
A purification of the state $\sum_{x}$Tr$_{A}\left\{  \mathcal{N}_{x}\left(
\phi^{RAB}\right)  \right\}  $ is the state $\phi^{RAB}$, and a purification
of the state $\sum_{x,w}p_{X|W}\left(  x|w\right)  $Tr$_{A}\left\{
\mathcal{M}_{w}^{AB}\left(  \phi^{RAB}\right)  \right\}  $ is the following
state:%
\begin{equation}
\sum_{w,x,i}M_{w,i}^{AB}\left\vert \phi\right\rangle ^{RAB}\left\vert
w\right\rangle ^{W}\sqrt{p_{X|W}\left(  x|w\right)  }\left\vert x\right\rangle
^{X}\left\vert i\right\rangle ^{I},\label{eq:purified-state-MC-QSI-nonf}%
\end{equation}
where we assume that the completely positive maps $\mathcal{M}_{w}^{AB}$ have
the following Kraus representation:%
\[
\mathcal{M}_{w}^{AB}\left(  \rho^{AB}\right)  =\sum_{i}M_{w,i}^{AB}\rho
^{AB}\left(  M_{w,i}^{\dag}\right)  ^{AB}.
\]
By Uhlmann's theorem, there exists an isometry $U^{A\rightarrow AWXI}$
such that the trace distance between
$U^{A\rightarrow AWXI}\left(  \phi^{RAB}\right)  $ and the state in
(\ref{eq:purified-state-MC-QSI-nonf}) is less
than $2\sqrt{\epsilon}$. To have the map $\mathcal{M}^{AB\rightarrow ABX}$ be
implemented solely on Alice's system, we can simply perform the isometry
$U^{A\rightarrow AWXI}$, discard the register $I$, and perform von Neumann
measurements of the registers $W$ and $X$. (One could also discard register
$X$, and then process $W$ with $p_{X|W}\left(  x|w\right)  $ to produce
$X$---it is possible to do this since $X$ is classical.) This amounts to an approximate
implementation of the following instrument:%
\[
\sum_{x,w}p_{X|W}\left(  x|w\right)  \left\vert x\right\rangle \left\langle
x\right\vert ^{X}\otimes\left\vert w\right\rangle \left\langle w\right\vert
^{W}\otimes\mathcal{M}_{w}^{AB}\left(  \phi^{RAB}\right)  ,
\]
from which we can discard register $W$ to obtain an approximation of the map $\mathcal{M}%
^{AB\rightarrow ABX}$. Thus, from the map $\mathcal{M}^{AB\rightarrow ABX}$,
it is possible to construct a nearby map of the form in
(\ref{eq:non-feed-MC-QSI-decomp}), so that it suffices to optimize over the
class of decompositions given in (\ref{eq:non-feed-MC-QSI-decomp}).

We now prove the second bound:%
\begin{align*}
n\left(  R+S\right)   &  \geq H\left(  LM\right)  _{\theta}\\
&  \geq H\left(  LM|B^{n}\right)  _{\theta}\\
&  \geq I(LM;\hat{X}^{n}R^{n}|B^{n})_{\theta}\\
&  =I(LMB^{n};\hat{X}^{n}R^{n})_{\theta}-I(B^{n};\hat{X}^{n}R^{n})_{\theta}\\
&  \geq I(LMSB^{n};\hat{X}^{n}R^{n})_{\omega}-I(B^{n};\hat{X}^{n}%
R^{n})_{\omega}-n\epsilon^{\prime}\\
&  =H(\hat{X}^{n}R^{n}|B^{n})_{\omega}-H(\hat{X}^{n}R^{n}|LMSB^{n})_{\omega
}-n\epsilon^{\prime}\\
&  \geq\sum_{k}\left[  H(\hat{X}_{k}R_{k}|B_{k})_{\omega}-H(\hat{X}_{k}%
R_{k}|LMSB_{k})_{\omega}\right]  -n2\epsilon^{\prime}\\
&  =\sum_{k}I(LMS;\hat{X}_{k}R_{k}|B_{k})_{\omega}-n2\epsilon^{\prime}\\
&  =nI(LMS;\hat{X}R|KB)_{\sigma}-n2\epsilon^{\prime}\\
&  \geq nI(LMS;\hat{X}R|KB)_{\sigma}+nI(K;\hat{X}R|B)_{\sigma}-n3\epsilon
^{\prime}\\
&  =nI(LMSK;\hat{X}R|B)_{\sigma}-n3\epsilon^{\prime}.
\end{align*}
The first three inequalities follow from similar reasons as our previous
inequalities. The first equality is an identity. The fourth inequality follows
from quantum data processing of $LMB^{n}$ to produce $LMSB^{n}$ and from the
fact that this does not change the state too much (we apply the condition in
(\ref{eq:good-non-feed-MC-QSI-protocol}) and the Alicki-Fannes' inequality).
The fifth inequality follows from strong subadditivity of entropy:%
\[
H(\hat{X}^{n}R^{n}|LMSB^{n})_{\omega}\leq\sum_{k}H(\hat{X}_{k}R_{k}%
|LMSB_{k})_{\omega},
\]
and from the fact that the measurement simulation is faithful so that%
\[
\left\vert H(\hat{X}^{n}R^{n}|B^{n})_{\omega}-\sum_{k}H(\hat{X}_{k}R_{k}%
|B_{k})_{\omega}\right\vert \leq n\epsilon^{\prime},
\]
where we have applied a variation of Lemma~\ref{lem:entropy-for-close-IID}.
The third equality is an identity. The fourth equality follows by considering
the state $\sigma$ as defined in (\ref{eq:helper-state-no-fdbk-MC-QSI}). The
sixth inequality follows from (\ref{eq:no-fdbk-KX-info-bnd-MC-QSI}). The final
equality is the chain rule for quantum mutual information. We can then
consider the same argument as stated before in order to construct a map acting
only on $A$ from one acting on $AB$. Similarly, the resulting state
has the form in (\ref{eq:non-feed-MC-QSI-decomp}).
\end{proof}

\section{Conclusion}

This paper provided a review of Winter's measurement compression theorem
\cite{W01}, detailing the information processing task, providing examples for
understanding it, reviewing Winter's achievability proof, and detailing a new
approach to its single-letter converse theorem. We proved a new theorem
characterizing the optimal rates for classical communication and common
randomness for a measurement compression protocol where the sender is not
required to obtain the outcome of the measurement simulation. We then reviewed
the Devetak-Winter theorem on classical data compression with quantum side
information, providing new proofs of the achievability and converse parts of
this theorem. From there, we presented a new protocol called measurement
compression with quantum side information (a protocol first announced in
Ref.~\cite{HW10}). This protocol has several applications, including its part
in the \textquotedblleft classically-assisted state
redistribution\textquotedblright\ protocol, which is the most general protocol
on the static side of the quantum information theory tree, and its role in
reducing the classical communication cost in local purity distillation
\cite{KD07}. We then outlined a connection between this protocol and recent
work in entropic uncertainty relations. Finally, we proved a single-letter
theorem for the task of measurement compression with quantum side information
when the sender is not required to obtain the outcome of the measurement simulation.

There are several open questions to consider going forward from here. First,
are there applications of the MC-QSI\ protocol to rate distortion, as was the
case for the Luo-Devetak protocol in Ref.~\cite{LD09}? Are there further
applications of the measurement compression protocol in general?\ Is it
possible to formulate a measurement compression protocol that is independent
of the state on which it acts (similar to the general reverse Shannon theorem
from Refs.~\cite{BDHSW09,BCR11})? The answers to these questions could further
illuminate our understanding of quantum measurement and address other
important areas of quantum information theory.

We thank Ke Li (Carl) for suggesting the possibility of the measurement
compression with quantum side information protocol to us. We are grateful to Cedric Beny,
Aram Harrow, Daniel Gottesman, Masanao Ozawa, and Andreas Winter for useful
discussions. MMW\ acknowledges support from the Centre de Recherches
Math\'{e}matiques at the University of Montreal. He also acknowledges both
Nagoya University and the Perimeter Institute for Theoretical Physics, where
some of this work was conducted. PH\ acknowledges support from the Canada
Research Chairs program, the Perimeter Institute, CIFAR, FQRNT's INTRIQ,
NSERC, and ONR through grant N000140811249. FB acknowledges support from the
Program for Improvement of Research Environment for Young Researchers from
Special Coordination Funds for Promoting Science and Technology (SCF)
commissioned by the Ministry of Education, Culture, Sports, Science and
Technology (MEXT) of Japan. MH received support from the Chancellor's
postdoctoral research fellowship, University of Technology Sydney (UTS), and
was also partly supported by the National Natural Science Foundation of China
(Grant Nos.~61179030) and the Australian Research Council (Grant No.~DP120103776).

\appendix

\section{Typical sequences and typical subspaces}

\label{sec:typ-review}A sequence $x^{n}$\ is typical with respect to some
probability distribution $p_{X}\left(  x\right)  $ if its empirical
distribution has maximum deviation $\delta$ from $p_{X}\left(  x\right)  $.
The typical set $T_{\delta}^{X^{n}}$ is the set of all such sequences:%
\[
T_{\delta}^{X^{n}}\equiv\left\{  x^{n}:\left\vert \frac{1}{n}N\left(
x|x^{n}\right)  -p_{X}\left(  x\right)  \right\vert \leq\delta\ \ \ \ \forall
x\in\mathcal{X}\right\}  ,
\]
where $N\left(  x|x^{n}\right)  $ counts the number of occurrences of the
letter $x$ in the sequence $x^{n}$. The above notion of typicality is the
\textquotedblleft strong\textquotedblright\ notion (as opposed to the weaker
\textquotedblleft entropic\textquotedblright\ version of typicality sometimes
employed~\cite{CT91}). The typical set enjoys three useful properties:\ its
probability approaches unity in the large $n$ limit, it has exponentially
smaller cardinality than the set of all sequences, and every sequence in the
typical set has approximately uniform probability. That is, suppose that
$X^{n}$ is a random variable distributed according to $p_{X^{n}}\left(
x^{n}\right)  \equiv p_{X}\left(  x_{1}\right)  \cdots p_{X}\left(
x_{n}\right)  $, $\epsilon$ is positive number that becomes arbitrarily small
as $n$ becomes large, and $c$ is some positive constant. Then the following
three properties hold \cite{CT91}%
\begin{align}
\Pr\left\{  X^{n}\in T_{\delta}^{X^{n}}\right\}   &  \geq1-\epsilon
,\label{eq:typ-1}\\
\left\vert T_{\delta}^{X^{n}}\right\vert  &  \leq2^{n\left[  H\left(
X\right)  +c\delta\right]  },\label{eq:typ-2}\\
\forall x^{n}\in T_{\delta}^{X^{n}}:\ \ \ 2^{-n\left[  H\left(  X\right)
+c\delta\right]  }  &  \leq p_{X^{n}}\left(  x^{n}\right)  \leq2^{-n\left[
H\left(  X\right)  -c\delta\right]  }. \label{eq:typ-3}%
\end{align}
We omit using $c$ in the main text and instead subsume it as part of $\delta$.

These properties translate straightforwardly to the quantum setting by
applying the spectral theorem to a density operator $\rho$. That is, suppose
that%
\[
\rho\equiv\sum_{x}p_{X}\left(  x\right)  \left\vert x\right\rangle
\left\langle x\right\vert ,
\]
for some orthonormal basis $\left\{  \left\vert x\right\rangle \right\}  _{x}%
$. Then there is a typical subspace defined as follows:%
\[
T_{\rho,\delta}^{n}\equiv\text{span}\left\{  \left\vert x^{n}\right\rangle
:\left\vert \frac{1}{n}N\left(  x|x^{n}\right)  -p_{X}\left(  x\right)
\right\vert \leq\delta\ \ \ \ \forall x\in\mathcal{X}\right\}  ,
\]
and let $\Pi_{\rho,\delta}^{n}$ denote the projector onto it. Then properties
analogous to (\ref{eq:typ-1}-\ref{eq:typ-3}) hold for the typical subspace.
The probability that a tensor power state $\rho^{\otimes n}$\ is in the
typical subspace approaches unity as $n$ becomes large, the rank of the
typical projector is exponentially smaller than the rank of the full $n$-fold
tensor-product Hilbert space of $\rho^{\otimes n}$, and the state
$\rho^{\otimes n}$ \textquotedblleft looks\textquotedblright\ approximately
maximally mixed on the typical subspace:%
\begin{align}
\text{Tr}\left\{  \Pi_{\rho,\delta}^{n}\ \rho^{\otimes n}\right\}   &
\geq1-\epsilon,\label{eq:typ-q-1}\\
\text{Tr}\left\{  \Pi_{\rho,\delta}^{n}\right\}   &  \leq2^{n\left[  H\left(
B\right)  +c\delta\right]  },\label{eq:typ-q-2}\\
2^{-n\left[  H\left(  B\right)  +c\delta\right]  }\ \Pi_{\rho,\delta}^{n}  &
\leq\Pi_{\rho,\delta}^{n}\ \rho^{\otimes n}\ \Pi_{\rho,\delta}^{n}%
\leq2^{-n\left[  H\left(  B\right)  -c\delta\right]  }\ \Pi_{\rho,\delta}^{n},
\label{eq:typ-q-3}%
\end{align}
where $H\left(  B\right)  $ is the entropy of $\rho$.

Suppose now that we have an ensemble of the form $\left\{  p_{X}\left(
x\right)  ,\rho_{x}\right\}  $, and suppose that we generate a typical
sequence $x^{n}$ according to the pruned distribution in (\ref{eq:pruned-dist}%
), leading to a tensor product state $\rho_{x^{n}}\equiv\rho_{x_{1}}%
\otimes\cdots\otimes\rho_{x_{n}}$. Then there is a conditionally typical
subspace with a conditionally typical projector defined as follows:%
\[
\Pi_{\rho_{x^{n}},\delta}^{n}\equiv\bigotimes\limits_{x\in\mathcal{X}}%
\Pi_{\rho_{x},\delta}^{I_{x}},
\]
where $I_{x}\equiv\left\{  i:x_{i}=x\right\}  $ is an indicator set that
selects the indices $i$\ in the sequence $x^{n}$ for which the $i^{\text{th}}$
symbol $x_{i}$\ is equal to $x\in\mathcal{X}$ and $\Pi_{\rho_{x},\delta
}^{I_{x}}$ is the typical projector for the state $\rho_{x}$. The
conditionally typical subspace has the three following properties:%
\begin{align}
\text{Tr}\left\{  \Pi_{\rho_{x^{n}},\delta}^{n}\ \rho_{x^{n}}\right\}   &
\geq1-\epsilon,\\
\text{Tr}\left\{  \Pi_{\rho_{x^{n}},\delta}^{n}\right\}   &  \leq2^{n\left[
H\left(  B|X\right)  +c\delta\right]  },\\
2^{-n\left[  H\left(  B|X\right)  +c\delta\right]  }\ \Pi_{\rho_{x^{n}}%
,\delta}^{n}  &  \leq\Pi_{\rho_{x^{n}},\delta}^{n}\ \rho_{x^{n}}\ \Pi
_{\rho_{x^{n}},\delta}^{n}\leq2^{-n\left[  H\left(  B|X\right)  -c\delta
\right]  }\ \Pi_{\rho_{x^{n}},\delta}^{n},
\end{align}
where $H\left(  B|X\right)  =\sum_{x}p_{X}\left(  x\right)  H\left(  \rho
_{x}\right)  $ is the conditional quantum entropy.

Let $\rho$ be the expected density operator of the ensemble $\left\{
p_{X}\left(  x\right)  ,\rho_{x}\right\}  $ so that $\rho=\sum_{x}p_{X}\left(
x\right)  \rho_{x}$. The following properties are proved in
Refs.~\cite{ieee2005dev,itit1999winter,W11}:%
\begin{align}
\forall x^{n} \in T^{X^{n}}_{\delta}: \text{Tr}\left\{  \rho_{x^{n}}%
\ \Pi_{\rho}\right\}   &  \geq1-\epsilon,\nonumber\\
\sum_{x^{n}}p_{X^{\prime n}}^{\prime}\left(  x\right)  \rho_{x^{n}}  &
\leq\left[  1-\epsilon\right]  ^{-1}\rho^{\otimes n}.
\label{eq:prune-avg-op-ineq}%
\end{align}

\section{Useful lemmas}

\label{sec:useful-lemmas}Here we collect some useful lemmas.

\begin{lemma}
[Gentle Operator Lemma \cite{itit1999winter,ON07}]\label{lem:gentle-operator}%
Let $\Lambda$ be a positive operator where $0\leq\Lambda\leq I$ (usually
$\Lambda$ is a POVM\ element), $\rho$ a state, and $\epsilon$ a positive
number such that the probability of detecting the outcome $\Lambda$ is high:%
\[
\emph{Tr}\left\{  \Lambda\rho\right\}  \geq1-\epsilon.
\]
Then the measurement causes little disturbance to the state $\rho$:%
\[
\left\Vert \rho-\sqrt{\Lambda}\rho\sqrt{\Lambda}\right\Vert _{1}\leq
2\sqrt{\epsilon}.
\]

\end{lemma}

\begin{lemma}
[Gentle Operator Lemma for Ensembles \cite{itit1999winter,ON07,W11}%
]\label{lem:gentle-operator-ens} Given an ensemble $\left\{  p_{X}\left(
x\right)  ,\rho_{x}\right\}  $ with expected density operator $\rho\equiv
\sum_{x}p_{X}\left(  x\right)  \rho_{x}$, suppose that an operator $\Lambda$
such that $I\geq\Lambda\geq0$ succeeds with high probability on the state
$\rho$:%
\[
\emph{Tr}\left\{  \Lambda\rho\right\}  \geq1-\epsilon.
\]
Then the subnormalized state $\sqrt{\Lambda}\rho_{x}\sqrt{\Lambda}$ is close
in expected trace distance to the original state $\rho_{x}$:%
\[
\mathbb{E}_{X}\left\{  \left\Vert \sqrt{\Lambda}\rho_{X}\sqrt{\Lambda}%
-\rho_{X}\right\Vert _{1}\right\}  \leq2\sqrt{\epsilon}.
\]

\end{lemma}

\begin{lemma}
\label{lem:trace-inequality}Let $\rho$ and $\sigma$ be positive operators and
$\Lambda$ a positive operator such that $0\leq\Lambda\leq I$. Then the
following inequality holds%
\[
\emph{Tr}\left\{  \Lambda\rho\right\}  \leq\emph{Tr}\left\{  \Lambda
\sigma\right\}  +\left\Vert \rho-\sigma\right\Vert _{1}.
\]

\end{lemma}

\begin{lemma}
[Non-commutative union bound \cite{S11}]\label{lem-non-com-union-bound}Let
$\sigma$ be a subnormalized state such that $\sigma\geq0$ and $\emph{Tr}%
\left\{  \sigma\right\}  \leq1$. Let $\Pi_{1}$, \ldots, $\Pi_{N}$ be
projectors. Then the following \textquotedblleft non-commutative union
bound\textquotedblright\ holds%
\[
\emph{Tr}\left\{  \sigma\right\}  -\emph{Tr}\left\{  \Pi_{N}\cdots\Pi
_{1}\sigma\Pi_{1}\cdots\Pi_{N}\right\}  \leq2\sqrt{\sum_{i=1}^{N}%
\emph{Tr}\left\{  \left(  I-\Pi_{i}\right)  \sigma\right\}  }.
\]

\end{lemma}

\bibliographystyle{plain}
\bibliography{Ref}

\end{document}